\RequirePackage[l2tabu, orthodox]{nag}		
\documentclass[10pt]{article}

\usepackage[T1]{fontenc}
\usepackage[frenchmath]{mathastext}							

\usepackage{amsmath}                        
\numberwithin{equation}{section}
\usepackage{graphicx} 											
\usepackage{adjustbox,float}
\usepackage{enumitem}												
\usepackage{mdwlist}												
\usepackage[dvipsnames]{xcolor}							
\usepackage[plainpages=false, pdfpagelabels]{hyperref} 
\hypersetup{
colorlinks   = true,
citecolor    = RoyalBlue, 
linkcolor    = RubineRed, 
urlcolor     = Turquoise
}

\usepackage{subcaption}

\usepackage{amssymb}                        
\usepackage[mathscr]{eucal}                 
\usepackage{dsfont}													
\usepackage[paperwidth=8.5in,paperheight=11in,top=1.in, bottom=1.in, left=0.9in, right=0.9in]{geometry} 
\usepackage{mathtools}                      
\mathtoolsset{showonlyrefs=true}          
\usepackage{amsthm}                         
\allowdisplaybreaks                       
\theoremstyle{plain}
\newtheorem{theorem}{Theorem}
\numberwithin{theorem}{section}
\newtheorem{lemma}{Lemma}       	
\numberwithin{lemma}{section}

\numberwithin{proposition}{section}
\newtheorem{corollary}{Corollary}
\numberwithin{corollary}{section}
\theoremstyle{definition}
\newtheorem{definition}{Definition}
\numberwithin{definition}{section}

\numberwithin{example}{section}
\newtheorem{remark}{Remark}
\numberwithin{remark}{section}
\newtheorem{assumption}{Assumption}
\numberwithin{assumption}{section}

\newcommand{\cref}{\eqref}

\newcommand\Eb{\mathds{E}}
\newcommand\Fb{\mathds{F}}

\newcommand\Pb{\mathds{P}}

\newcommand\Rb{\mathds{R}}

\newcommand\Ac{\mathscr{A}}

\newcommand\Cc{\mathcal{C}}
\newcommand\Dc{\mathcal{D}}

\newcommand\Fc{\mathscr{F}}

\newcommand\Nc{\mathscr{N}}

\newcommand\Uc{\mathcal{U}}

\newcommand\eps{\epsilon}

\newcommand\Om{\Omega}
\newcommand\sig{\sigma}

\newcommand\gam{\gamma}

\newcommand\lam{\lambda}
\newcommand\del{\delta}

\newcommand\kap{\kappa}

\newcommand\yb{\bar{y}}

\newcommand\xu{\underline{x}}

\newcommand\zu{\underline{z}}

\newcommand\phih{\widehat{\varphi}}
\newcommand\psih{\widehat{\psi}}

\newcommand \ch {\hat{c}}

\newcommand\At{\widetilde{A}}

\newcommand\Bt{\widetilde{B}}

\newcommand\at{\tilde{a}}
\newcommand\bt{\tilde{b}}

\newcommand\yt{\widetilde{y}}
\newcommand\zt{\widetilde{z}}

\newcommand\Ut{\widetilde{U}}

\newcommand\Vt{\widetilde{V}}
\newcommand\Wt{\widetilde{W}}

\newcommand\phit{\widetilde{\varphi}}
\newcommand\psit{\widetilde{\psi}}

\newcommand\epst{\tilde{\epsilon}}

\newcommand\dd{\mathrm{d}}
\newcommand\ee{\mathrm{e}}

\newcommand\pis{{\pi^*}}

\newcommand{\cs}{{c^{*}}}

\newcommand{\Acr}{{\Ac_\text{\rm rel.}}}

\newcommand{\R}{\mathbb{R}}
\newcommand{\E}{\mathbb{E}\thinspace }

\newcommand \al {\alpha}

\usepackage[
backend=bibtex,
style=ieee-alphabetic, 
sorting=nyt,
giveninits=true,
]{biblatex}
\AtBeginBibliography{\footnotesize} 

\begin{document}

\title{Optimal consumption under loss-averse multiplicative habit-formation}

\author{
	Bahman Angoshtari\thanks{Department of Mathematics, University of Miami, Coral Gables, USA. 
		\url{bangoshtari@miami.edu}}
	\and Xiang Yu\thanks{Department of Applied Mathematics, Hong Kong Polytechnic University, Kowloon, Hong Kong. 
		\url{xiang.yu@polyu.edu.hk}}
	\and Fengyi Yuan\thanks{School of Science and Engineering, Chinese University of Hong Kong (Shenzhen), Shenzhen, China.  
		\url{yuanfengyi@cuhk.edu.cn}}
}
\date{\vspace{-0.2in}}

\maketitle

\begin{abstract}
This paper studies a loss-averse version of the multiplicative habit formation preference and the corresponding optimal investment and consumption strategies over an infinite horizon. The agent's consumption preference is depicted by a general S-shaped utility function of her consumption-to-habit ratio. By considering the concave envelope of the S-shaped utility and the associated dual value function, we provide a thorough analysis of the HJB equation for the concavified problem via studying a related nonlinear free boundary problem. Based on established properties of the solution to this free boundary problem, we obtain the optimal consumption and investment policies in feedback form. Some new and technical verification arguments are developed to cope with generality of the utility function. The equivalence between the original problem and the concavified problem readily follows from the structure of the feedback controls. We also discuss some quantitative properties of the optimal policies, complemented by illustrative numerical examples and their financial implications. 	\vspace{1em}
	
\end{abstract}



%
%

\section{Introduction}\label{sec:intro}
In the past decades, time non-separable preferences have been popularized to explain some empirically observed phenomena such as excessive consumption smoothing (\cite{Camp}, \cite{Sun89}) and the equity premium puzzle (\cite{Mehra}, \cite{Cons}).  Among them, habit formation preference $U(C_t, H_t)$ has been widely studied to steer consumption planning, where the agent's satisfaction and risk aversion depend on the relative change of consumption rate $C_t$ with respect to the habit level $H_t$, rather than on the absolute consumption rate as in  (\cite{Merton69}, \cite{Merton71}). Here, the consumption habit is typically modeled by the exponentially weighted average of the past consumption that
\begin{align}\label{eq:HABIT}
	H_t=he^{-\rho t} + \rho\int^t_0 \ee^{-\rho(t-u)} C_u \dd u,\quad t\ge0,
\end{align}
in which $h>0$ is the initial habit and $\rho>0$ is the habit persistence parameter. 

Two types of habit formation preferences can be found in the literature, namely the \emph{linear} (or \emph{additive}) habit formation (\cite{Cons}) and the \emph{multiplicative} habit formation (\cite{Abel}). The linear habit formation $U(C_t-H_t)$ measures the utility based on the difference between current consumption rate and the habit. Since it is commonly assumed that $U'(0^+)=+\infty$, an \textit{addictive} consumption constraint $C_t\geq H_t$ is imposed on admissible consumption policies. To sustain the control constraint without bankruptcy, the initial wealth also needs to fulfill a threshold constraint relative to the initial habit $W_0\geq ah$ for some constant $a>0$. Along this direction, extensive studies on the characterization of the optimal consumption and some financial insights can be found in \cite{DZ1992}, \cite{Eng}, \cite{Munk}, \cite{SS2002}, \cite{Yu2015}, \cite{Yu2017}, \cite{YY}, \cite{BWY}, \cite{GHLY}, to name a few. 

The multiplicative habit formation $U(C_t/H_t)$, on the other hand, imposes no subsistence constraint on the agent's initial wealth. Indeed, the consumption control $C_t$ under the multiplicative habit formation is allowed to fall below the habit formation process $H_t$ from time to time. 
This \textit{non-addictive} nature of multiplicative habit formation preferences has been favored by some recent studies (see \cite{Car2000}, \cite{CarrollEtal1997}, \cite{CarrollEtal2000}, \cite{Fuh}, \cite{corr}, \cite{VanBilsen2020}, \cite{LWY}, \cite{Kam22}) thanks to its flexibility in fitting diverse market environments. However, in contrast to the linear habit formation, it is generally more challenging to derive analytical characterization of the optimal consumption policies under multiplicative habit formation. In \cite{VanBilsen2020}, the authors replaced the common habit process \cref{eq:HABIT} with an alternative geometric specification $d\log(H_t)=\rho(\log C_t -\log H_t)$ and then considered the optimal lifetime consumption problem with a power-type multiplicative habit formation utility. Using the geometric specification of the habit process, they derived an approximate optimal consumption control in closed-form. To obtain the approximate control, however, they rely on a linear approximation of their budget constraint (see Section III.C therein for details) that relies on the optimal consumption-to-habit ratio to be around 1. In other words, the validity of their results rely on $C_t/H_t\approx 1$. Using the same geometric form of the habit formation process as in  \cite{VanBilsen2020}, \cite{Kam22} further develop a duality result under power-type multiplicative habit formation utility. 

The above studies rely on the Maximum Principle to analyze the underlying stochastic control problems. The following studies use the Dynamic Programming Principle. \cite{ABY22} and \cite{ABY23} considered the common habit process \cref{eq:HABIT} and study a new variation of multiplicative habit formation preference $U(C_t/H_t)=\frac{1}{1-\gam}(C_t/H_t)^{1-\gam}$, $\gam>1$, by mandating the additional habit constraint $C_t\ge \alpha H_t$ with a constant parameter $\alpha\in[0,1]$ into the set of admissible portfolio-consumption strategies. They derived the feedback form of the optimal investment and consumption policies by analyzing the corresponding Hamilton–Jacobi–Bellman (HJB) equation. Due to the control constraint, however, the initial wealth therein also needs to stay above a proportion of the initial habit level (similar to previous studies on linear habit formation). Most recently, \cite{ABG2025} considered a new variation the common habit process \cref{eq:HABIT} with added noise, i.e. $\dd H_t = \rho(C_t - H_t)\dd t + \beta H_t \dd \Wt_t$, in which $\Wt$ is a Brownian motion independent of Brownian motion in their budget constraint. They studied an infinite horizon optimal consumption problem with multiplicative habit formation utility $U(C_t/H_t)=\frac{1}{1-\gam}(C_t/H_t)^{1-\gam}$, $\gam>1$, and established a full verification result for the underlying stochastic control problem, i.e. that the value function is the classic solution of the HJB equation. They also provided several properties of their value function, while leaving the characterization of the optimal policies for future research.

Although fruitful studies can be found in both types of habit formation preferences, neither can address loss aversion tendencies of individuals, evidently supported by empirical observations (see \cite{Kos09} and \cite{KKP}) that the agent might suffer more from a reduction in the relative consumption than would benefit from the same size of increment. Indeed, linear habit formation models intrinsically rule out the possibility of consumption loss relative to habit level. Albeit the multiplicative habit formation allows the agent to strategically budget the consumption plan below the habit level, the agent exhibits the same risk aversion attitude on gains and losses, failing to reflect her loss aversion in its conventional formulation.

As suggested by \cite{Kahneman79} and \cite{TverskyKahneman1992}, it is more suitable to employ S-shaped two-part utility functions such that the agent's risk aversion attitudes differ when the consumption rate is above or below a reference level. In the literature of behavioral finance, the loss-averse two-part utilities with a reference point has been predominantly studied for terminal wealth optimization, see \cite{JZhou},  \cite{HZhou}, \cite{HYang},  \cite{DongZheng}, \cite{BKP}, among others. Only a handful of studies can be found that focuses on the agent's loss aversion on relative consumption towards the accustomed habit. In \cite{Cur17}, the S-shaped power utilities $U(C_t-H_t)$ are considered similar to the linear habit formation preferences where, instead of the common dynamics \cref{eq:HABIT}, their habit formation process is simplified to the form of $H_t=h+(1-\alpha)\nu ht+ \alpha C_t$ with $0\leq \alpha<1$, $\nu\leq 0$ and the initial habit $h\geq 0$ such that the past consumption influence actually disappears. Later, to accommodate the conventional path-dependent habit formation process $H_t$ in the problem $U(C_t-H_t)$ under S-shaped power utilities, \cite{VanBilsenvLaevenNijman2020} develop the martingale duality method by taking advantage of the linear structure in terms of the consumption and imposing an artificial lower bound on the difference $C_t-H_t\ge -M$ for a fixed constant $M$. Recently, \cite{LYZ} investigate the loss-averse two-part power utilities on relative consumption $U(C_t-H_t)$ based on PDE analysis where the reference process $H_t$ is chosen as a proportion of the historical running maximum of the consumption process.

The present paper aims to pioneer the study of optimal consumption by featuring both the multiplicative habit formation and the loss-averse two-part utilities. We consider the loss-averse multiplicative habit formation $U(C_t/H_t)$ on the consumption-to-habit ratio with a reference level $\alpha>0$, namely,
\begin{align}
	U(c) = 
	\begin{cases}
		U_+(c-\al),\quad c>\al,\\
		-U_-(\al-c),\quad 0\le c\le \al,
	\end{cases}
\end{align}
in which $U_\pm(\cdot)$ are two general utility functions with possibly distinct risk aversion (see \cref{eq:Uc}). Contrary to \cite{VanBilsenvLaevenNijman2020} that imposes the lower bound $C_t-H_t\geq -M$ for some $M>0$ in their loss-averse linear habit formation, our formulation and methodology allow us to consider any nonnegative consumption and its aggregated habit formation process without model restrictions. Furthermore, instead of imposing a strict habit formation constraint $C_t/H_t\ge \al$ considered in \cite{ABY22}, we take a more relaxed formation by assigning a loss averse utility whenever the consumption-to-habit ratio falls below the reference level $\alpha$. Thus, the constraint on the initial wealth in \cite{ABY22} is no longer needed in our formulation, making our model applicable to agents starting at any financial situation. Finally, in contrast to most existing literature on consumption habit formation, we also allow for the case $\alpha>1$, which reflects individuals with a high consumption–aspiration level.

It is well-known that the optimization problem is no longer concave under the above S-shaped utility. We choose to employ the concave envelope (of the utility function) and first study the HJB equation associated to the concavified stochastic control problem. To circumvent the challenge caused by the path-dependency due to habit formation, we further consider the auxiliary controls and state process, namely, the relative consumption, the relative investment, and the relative wealth, all respect to the habit formation process (see \cref{eq:ratios} and \cref{eq:X}). Consequently, it is sufficient to investigate an equivalent one-dimensional HJB equation under the concavified utility. Thanks to the structure of the concave envelope, it is natural to conjecture the existence of a free boundary threshold $x_0\geq 0$ for the relative wealth such that the optimal relative consumption is completely suspended, i.e., $c^*_t=0$, whenever the relative wealth falls below $x_0$; and the optimal relative consumption is characterized by the first order condition when the relative wealth diffuses above or at $x_0$. We eventually encounter a nonlinear free boundary problem (see problem \cref{eq:HJB-FBP1-moderate}), and our mathematical task is to show the existence of a classical solution to this free boundary problem and establish the rigorous verification proof for the conjectured optimal controls. An added level of difficulty stems from the generality of our utility function. As noted earlier, all existing studies on multiplicative habit formation assume a restricted class of power utility functions, namely, $U(C_t/H_t)=\frac{1}{1-\gam}(C_t/H_t)^{1-\gam}$ for $\gam>1$.

The contribution of our paper is three-fold:\vspace{1ex}

\noindent\textbf{(i)} Our work is the first study on loss-averse multiplicative habit formation that yields feedback optimal strategies. Our characterization of feedback controls leads to a straightforward numerical scheme (namely, a bisection search for the free-boundary $x_0$ coupled with an ODE solver) with theoretical guarantees, which allows us to discuss some interesting financial implications. Moreover, our model encompasses many existing models as limiting cases, and also allows generic utility functions; see Section \ref{sec:Examples} for detailed discussions.\vspace{1ex}

\noindent \textbf{(ii)} We develop a methodology to analyze a type of nonlinear free boundary problem under the general S-shaped utility. In a nutshell, we choose to transform the targeted problem into several auxiliary problems, for which we are able to obtain some additional conditions and properties to assist our analysis. First, we adopt the dual transform of the original problem, leading to another nonlinear free boundary problem (see \cref{eq:modFBPDual1}--\cref{eq:modFBPDual4}). However, the dual PDE problem when the dual variable is larger than the free boundary (see \cref{eq:modFBPDual1}) is linear, which provides an additional explicit free boundary condition (see \cref{eq:u1-moderate}) on strength of the smooth-fit principle. Next, to study the dual nonlinear free boundary problem (see \cref{eq:modFBPDual-reduced}) with two free boundary conditions on the function and its derivative respectively, we propose to investigate an auxiliary system of first-order free boundary ODEs \cref{eq:ODESys-moderate}-\cref{eq:FB-moderate}, whose solutions are related to the dual free boundary problem via the transformations \cref{up_phi_main} and \cref{upp_phi_psi_main} in our first main result Theorem \ref{thm:DualFBP-moderate}. For this coupled system of free boundary problems, we develop some delicate arguments, new to the literature, to address the existence of the unique classical solution and to derive some important asymptotic conditions of the solutions at the boundary $0$.\vspace{1ex}

\noindent\textbf{(iii)} Another theoretical contribution of the present paper is our verification arguments for general utility functions that satisfy mild growth conditions. Thanks to two asymptotic conditions of the solutions to the coupled system, we are able to first prove the key step of transversality condition (see Lemma \ref{lem:trans}) in two separate cases via different estimations and techniques. Despite the lack of an explicit structure due to the generality of the utility function and its concave envelope, we are still able to show the existence of a unique strong solution to the state SDE under the optimal feedback controls, and further verify that the solution of the HJB free boundary problem indeed coincide with the value functions of the concavified problem and the original problem.\vspace{1ex}

The remainder of the paper is organized as follows. Section \ref{sec:model} introduces the market model and the problem formulation under the loss-averse multiplicative habit formation. Section \ref{sec:OptimalPolicies} studies the HJB equation for the concavified problem and transforms it into another auxiliary free boundary problem, for which the existence of a unique classical solution is established. In addition, after obtaining some key boundary conditions of the solution, the verification theorem of the optimal feedback controls is established under general S-shaped utility functions. Section \ref{sec:Examples} presents some numerical illustrations and sensitivity analysis with respect to model parameters. Numerous financial implications, particularly the impact of habit formation and loss aversion levels, are also discussed and illustrated therein. Finally, Section \ref{sec:techproofs} collects lengthy and technical proofs of the main results in the earlier sections.  

%
%
\section{Model Setup}\label{sec:model}

We consider a financial market model consisting of a riskless asset with short rate $r\ge0$ and a risky asset whose price process $\{S_t\}_{t\ge0}$ is governed by
\begin{align}\label{eq:S-SDE}
	\dd S_t = (r+\mu) S_t \dd t + \sig S_t\dd B_t, \quad t\ge0,
\end{align}
where $\mu,\ \sig>0$ are the expected excess rate of return and volatility of the risky asset, and $B=\{B_t\}_{t\ge0}$ is a standard Brownian motion supported on the filtered probability space $(\Om, \Fc, \Fb=\{\Fc_t\}_{t\ge0}, \Pb)$ with $\Fb$ being the augmented filtration of $B$. 

An individual (henceforth, the agent) funds her lifetime consumption by investing in this market. Let $\Pi_t$ stand for the amount of wealth invested in the risky asset and $C_t$ be the consumption rate at time $t$. The agent's self-financing wealth process $\{W_t\}_{t\ge 0}$ then satisfies
\begin{align}\label{eq:W}
	\dd W_t = (r W_t + \mu \Pi_t - C_t) \dd t + \sig \Pi_t \dd B_t,\quad t\ge0,
\end{align}
with the initial wealth $W_0=w>0$.

\begin{definition}
	Given $w>0$, a progressively measurable process $(\Pi,C)=\{(\Pi_t,C_t)\}_{t\ge0}$ is an admissible investment and consumption control pair if $\int_0^t\big(|C_u|+\Pi_u^2\big)\dd u<+\infty$ and $W_t\ge0$ for all $t\ge0$, in which $\{W_t\}_{t\ge0}$ is the strong solution of \cref{eq:W}. $\Ac_0(w)$ denotes the set of all admissible controls starting with the initial wealth $w$.\qed
\end{definition}

In our model, it is assumed that the agent gradually develops consumption habits such that her performance and risk aversion on consumption depend on the relative changes with respect to her accustomed habit level. In particular, her habit formation process $\big\{H_t\big\}_{t\ge0}$, also called the standard of living process, is conventionally defined as the exponentially-weighted average of the past consumption up to date, namely
\begin{align}\label{eq:H}
	H_t := h + \rho\int_0^t \left(C_u - H_u\right)\dd u = he^{-\rho t} + \rho\int^t_0 \ee^{-\rho(t-u)} C_u \dd u,\quad t\ge0,
\end{align}
in which $h>0$ is the initial habit level and $\rho>0$ stands for the persistence rate of the past consumption levels. A larger value of $\rho$ indicates that the agent's habit is more inclined to recent consumption behavior compared to past consumption rates in distant time.

In contrast to previous studies in the literature, we aim to study the \emph{loss averse version of the multiplicative habit formation preference} to account for the agent's psychological difference towards the same sized gain and loss of relative consumption with reference to the habit formation process. To this end, we adapt a general S-shaped two-part utility (see \cite{Kahneman79} and \cite{TverskyKahneman1992}) to encode the agent's loss aversion in the multiplicative habit formation, namely,
\begin{align}\label{eq:U}
	U(c) = 
	\begin{cases}
		U_+(c-\al),\quad c>\al,\\
		-U_-(\al-c),\quad 0\le c\le \al.
	\end{cases}
\end{align}
Here, $0<\al$ is the reference point for the agent's consumption-to-habit ratio, indicating the scenario of the relative consumption $0\le C_t/H_t< \al$ as loss, and the scenario $C_t/H_t>\al$ as gain. Distinct risk aversion stemming from utilities $U_\pm(\cdot)$ then effectively depict the agent's different feelings over same-sized gains and losses in the relative consumption rates. That is, $c\mapsto -U_-(\al-c)$ is the (convex) loss-averse utility function in the loss region $0\le c\le \al$, and $c\mapsto U_+(c-\al)$ is the (concave) risk-averse utility function in the gain region $c>\al$. In \cref{eq:U}, it is assumed that $U_\pm(\cdot)$ are concave utility functions satisfying
\begin{align}\label{eq:Uc}
	U_\pm \in\Uc:=\Big\{f\in\Cc^2(\Rb_+): f'>0, f^{\prime\prime}<0, f'(+\infty)=0, f(0)=0\Big\}.
\end{align}
Note that $U(\al)=U_+(0)=-U_-(0)=0$ since $U_\pm(0)=0$ by \cref{eq:Uc}. This assumption is taken without loss of generality, as we can always shift a utility function by a constant.

In contrast to most of the existing literature on consumption habit formation, we also include the case $\alpha>1$, in which the consumption reference level $\alpha H_t$ exceeds the habit process $H_t$. This regime represents an agent with a high consumption–aspiration level.

Later, in Subsection \ref{subsec:candidatepolicies}, we will impose two additional assumptions on the utility function $U$, namely, Assumptions \ref{assumption:moderateLA} and \ref{assum:Growth} below. These assumptions are technical in nature, but, they are common in the literature. See Remarks \ref{rem:moderateLA} and \ref{rem:Growth} for further discussion.

To find her optimal investment and consumption policies, the agent needs to solve the following stochastic control problem
\begin{align}\label{eq:V1}
	V_0(w,h) = \sup_{(\Pi,C)\in\Ac_0(w)} \Eb\left[\int_0^{+\infty}\ee^{-\del t} U\left(\frac{C_t}{H_t}\right)\dd t\right],\quad w,h>0,
\end{align}
in which $\del>0$ is her discount rate for utility of consumption. This stochastic control problem has two state processes, namely $\{W_t\}_{t\ge0}$ and $\{H_t\}_{t\ge0}$. Next, we reduce the number of state processes to one by considering an equivalent form of \cref{eq:V1}.

Fix the values of $w,h>0$ and consider an arbitrary admissible policy $(\Pi,C)\in\Ac_0(w)$. Let $\{W_t\}_{t\ge0}$ and $\{H_t\}_{t\ge0}$ be the corresponding wealth and habit processes given by \cref{eq:W} and \cref{eq:H}, respectively. Note that, by \cref{eq:H}, $H_t\ge h\ee^{-\rho t}>0$. Thus, we can define the relative wealth process $\{X_t\}_{t\ge0}$, the relative investment process $\{\pi_t\}_{t\ge0}$, and the relative consumption process $\{c_t\}_{t\ge0}$, respectively, by
\begin{align}\label{eq:ratios}
	X_t:= \frac{W_t}{H_t}, \quad \pi_t:=\frac{\Pi_t}{H_t}, \quad c_t:= \frac{C_t}{H_t},  \quad t \geq 0.
\end{align}

Combining It\^o's formula, \cref{eq:W}, and \cref{eq:H}, we can write the dynamics of $X$ as
\begin{align}\label{eq:X}
	\dd X_t = \Big((r+\rho)X_t + \mu\pi_t - (1+\rho X_t)c_t\Big) \dd t + \sig \pi_t\dd B_t,\quad t\ge0,
\end{align}
with $X_0=x:=w/h>0$. With these observations in mind, we define the set of admissible relative consumption and investment controls as follows.

\begin{definition}\label{def:rel-admis}
	Given $x>0$, a progressively measurable process $(\pi,c)=\{(\pi_t,c_t)\}_{t\ge0}$ is an admissible relative investment and consumption policy if $\int_0^t\big(|c_u|+\pi_u^2\big)\dd u<+\infty$ and $X_t\ge0$ for all $t\ge0$, in which $\{X_t\}_{t\ge0}$ is the strong solution of \cref{eq:X}.
	We denote by $\Acr(x)$ the set of all admissible relative investment and consumption policies starting with relative wealth $x$.\qed
\end{definition}

As the next result shows, for any $h,w>0$, the sets $\Ac_0(w)$ and $\Acr(w/h)$ are equivalent in the sense that each admissible investment and consumption policy $(\Pi,C)\in\Ac_0(w)$ corresponds to an admissible relative investment and consumption policy $(\pi,c)\in\Acr(w/h)$, and vice versa. We omit its proof, which is a straightforward application of It\^o's formula.

\begin{lemma}\label{lem:WX}
	Assume that $w,h>0$. For any $(\Pi,C)\in\Ac_0(w)$, we have $(\Pi/H, C/H)\in \Acr(w/h)$. Conversely, if $(\pi,c)\in\Acr(h/w)$, then $(\pi W/X, c W/X)\in\Ac_0(w)$ in which the relative wealth $\{X_t\}_{t\ge0}$ is given by \cref{eq:X} and the wealth process $\{W_t\}_{t\ge0}$ is given by
	\begin{align*}
		\frac{\dd W_t}{W_t} = \left(r+\mu\frac{\pi_t}{X_t} - \frac{c_t}{X_t}\right)\dd t
		+\sig\frac{\pi_t}{X_t}\dd B_t,\quad t\ge0,
	\end{align*}
	with $W_0=w$.\qed
\end{lemma}

From Lemma \ref{lem:WX}, it follows that \cref{eq:V1} is equivalent to the following one-dimensional stochastic control problem
\begin{align}\label{eq:VF}
	V(x) = \sup_{(\pi,c)\in\Acr(x)} \Eb\left[\int_0^{+\infty}\ee^{-\del t} U\left(c_t\right)\dd t\right],\quad x>0,
\end{align}
in which the only state variable is the wealth-to-habit ratio $\{X_t\}_{t\ge0}$. Once we solved \cref{eq:VF}, we can use Lemma \ref{lem:WX} to obtain the optimal investment and consumption polices and the corresponding optimal wealth.

\begin{remark}\label{rem:V0_VF}
	In light of Lemma \ref{lem:WX} and by \cref{eq:V1} and \cref{eq:VF}, it follows that $V_0(w,h)=V(w/h)$ for $w,h>0$. That is, our dimension reduction notably does \emph{not} rely on a specific form of the utility functions $U_\pm(c)$. In particular, \cref{eq:V1} and \cref{eq:VF} are equivalent even if the utility functions $U_\pm(c)$ in \cref{eq:U} are not standard power, log, or exponential utilities.\qed
\end{remark}

%
%
\section{Methodology and Main Results}\label{sec:OptimalPolicies}

Our goal in this section is to solve the stochastic control problem \cref{eq:VF} under the general S-shaped multiplicative habit formation preference and to obtain the optimal relative investment and consumption policies in feedback form. To this end, we follow a three-step procedure described below: 

\noindent\textbf{Step-1}: We first consider the concavified version of \cref{eq:VF} by replacing the utility function $U(c)$ with its concave envelope $\Ut(c)$, see \cref{eq:Vt} below. We write down the Hamilton-Jacobi-Bellman (HJB) equation for the concavified problem, and transform it to a free-boundary problem \cref{eq:HJB-FBP1-moderate}  based on the conjectured optimal consumption control in 
\cref{eq:sup_c-moderate}. The arguments in this step are mainly motivational. Rigorous arguments are postponed to Steps 2 and 3.

\noindent\textbf{Step-2}: We solve the free-boundary problem and obtain a candidate value function $v(x)$ for the concavified problem (see corollary \ref{coro:solveHJB}). Along the way, we also establish several important properties of $v(x)$ and verify the conjectures in \textbf{Step-1}. 

\noindent\textbf{Step-3}: Based on the obtained properties of $v(x)$, we verify the optimality of the feedback controls and show that the solution $v(x)$ to the HJB equation coincides with both the value function $\Vt(x)$ of the concavified problem \cref{eq:Vt} and the value function $V(x)$ of the original problem \cref{eq:VF}. 

\subsection{Step-1: The concavified problem and its HJB equation}\label{subsec:candidatepolicies}
To study the non-concave stochastic control problem \cref{eq:VF} under the S-shaped utility, we follow the routine by considering its \emph{concavified} formulation (see \cite{Carp}), namely,
\begin{align}\label{eq:Vt}
	\Vt(x) = \sup_{(c,\pi)\in\Acr(x)} \Eb\left[\int_0^{+\infty}\ee^{-\del t} \Ut\left(c_t\right)\dd t\right],\quad x>0,
\end{align}
where $\Ut$  is the concave envelope of the S-shaped utility $U(c)$ in \cref{eq:U}, that is,
\begin{align}\label{eq:Ut_defined}
	\Ut(c) := \sup_{s,t}\left\{\frac{(t-c)U(t)+(c-s)U(s)}{t-s}\,:\, 0\le s\le c\le t\right\},\quad c\ge0.
\end{align} 

To this end, some growth conditions on the S-shaped utility are necessarily needed. In fact, Figure \ref{fig:U_smooth_nonsmooth} illustrates two possible choices of the utility function $U(c)$ along with their concave envelopes $\Ut(c)$ given by \cref{eq:Ut_defined}. It is worth noting here that, unlike most of the literature on loss-averse preferences, the general S-shaped utility allows for the loss-averse and the risk-averse parts of the utility to be of different classes. For example, in \cref{eq:U}, we may choose a power utility $U_+(c)=\frac{(c+\eps)^p}{p}-\frac{\eps^p}{p}$ for the risk-averse part, and an exponential utility $U_-(c) =1-\ee^{-q c}$ (or a log-utility $U_-(c)=\log\big((c+q)/q\big)$) for the loss-averse part, with $\eps>0$, $p<1$, and $q>0$. Furthermore, we allow for the S-shaped utility to be non-differentiable at the loss reference point (i.e. $U'(\al^-)\ne U'(\al^+)$), as in the case of the left plot in Figure \ref{fig:U_smooth_nonsmooth}.

%
%

\begin{figure}[htbp]
	\centerline{
			\adjustbox{trim={0.0\width} {0.02\height} {0.0\width} {0.0\height},clip}
			{\includegraphics[scale=0.325, page=1]{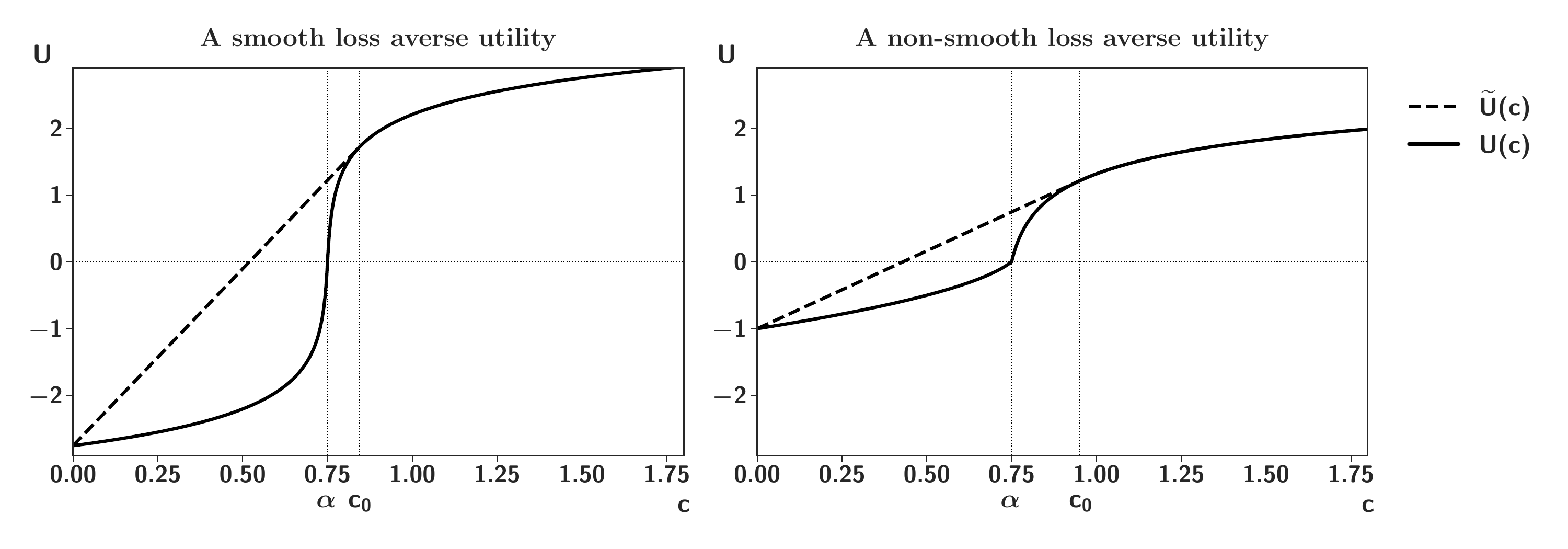}}
	}
	\caption{Plots of two possible S-shaped utility functions $U(c)$ (the solid curves) and their concave envelopes $\Ut(c)$ (the dashed curves), where we have set $\al=0.75$. Note that $U(c)$ and $\Ut(c)$ coincide on $[c_0,+\infty)$, and that $\Ut$ is linear on $[0,c_0]$. The constant $c_0$ is defined in Lemma \ref{lem:Ut} in Section \ref{sec:OptimalPolicies}, and is the unique solution of $c_0 U_+'(c_0-\al) - U_+(c_0-\al)=U_-(\al)$. Note, also, that the loss and gain utility functions $U_\pm$ in \cref{eq:Uc} can be different, and the S-shaped utility $U(c)$ can be non-differentiable at the reference point $\al$ (as in the right plot).
		\label{fig:U_smooth_nonsmooth}}
\end{figure}

We impose two conditions on the general utility function $U(c)$, which will be used later to obtain and verify the optimal controls.  The first assumption guarantees that $U(c)$ has a smooth concave envelope; see Lemma \ref{lem:Ut} below.
\begin{assumption}\label{assumption:moderateLA}
	The utility functions $U_\pm(c)$ in \cref{eq:U} satisfy $U_-(\al)\le\al U_+'(0)$.\qed
\end{assumption}

\begin{remark}\label{rem:moderateLA}
	Assumption \ref{assumption:moderateLA} is satisfied by the S-shaped utility functions commonly used in the literature, on which one of the following assumptions is typically imposed:	
	\begin{itemize}
		\item $U'(\al^+)=U_+'(0^+)=+\infty$, such as two-part power utilities (see \cref{eq:U_power}); or,
		\item $U(c)$ is differentiable at the reference point $\al$, that is, $U_+'(0)=U_-'(0)$.
	\end{itemize}
	Indeed, Assumption \ref{assumption:moderateLA} trivially holds under the first condition, while under the second condition, Assumption \ref{assumption:moderateLA} also holds because, by the concavity of $U_-(c)$, we have $U_-(\al)\le \al U_-'(0)=\al U_+'(0)$.\qed
\end{remark}

Our second assumption on the S-shaped utility function $U(c)$ is a growth condition on the gain utility $U_+(c)$ as $c\to+\infty$. This assumption is needed to facilitate the proof of the transversality condition for the stochastic control problem \cref{eq:VF}, see Lemma \ref{lem:trans} in Subsection \ref{subsec:verify}.
\begin{assumption}\label{assum:Growth}
	There exist constants $a_1,a_2, a_3>0$, and $a_4\ge0$ such that the gain utility function $U_+(\cdot)$ in \cref{eq:U} satisfies the growth condition $U_+(x) \le a_1+x U_+'(x)+a_2 U_+'(x)^{-a_4}$ for all $x>a_3$.\qed
\end{assumption}

\begin{remark}\label{rem:Growth}
	It is straightforward to check that Assumption \ref{assum:Growth} holds for the power, log, and Symmetric Asymptotic Hyperbolic Absolute Risk Aversion (SAHARA) utility functions (see \cref{eq:SAHARA}), as well as any upper bounded utility function such as the exponential utility. 
	In particular, note that we allow for the full range of power utilities, i.e. $U(x)=x^{1-\gam}/(1-\gam)$ for $\gam>0$. As pointed out in the introduction, existing results on multiplicative habit formation only consider power utilities for $\gam>1$.\qed
\end{remark}

We expect that problems \cref{eq:VF} and \cref{eq:Vt} are equivalent in the sense that $V(x)=\Vt(x)$ for $x>0$, and that they share a common optimal control $(\cs,\pis)$. These results will be proved in Theorem \ref{thm:verification_moderate} in Subsection \ref{subsec:verify} below.

The Hamilton-Jacobi-Bellman (HJB) equation corresponding to \cref{eq:Vt} is
\begin{align}\label{eq:HJB-concavified}
	&\sup_{\pi\in\Rb, c\ge 0}\left\{ \textstyle
	-\del v(x) + \big((r+\rho) x + \mu \pi - (1+\rho x) c\big) v'(x) + \frac{1}{2}\sig^2\pi^2 v^{\prime\prime}(x) + \Ut(c)
	\right\}= 0,
\end{align}
for $x>0$. Taking the ansatz $v^{\prime\prime}(x)<0$ for the solution of \cref{eq:HJB-concavified}, we have the maximizer
\begin{align}\label{eq:piStar}
	\pi^*(x):=\underset{\pi\in\Rb}{\arg\max}\left\{\frac{1}{2}\sig^2\pi^2 v^{\prime\prime}(x)+\mu \pi v'(x)\right\}
	=-\frac{\mu}{\sig^2}\frac{v'(x)}{v^{\prime\prime}(x)}.
\end{align}
The HJB equation \cref{eq:HJB-concavified} then becomes
\begin{align}\label{eq:HJB_c}
	&\sup_{c\ge 0}\left\{\Ut(c) - c(1+\rho x)v'(x)\right\}
	-\frac{\mu^2}{2\sig^2} \frac{v'(x)^2}{v^{\prime\prime}(x)}
	+(r+\rho) x v'(x) - \del v(x)=0,
	\quad x\ge0.\qquad
\end{align}
We use the next lemma to evaluate the term involving the maximization over $c$. The lemma relies on Assumption \ref{assumption:moderateLA} above. Its proof is straightforward and thus omitted.
\begin{lemma}\label{lem:Ut}
	Consider the S-shaped utility $U(c)$ in \cref{eq:U} with $\al>0$ and $U_\pm(c)$ satisfying \cref{eq:Uc} and Assumption \ref{assumption:moderateLA}. 
	The concave envelope $\Ut(c)$ in \cref{eq:Ut_defined} is given by
	\begin{align}\label{eq:Ut-moderate}
		\Ut(c)=
		\begin{cases}
			-U_-(\al)+\phi_0 c, &\quad 0\le c\le c_0,\\
			U(c)= U_+(c-\al), &\quad c>c_0,
		\end{cases}
	\end{align}
	in which $\phi_0:= U_+'(c_0-\al)$ and $c_0\ge \al$ is the unique solution of $c_0 U_+'(c_0-\al) - U_+(c_0-\al)=U_-(\al)$. 
	Furthermore, for any $\phi>0$, $\ch(\phi)$ given by
	\begin{align}\label{eq:Ut_c-moderate}
		\ch(\phi):=
		\begin{cases}
			0,&\quad \phi>\phi_0,\\
			\al+(U_+')^{-1}(\phi),&\quad 0<\phi\le \phi_0,
		\end{cases}\qquad
	\end{align}
	is a common maximizer for both problems $\displaystyle\max_{c\ge0} \left\{U(c)-\phi c\right\}$ and $\displaystyle\max_{c\ge0}\left\{\Ut(c)-\phi c\right\}$.\qed	
\end{lemma}


With \cref{eq:Ut_c-moderate} in mind, we speculate that there exists a (yet to be determined) constant $x_0\ge0$  such that $v(x)$ (i.e. the solution of \cref{eq:HJB_c}) satisfies
\begin{align}\label{eq:Ansatz2-moderate}
	\begin{cases}
		(1+\rho x)v'(x)> \phi_0,\quad 0\le x < x_0,\\
		(1+\rho x_0)v'(x_0)= \phi_0,\\
		0<(1+\rho x)v'(x)< \phi_0,\quad x> x_0,
	\end{cases}
\end{align}
for the constant $\phi_0$ given in Lemma \ref{lem:Ut}. These conjectured inequalities will be verified later in Corollary \ref{coro:solveHJB} in Subsection \ref{subsec:HJB_sol}. Setting $\phi= (1+\rho x)v'(x)$ in \cref{eq:Ut_c-moderate} and using \cref{eq:Ansatz2-moderate}, we readily get that
\begin{align}\label{eq:sup_c-moderate}
	\sup_{c\ge 0}\left\{\Ut(c) - c(1+\rho x)v'(x)\right\}=
	\begin{cases}
		\Ut(0)=U(0)
		,&\quad 0\le x < x_0,\\
		G\big((1+\rho x)v'(x)\big),&\quad x\ge x_0,
	\end{cases}\qquad
\end{align}
where we define
\begin{align}\label{eq:G}
	G(\phi) &:= 
	U_+\big((U_+')^{-1}(\phi)\big) - \phi \big(\al+(U_+')^{-1}(\phi)\big),
	\quad 0<\phi\le \phi_0.\qquad
\end{align}
For future reference, we also note that the maximizer of $c$ in \cref{eq:sup_c-moderate} is given by
\begin{align}\label{eq:cs-moderate}
	\cs(x):=
	\begin{cases}
		0,&\quad 0\le x < x_0,\\
		\al+(U_+')^{-1}\big((1+\rho x)v'(x)\big),&\quad x\ge x_0.
	\end{cases}\qquad
\end{align}

\begin{remark}
	Note that the maximizer $\cs(x)$ in \cref{eq:cs-moderate} has a jump at $x=x_0$ because $\lim_{x\to x_0^-}\cs(x)=0$, but
	\begin{align*}
		\cs(x_0)=\al+(U_+')^{-1}\big((1+\rho x_0)v'(x_0)\big)
		=\al+(U_+')^{-1}\big(\phi_0\big)=c_0\ge \al >0.
	\end{align*}
	Here, we have used \cref{eq:Ansatz2-moderate} and that $\phi_0:=U_+'(c_0-\al)$ with the constant $c_0$ in Lemma \ref{lem:Ut}. 
	Note also that there is no discontinuity in the optimal value given by \cref{eq:sup_c-moderate} at $x=x_0$. This can be confirmed as follows. By \cref{eq:Ansatz2-moderate} and \cref{eq:G}, we have
	\begin{align}\nonumber
		G\big((1+\rho x_0)v'(x_0)\big)
		&=G(\phi_0)
		=U_+\big((U_+')^{-1}(\phi_0)\big) - \phi_0 \big(\al+(U_+')^{-1}(\phi_0)\big)\\
		\label{eq:Psi-phi0}
		&=U_+\big(c_0-\al\big) - U_+'(c_0-\al) c_0
		=-U_-(\al)=U(0),
	\end{align}
	in which the second-to-last step follows from $c_0 U_+'(c_0-\al) - U_+(c_0-\al)=U_-(\al)$ (which holds by the definition of $c_0$ in Lemma \ref{lem:Ut}), and the last step follows from \cref{eq:Ut-moderate}.\qed
\end{remark}

By using \cref{eq:sup_c-moderate}, the HJB equation \cref{eq:HJB_c} becomes the following free-boundary problem for the function $v(x)$ and with the free boundary $x_0\ge 0$,
\begin{align}\label{eq:HJB-FBP1-moderate}
	&\begin{cases}
		\displaystyle
		-\frac{\mu^2}{2\sig^2} \frac{v'(x)^2}{v^{\prime\prime}(x)}
		+(r+\rho) x v'(x) - \del v(x) + U(0)=0
		,&\quad 0\le x < x_0,\vspace{1ex}\\
		\displaystyle
		-\frac{\mu^2}{2\sig^2} \frac{v'(x)^2}{v^{\prime\prime}(x)}
		+ G\big((1+\rho x)v'(x)\big)
		+ (r+\rho) x v'(x) - \del v(x)=0,&\quad x\ge x_0,\vspace{1ex}\\
		(1+\rho x_0)v'(x_0) = \phi_0,\vspace{1ex}\\
		\displaystyle
		\lim_{x\to0^+}\frac{v'(x)}{v^{\prime\prime}(x)} = 0,
	\end{cases}
\end{align}
in which $G(\phi)$ is given by \cref{eq:G} and $\phi_0$ is defined in Lemma \ref{lem:Ut}. In \cref{eq:HJB-FBP1-moderate}, the free boundary conditions at $x=x_0$ follows from \cref{eq:Ansatz2-moderate}. Furthermore, in light of \cref{eq:piStar}, the initial condition at $x\to0^+$ means that once the wealth hits zero, it is optimal not to invest in the risky asset (in fact, this is the only admissible investment policy at $x=0$).

By reversing the argument leading to \cref{eq:HJB-FBP1-moderate}, it follows that if a solution $\big(v(x),x_0\big)$ of the free-boundary problem \cref{eq:HJB-FBP1-moderate} satisfies  \cref{eq:Ansatz2-moderate} and $v^{\prime\prime}(x)<0$ for $x>0$, then $v(x)$ is a solution of the HJB equation \cref{eq:HJB_c}. In the next subsection, we establish the existence of a solution $\big(v(x),x_0\big)$ to this free boundary problem, see Corollary \ref{coro:solveHJB}.

\subsection{Step-2: Auxiliary nonlinear free boundary problems} \label{subsec:HJB_sol}
To study problem \cref{eq:HJB-FBP1-moderate}, we first consider the Legendre transform
\begin{align}\label{eq:Legendre}
	u(y) := \sup_{x>0}\big\{v(x)-xy\big\}, y>0,
\end{align}
which implies the following relationships between $v(x)$ and $u(y)$
\begin{align}\label{eq:Legendre-Relations}
	\begin{cases}
		v\big((v')^{-1}(y)\big) = u(y) - y u'(y),\\
		(v')^{-1}(y) = -u'(y),\\
		v^{\prime\prime}\big((v')^{-1}(y)\big)=-1/u^{\prime\prime}(y).
	\end{cases}
\end{align}
By applying the change of variable $y=v'(x)~\Leftrightarrow~x=(v')^{-1}(y)$ and using \cref{eq:Legendre-Relations}, we can rewrite the nonlinear free-boundary problem \cref{eq:HJB-FBP1-moderate} for $x\geq 0$ as the dual nonlinear free-boundary problem for $u(y)$ for $y>0$ with the free boundary $y_0=v'(x_0)>0$,
\begin{align}
	&\displaystyle\label{eq:modFBPDual1}
	\frac{\mu^2}{2\sig^2} y^2 u^{\prime\prime}(y) + (\del - r-\rho) y u'(y) - \del u(y) + U(0)=0,
	\hspace{10.5em} y> y_0,\\
	&\displaystyle\label{eq:modFBPDual2}
	\frac{\mu^2}{2\sig^2} y^2 u^{\prime\prime}(y) + G\big(y-\rho y u'(y)\big)
	+ (\del - r - \rho) y u'(y) - \del u(y)=0,
	\hspace{4.25em}  0< y \le y_0,\\
	&\displaystyle\label{eq:modFBPDual3}
	y_0 - \rho y_0 u'(y_0) = \phi_0,\\
	\intertext{and}
	\label{eq:modFBPDual4}
	&\lim_{y\to+\infty} u'(y) = \lim_{y\to+\infty}y u^{\prime\prime}(y) = 0.
\end{align}

One advantage of working with the dual problem \cref{eq:modFBPDual1}--\cref{eq:modFBPDual4} (instead of the original nonlinear free boundary problem \cref{eq:HJB-FBP1-moderate}) lies in the fact that \cref{eq:modFBPDual1} is a linear Euler equation, which admits a general solution of the explicit form
\begin{align}\label{eq:Euler-moderate}
	u(y) = A_1 y^{\lam} + A_2 y^{\lam'} + \frac{U(0)}{\del}, &\quad y>y_0,
\end{align}
in which $A_1$ and $A_2$ are constants to be determined, 
\begin{align*}
	\lam &:= \frac{\sig^2}{\mu^2}\left( \frac{\mu^2}{2\sig^2} + r + \rho - \del - \sqrt{\left(\frac{\mu^2}{2\sig^2} + r + \rho - \del\right)^2 + 2\del\frac{\mu^2}{\sig^2}}\right)
	\in\left(-\frac{2\del\sig^2}{\mu^2},0\right),
	\intertext{and}
	\lam'&:= -\frac{2\del\sig^2}{\mu^2\lam}>1.
\end{align*}

\begin{remark}\label{rem:lam}
	That
	\begin{align}\label{eq:lam_bounds}
		-\frac{2\del\sig^2}{\mu^2}<\lam<0,\quad \text{and}\quad \lam'>1,
	\end{align}
	are shown as follows. Note that $\lam$ and $\lam'$ are, respectively, the positive and the negative roots of the quadratic equation $f(x):=\mu^2x^2/(2\sig^2)-\big(\mu^2/(2\sig^2)+r+\rho-\del\big)x - \del=0$. As $f(-2\del\sig^2/\mu^2)=2(r+\rho)\del\sig^2/\mu^2>0$ and $f(0)=-\del<0$, it follows that $-2\del\sig^2/\mu^2<\lam<0$ which, in turn, yields that $\lam':= -2\del\sig^2/(\mu^2\lam)>1$. 
	\qed
\end{remark}

Differentiating \cref{eq:Euler-moderate}  yields
$u'(y) = \lam\,A_1 y^{\lam - 1 } + \lam'A_2y^{\lam' - 1}$ and
$yu^{\prime\prime}(y) = \lam(\lam-1)A_1 y^{\lam - 1 } + \lam'(\lam' -1)A_2 y^{\lam' - 1}$, for $y>y_0$. 
In light of \cref{eq:lam_bounds}, the boundary condition \cref{eq:modFBPDual4} holds only if $A_2=0$. 
From \cref{eq:modFBPDual3}, we then obtain the value of $A_1$. That is,
\begin{align*}
	\phi_0 = y_0 - \rho y_0 u'(y_0)
	=y_0 - \rho y_0  \lam A_1 y_0^{\lam -1}
	\quad \Longrightarrow\quad
	A_1  = \frac{y_0 - \phi_0}{\rho \lam\, y_0^{\lam}}.
\end{align*}
Substituting $A_1$ and $A_2$ back into \cref{eq:Euler-moderate} yields
\begin{align}\label{eq:u1-moderate}
	u(y) = \frac{y_0 - \phi_0}{\rho \lam} \left(\frac{y}{y_0}\right)^{\lam} + \frac{U(0)}{\del}, &\quad y>y_0.
\end{align}

Thus, by the smooth-fit principle, we can reduce the free-boundary problem \cref{eq:modFBPDual1}--\cref{eq:modFBPDual4} to the following free-boundary problem only for $0<y\le y_0$ with the free boundary $y_0>0$,
\begin{align}\label{eq:modFBPDual-reduced}
	\begin{cases}\displaystyle
		\frac{\mu^2}{2\sig^2} y^2 u^{\prime\prime}(y)
		+ G\big(y-\rho y u'(y)\big) 
		+ (\del - r - \rho) y u'(y) - \del u(y)=0,
		&\quad 0< y \le y_0,\vspace{1ex}\\
		\displaystyle
		y_0 - \rho y_0 u'(y_0) = \phi_0,\vspace{1ex}\\
		\displaystyle
		u(y_0) = \frac{y_0 - \phi_0}{\rho \lam} + \frac{U(0)}{\del}.
	\end{cases}
\end{align}

In particular, the additional explicit free boundary condition $u(y_0) = \frac{y_0 - \phi_0}{\rho \lam} + \frac{U(0)}{\del}$ will play an important role in showing the existence of the classical solution to \cref{eq:modFBPDual-reduced}. Once this free-boundary problem is solved, we in turn obtain a solution of the free-boundary problem \cref{eq:modFBPDual1}--\cref{eq:modFBPDual4} by pasting the solution in \cref{eq:u1-moderate}.

We analyze the second-order differential equation in \cref{eq:modFBPDual-reduced} by transforming it into a first-order system. Consider the change of variables
\begin{align}
	\varphi(y)&:= y-\rho y u'(y),&&\quad0<y<y_0,\label{eq:u_phi_alt}\\
	\psi(y)&:= \frac{\rho y u''(y)}{1-\rho u'(y)},&&\quad0<y<y_0,\nonumber
	\intertext{which is equivalent to}
	u'(y)&=\frac{y-\varphi(y)}{\rho y}, &&\quad 0<y<y_0, \label{up_phi_main-alt}\\
	u^{\prime\prime}(y)&=\frac{1}{\rho y^2} \varphi(y)\psi(y), &&\quad 0<y<y_0.\label{upp_phi_psi_main-alt}
\end{align}
Differentiating \cref{eq:u_phi_alt} and substituting the expressions for $u'$ and $u''$ from \cref{up_phi_main-alt}--\cref{upp_phi_psi_main-alt} yields
\begin{align}\label{eq:phi-ODE-alt}
	\varphi'(y)=\frac{1}{y}\,\varphi(y)\bigl(1-\psi(y)\bigr),
	\qquad 0<y\le y_0 .
\end{align}
Next, substituting \cref{up_phi_main-alt}--\cref{upp_phi_psi_main-alt} into the differential equation \cref{eq:modFBPDual-reduced} gives
\begin{align*}
	\delta\,u(y)
	= \frac{\mu^2}{2\rho\sigma^2}\,\varphi(y)\psi(y)
	+ G\bigl(\varphi(y)\bigr)
	+ \frac{\delta-r-\rho}{\rho}\bigl(y-\varphi(y)\bigr),
	\qquad 0<y\le y_0.
\end{align*}
Differentiating this identity, using $G'(\phi)=-\alpha-(U_+')^{-1}(\phi)$ from \cref{eq:G}, and substituting $u'$ from \cref{up_phi_main-alt} together with $\varphi'$ from \cref{eq:phi-ODE-alt}, we obtain
\begin{align}\label{eq:psi-ODE-alt}
	\psi'(y) = \textstyle-\frac{2\rho\sig^2}{\mu^2} 
	\left[
	\frac{1-\psi(y)}{y}
	\left(\frac{\mu^2}{2\rho\sig^2}\psi(y) - (U_+')^{-1}\big(\varphi(y)\big)+\frac{r-\del}{\rho}+1-\al\right)
	-\frac{r+\rho}{\rho \varphi(y)} + \frac{\del}{\rho y}\right],
\end{align}
for $0<y\le y_0$. Equations \cref{eq:phi-ODE-alt}–\cref{eq:psi-ODE-alt} therefore form a system of first–order differential equations for $\varphi$ and $\psi$.

The following theorem is our first main result of this section. It addresses the solution of the nonlinear free-boundary problem \cref{eq:modFBPDual-reduced} by considering another auxiliary system of first-order free boundary problems. In particular, we first introduce the auxiliary functions $\varphi(y)$ and $\psi(y)$ as the solution of \cref{eq:ODESys-moderate}, for which we establish several important properties that will be further used in future verification arguments. Its proof is technical and lengthy, which is postponed to Subsections \ref{app:Aux} and \ref{proof:moderate_dual_FBP}.
\begin{theorem}\label{thm:DualFBP-moderate}
	$(i)$ There exists a constant $y_0\in\left(\frac{\phi_0\lam}{\lam\,-\,1},\phi_0\right)$, an increasing function $\varphi:(0,y_0]\to (0,\phi_0)$, and a function $\psi:(0,y_0]\to(0,1)$ satisfying the coupled system of first-order free boundary ODEs 
	\begin{align}\label{eq:ODESys-moderate}
		\begin{cases}
			\varphi'(y) = \frac{1}{y}\varphi(y)\big(1-\psi(y)\big),\\[1ex]
			\psi'(y) = -\frac{2\rho\sig^2}{\mu^2} 
			\left[
			\frac{1-\psi(y)}{y}
			\left(\frac{\mu^2}{2\rho\sig^2}\psi(y) - (U_+')^{-1}\big(\varphi(y)\big)+\frac{r-\del}{\rho}+1-\al\right)
			-\frac{r+\rho}{\rho \varphi(y)} + \frac{\del}{\rho y}
			\right],
		\end{cases}
	\end{align}
	for $0<y\le y_0$ with the free boundary conditions
	\begin{align}\label{eq:FB-moderate}
		\begin{cases}
			\varphi(y_0) = \phi_0,\vspace{1ex}\\
			\psi(y_0)
			= \frac{2\sig^2}{\phi_0\mu^2}\left(\frac{\del}{ \lam} +r + \rho-\del\right) (y_0-\phi_0).
		\end{cases}
	\end{align}
	Moreover, at least one of these equations holds: $\lim_{y\to 0+} \varphi(y)>0$ or $\lim_{y\to 0+}\psi(y)=1$.\vspace{1em}
	
	\noindent$(ii)$ Given the solution $(\varphi(y), \psi(y), y_0)$ in part $(i)$, 
	let us define	
	\begin{align}\label{eq:u_phi_psi_moderate}
		u(y) := \frac{1}{\del}\left[
		\frac{\mu^2}{2\rho\sig^2} \varphi(y)\psi(y)
		+ G\big(\varphi(y)\big) +\frac{\del - r - \rho}{\rho} \big(y-\varphi(y)\big)
		\right],
		\quad 0<y\le y_0.
	\end{align}
	It then holds that $\big(u(y),y_0\big)$ is a solution of the free-boundary problem \cref{eq:modFBPDual-reduced}. Furthermore, we have that
	\begin{align}
		&u'(y)=\frac{y-\varphi(y)}{\rho y}, \quad 0<y<y_0, \label{up_phi_main}\\
		&u^{\prime\prime}(y)=\frac{1}{\rho y^2} \varphi(y)\psi(y), \quad 0<y<y_0,\label{upp_phi_psi_main}
	\end{align}
	$u(y)$ is decreasing and convex on $(0,y_0)$, and $\lim_{y\to 0+}u'(y)=-\infty$.\qed
\end{theorem}

\begin{proof}
	See Subsections \ref{app:Aux} and \ref{proof:moderate_dual_FBP}.
\end{proof}

As a result, we obtain the solution of the free-boundary problem \cref{eq:modFBPDual1}--\cref{eq:modFBPDual4}.

\begin{corollary}\label{coro:solution_dual}
	Let $y_0$ be the constant in Theorem \ref{thm:DualFBP-moderate}.(i), and assume that $u:(0,\infty)\to \R$ is given by \cref{eq:u1-moderate} and \cref{eq:u_phi_psi_moderate}. Then, $(u(y), y_0)$ solves the free boundary problem \cref{eq:modFBPDual1}--\cref{eq:modFBPDual4}. Furthermore, $u\in C^2(0,\infty)$ is strictly decreasing and strictly convex, $u'(0^+)=-\infty$, and $u'(+\infty)=0$.\qed
\end{corollary}
\begin{proof}
	These statements are direct consequences of Theorem \ref{thm:DualFBP-moderate} and the explicit expression \cref{eq:u1-moderate}.
\end{proof}

As a second corollary of Theorem \ref{thm:DualFBP-moderate}, we can now solve the original free-boundary problem \cref{eq:HJB-FBP1-moderate} for $(v(x), x_0)$ by the inverting the Legendre transform \cref{eq:Legendre}.

\begin{corollary}\label{coro:solveHJB}
	Let $(u(y), y_0)$ be the solution of the free boundary problem \cref{eq:modFBPDual1}--\cref{eq:modFBPDual4} in Corollary \ref{coro:solution_dual}. Then, $(v(x),x_0)$ given by
	\begin{align}\label{v_moderate}
		\begin{cases}
			\displaystyle v(x) := u\left((u')^{-1}(-x)\right)+x\,(u')^{-1}(-x),\quad x>0,\\[1ex]
			\displaystyle x_0 :=\frac{\phi_0-y_0}{\rho y_0},
		\end{cases}
	\end{align}
	satisfy \cref{eq:HJB-FBP1-moderate}, \cref{eq:Ansatz2-moderate}, $v'(+\infty)=0$, $v'(0^+)=+\infty$, $v'(x)=(u')^{-1}(-x)>0$, and $v^{\prime\prime}(x)=-1/u^{\prime\prime}\big(v'(x)\big)<0$ for $x>0$.\qed
\end{corollary}
\begin{proof}
	For ease of notation, let us define $J:(-\infty,0)\to (0,\infty)$ by $J:=(u')^{-1}$. We derive from \cref{v_moderate} and the chain rule that $v'(x) = J(-x)>0$ and $v^{\prime\prime}(x) = -1/u^{\prime\prime}(J(-x))<0$ for $x>0$, in which the inequalities hold thanks to the fact that $u(y)$ is strictly decreasing and convex by Corollary \ref{coro:solution_dual}. That $v'(+\infty)=0$ and $v'(0^+)=+\infty$ follows from $u'(0^+)=-\infty$ and $u'(+\infty)=0$ shown in Corollary \ref{coro:solution_dual}. By \cref{up_phi_main}, $x_0 :=\frac{\phi_0-y_0}{\rho y_0} = -u'(y_0)$. 
	Therefore, by \cref{eq:Legendre-Relations} and the correspondence $y = v'(x)=J(-x) ~\Leftrightarrow~ x = -u'(y)$, we conclude that $v(x)$ and $x_0$ (given by \cref{v_moderate}) satisfy \cref{eq:HJB-FBP1-moderate}. To verify \cref{eq:Ansatz2-moderate}, we note that
	\begin{align*}
		\frac{\dd}{\dd x} \Big( (1+\rho x)v'(x)\Big)
		&= \rho v'(x) + (1+\rho x)v^{\prime\prime}(x)
		= \rho y - \frac{1-\rho u'(y)}{u^{\prime\prime}(y)}
		= \rho y\left(1-\frac{1}{\psi(y)}\right) < 0,
	\end{align*}
	for $0<y<y_0$ (equivalently, $x>x_0$). Here, we have used \cref{up_phi_main} and \cref{upp_phi_psi_main}. Moreover,  we have $u'(y) = \frac{y_0-\phi_0}{\rho} \frac{y^{\lambda-1}}{y_0^\lambda}$ and $u^{\prime\prime}(y) = \frac{(y_0-\phi_0)(\lambda-1)}{\rho}\, \frac{y^{\lambda-2}}{y_0^\lambda}$, for $y>y_0$.  
	Therefore,
	\begin{align*}
		\frac{\dd}{\dd x} \Big( (1+\rho x)v'(x) \Big) 
		&= \rho v'(x) + (1+\rho x)v^{\prime\prime}(x)
		= \rho y - \frac{1-\rho u'(y)}{u^{\prime\prime}(y)}
		= \lambda \rho y - \frac{1}{u^{\prime\prime}(y)}
		< 0.
	\end{align*}
	In view of $(1+\rho x_0)v'(x_0)=\phi_0$, \cref{eq:Ansatz2-moderate} readily follows.
\end{proof}

\subsection{Step-3: The verification theorem and the optimal feedback controls} \label{subsec:verify}

In this subsection, our goal is to establish the verification theorem that provides the optimal investment and consumption policies in feedback form and in terms of the solution $(\varphi(y),\psi(y), y_0)$ of the auxiliary system of first-order free boundary problems \cref{eq:ODESys-moderate}-\cref{eq:FB-moderate}.

We need the next result regarding $v(x)$ (of Corollary \ref{coro:solveHJB}), which is the so-called transversality condition. Its proof, included in Subsection \ref{app:trans}, relies on Assumption \ref{assum:Growth} as well as various properties of the solution $(\varphi(y),\psi(y), y_0)$ of  \cref{eq:ODESys-moderate}-\cref{eq:FB-moderate} that was established by Theorem \ref{thm:DualFBP-moderate}.

\begin{lemma}\label{lem:trans}
	Suppose that the utility function $U(\cdot)$ satisfies \cref{eq:U}, \cref{eq:Uc}, Assumption \ref{assumption:moderateLA}, and Assumption \ref{assum:Growth}. Then, the function $v(\cdot)$ of \cref{v_moderate} fulfills the transversality condition
	\begin{equation}\label{eq:trans}
		\lim_{T \to \infty}\Eb \left[e^{-\delta T}v(X_T) \right]=0,
	\end{equation}
	for all $(\pi_t,c_t)_{t\geq 0}\in \Acr(x)$ and all $x>0$, in which $\{X_t\}_{t\ge0}$ is given by \cref{eq:X}.\qed
\end{lemma}
\begin{proof}
	See Subsection \ref{app:trans}.
\end{proof}

The next theorem is the main result of the paper. It verifies that the function $v(x)$ given by Corollary \ref{coro:solveHJB} coincides with both the value function $\Vt(x)$ of the concavified problem \cref{eq:Vt} and the value function $V(x)$ of the original problem \cref{eq:VF}. It also states that the functions $\pi^*(x)$ and $\cs(x)$, respectively given by \cref{eq:piStar} and \cref{eq:cs-moderate}, provide feedback optimal controls for both the concavified and the original problems.

\begin{theorem}\label{thm:verification_moderate}
	Let $x_0$ and $v(x)$ be given by \cref{v_moderate}, and $\varphi(y)$ and $\psi(y)$ be as in Theorem \ref{thm:DualFBP-moderate}.
	Define also the feedback functions
	\begin{align}
		\label{c_optimal_moderate}
		&c^*(x) = 
		\begin{cases}
			0,&\quad 0< x < x_0,\\
			\al+(U_+')^{-1}\big(\varphi(v'(x))\big),&\quad x\ge x_0,
		\end{cases}\\
		\intertext{and}
		\label{pi_optimal_moderate}
		&\pi^*(x) = 
		\begin{cases}
			\frac{\mu(1-\lambda)}{\sigma^2}x,&\quad 0<x<x_0, \\
			\frac{\mu}{\sigma^2\rho}(1+\rho x)\psi(v'(x)), &\quad x\geq x_0.
		\end{cases}
	\end{align}
	It then holds that $V(x)=\Vt(x)=v(x)$ for all $x>0$, in which $V(x)$ and $\Vt(x)$ are value functions of problems \cref{eq:VF} and \cref{eq:Vt}, respectively. Furthermore, for any $x>0$, the SDE
	\begin{equation}\label{SDE:optimal}
		\begin{cases}
			& \dd X^*_t = \Big((r+\rho)X^*_t + \mu\pi^*(X^*_t)- (1+\rho X^*_t)c^*(X^*_t)\Big) \dd t + \sig \pi^*(X^*_t)\dd B_t,\\
			& X^*_0 = x,
		\end{cases}
	\end{equation}
	admits a unique strong solution $\{X^*_t\}_{t\ge0}$, and $\big\{\big(\pi^*(X^*_t),c^*(X^*_t)\big)\big\}_{t\geq 0}$ is a common optimal feedback control pair for both problems \cref{eq:VF} and \cref{eq:Vt}.\qed
\end{theorem}
\begin{proof}
	See Subsection \ref{app:verify}.
\end{proof}

We end this section by a brief discussion of the optimal policies characterized in \cref{c_optimal_moderate} and \cref{pi_optimal_moderate}. A comprehensive examination of these policies and numerical illustrations are given in the next section.
The agent's optimal investment and consumption depend on whether she should take austerity measures or not. At time $t\ge0$, let $W^*_t$ be her (optimally controlled) current wealth and $H^*_t$ be her current consumption habit under the consumption control process $\{C^*_s\}_{0\le s<t}$. Then,
\begin{itemize}
	\item If her wealth-to-habit ratio $X^*_t=W^*_t/H^*_t$ is below the \emph{austerity threshold} $x_0$, she will take austerity measures by not consuming (i.e. $C^*_t=c =0$ by \cref{eq:ratios} and \cref{c_optimal_moderate}) and investing a constant proportion $\mu(1-\lam)/\sig^2$ of her wealth in the risky asset (in view that $\Pi^*_t/W^*_t=\pi^*(X^*_t)/X^*_t=\mu(1-\lam)/\sig^2$ according to \cref{eq:ratios} and \cref{pi_optimal_moderate}). In particular, the agent takes this austerity measure because she is loss averse when her consumption rate is in the range $(0, \al H^*_t)$. For low levels of wealth (determined precisely by the condition $0<W^*_t\le x_0 H^*_t$), the agent prefers to avoid consumption and build up the wealth instead. More importantly, it is interesting to highlight two observations on the constant proportion $\mu(1-\lam)/\sig^2$ of wealth invested in the risky asset when the wealth-to-habit falls below the austerity threshold $x_0$: (i) It is strictly larger than the Kelly's fraction $\mu/\sigma^2$  as $\lambda<0$ in our model; (ii) It is universal for any choice of utility function $U(x)$. The first observation can be explained by the fact that the loss averse agent at low wealth levels ($0<W^*_t\le x_0 H^*_t$) is incentivized to strategically invest more aggressively in the risky asset to accumulate wealth in order to reach the threshold $x_0$ and start the positive consumption as soon as possible. The robustness of the constant proportion with respect to utility in the second observation is a consequence of the fact that the S-shaped utility forces the optimal consumption $C_t^*=0$ when the current wealth level is low and hence the specific form of utility function (particularly, $U_{-}(x)$) becomes irrelevant as only $U(0)$ matters in the objective function.
	
	\item If her wealth-to-habit ratio is above the austerity threshold (i.e. $X^*_t=W^*_t/H^*_t\ge x_0$), the agent will optimally consume at a rate above her loss reference point $\al H^*_t$ since, by \cref{c_optimal_moderate},
	\begin{align*}
		C^*_t = \cs(X^*_t) H^*_t = \Big[\al+(U_+')^{-1}\big(\varphi(v'(X^*_t))\big)\Big] H^*_t\ge \al H^*_t.
	\end{align*}
	Furthermore, our numerical investigation in the next section indicates that the agent adjust her consumption and investment behavior to a more moderate regime as the wealth-to-habit ratio increases. In particular, the portfolio weight of the risky asset and the ratio of the net consumption rate to wealth are both adjusted from higher to lower levels as wealth-to-habit ratio increases. See the bottom two plots of Figure \ref{fig:CRRA} in the next section.
\end{itemize}

In summary, the agent's optimal policy is to invest a more aggressive constant portion of her wealth and consume nothing in the austerity region (i.e. $W^*_t< x_0 H^*_t$). In the prosperity region (i.e. $W^*_t\ge x_0 H^*_t$), the agent consumes above her habit reference point and invest less aggressively in the stock. As the ratio of wealth to habit increases, the agent consumes and invests more moderately, in that the portfolio rate of the stock and the consumption-to-wealth ratio are smaller compare to their levels near the austerity region. A more detailed investigation of the optimal policies and their characteristics is provided in the next section.

\begin{remark}
	In the infinite-horizon Merton consumption problem \cite{Merton69}, the objective
	\[
	\mathbb{E}\left[\int_0^{+\infty}e^{-\del\,t} \frac{C_t^p}{p} d t\right]
	\]may become infinite when $p\in(0,1)$ and the discount rate $\delta>0$ is too small. To ensure finiteness of the value function, one must impose a lower bound on $\del$ which prevents unbounded expected utility over admissible consumption policies.\\
	In our setting, no such additional restriction on $\delta$ is required because preferences are defined over \emph{relative} consumption. The objective depends on the ratio $C_t/H_t$, where $H_t$ is the exponentially weighted average of past consumption. Increasing current consumption therefore raises future habit, reducing the future value of the ratio $C_t/H_t$. This endogenous feedback prevents the agent from escalating consumption in a way that would otherwise lead to an unbounded value function. In other words, the strictly positive habit process $\{H_t\}_{t\ge0}$ acts as an \emph{implicit discounting mechanism}: higher consumption today forces an increase in the future benchmark, thereby keeping the value function finite without requiring a lower bound on~$\delta$.
\end{remark}

%
%
\section{Illustrative Examples and Financial Implications}\label{sec:Examples}
This section illustrates our theoretical results in Section \ref{sec:OptimalPolicies} using several numerical examples. In Subsection \ref{subsec:power}, we provide a thorough numerical experiment for the case of a power S-shaped utility function of the form \cref{eq:U_power}. Subsection \ref{subsec:sens_RALA} investigates sensitivity of the optimal policies to changes in the risk-aversion and loss-aversion parameters of the power S-shaped utility, while sensitivity to changes in the market excess risk return, habit persistent, and utility discounting is explored in Subsection \ref{subsec:Sens_mu_rho_del}. In subsection \ref{subsec:limits}, we show that several other models are limiting case of our model by reproducing their numerical experiments with our model. Finally, in subsection \ref{subsec:non-power}, we provide numerical experiments for some non-power S-shaped utility functions, namely, an exponential S-shaped utility function and a SAHARA S-shaped utility function.

The numerical experiment in this section are based on standard and commonly available numerical algorithms. Specifically, to solve the free-boundary problem \cref{eq:ODESys-moderate} and \cref{eq:FB-moderate}, we use the explicit Runge-Kutta method of order 5(4) (provided by the Python package `scipy.integrate' with option `method="RK45"') and a bisection search (to find the free boundary $y_0$ in Theorem \ref{thm:DualFBP-moderate}.$(i)$). The value function $V(x)$ and the optimal feedback relative policies $\pi^*(x)$ and $c^*(x)$ are then obtain using Corollary \ref{coro:solveHJB} and Theorem \ref{thm:verification_moderate}.
%
%
\subsection{Optimal policies for the power S-shaped utility functions}\label{subsec:power}
In this subsection, we assume that $U(c)$ in \cref{eq:U} is an S-shaped power utility, namely,
\begin{align}\label{eq:U_power}
	U(c) = 
	\begin{cases}
		\frac{1}{p}\big[(c-\al+\eps_1)^p-\eps_1^p\big],\quad c>\al,\\[1ex]
		-\frac{\kap}{q}\big[(\al-c+\eps_2)^q - \eps_2^q\big],\quad 0\le c\le \al,
	\end{cases}
\end{align}
in which we have set the following default parameter values
\begin{align}\label{eq:default_param_utility}
	\boxed{
		\al=0.5,\, p=-1.0,\, q=-0.5, \eps_1=\eps_2=0.05,\text{ and } \kap=1.
	}
\end{align}
Throughout this subsection, we have considered the following default values for the other model parameters:
\begin{align}\label{eq:default_param_market}
	\boxed{
		r=0.04,\, \mu=0.05,\, \sig=0.2,\, \rho=0.7,\text{ and } \del=0.1.
	}
\end{align}
Note that the interest rate is $r=4\%$ and the expected excess stock return is $\mu=5\%$. Thus, the stock expected return is $\mu+r=9\%$.

Figure \ref{fig:CRRA} shows the value function $V(x)$ in \cref{eq:VF} for the power S-shaped utility and the default parameter values. The top leftmost plot shows the utility function $U(c)$ (the solid line) and its concave envelope $\Ut(c)$ (the dotted line) given by \cref{eq:Uc}. See the caption of Figure \ref{fig:U_smooth_nonsmooth} (in Section \ref{sec:model}) for further details in this plot. The next plot (that is, top right plot) shows the value function $V(x)$ and the free-boundary $x_0$, both given by Corollary \ref{coro:solveHJB}. Note that the value function $V(x)$ is lower bounded, with the lower bound $V(x)\ge V(0)=U(0)/\del$ (see \cref{eq:VF_lowerbound} below).

The feedback optimal relative investment policy $\pi^*(x)$ and the feedback optimal relative consumption policy $c^*(x)$ are illustrated in the middle two plots of Figure \ref{fig:CRRA}. Recall that $\pi^*(x)$ (respectively, $c^*(x)$) is the optimal ratio of the amount invested in the stock (respectively, the consumption rate) divided by the current habit, assuming that $x$ is the ratio of current wealth divided by current habit.  If $x=w/h\in(0,x_0)$, then $c^*=C^*=0$ and $\pi^*/x = \Pi^*/w =\mu(1-\lam)/\sig^2$, according to \cref{c_optimal_moderate} and \cref{pi_optimal_moderate}. That is, if the wealth-to-habit ratio $x=w/h$ is below the threshold $x_0$, the agent takes austerity measures by not consuming while investing a fixed proportion of her wealth in the risky asset. If, on the other hand, $x\ge x_0$, then the agent consumes at a rate above her habit reference point, that is $c^*(x)=C^*/h \ge\al$.

Interestingly, as $x=w/h$ becomes larger than the threshold $x_0$, the optimal risky investment \emph{decreases} to a minimum and then increases. This behavior (specifically, the initial decrease in risky investment) is an indication that the agent is trying to keep her wealth-to-habit ratio $x$ above threshold $x_0$. Specifically, if $x$ is sufficiently larger than the threshold $x_0$, the agent decreases her proportion of stock investment if $x$ decreases.

%
%

\begin{figure}[H]
	\centerline{
			\adjustbox{trim={0.0\width} {0.02\height} {0.0\width} {0.0\height},clip}
			{\includegraphics[scale=0.3, page=1]{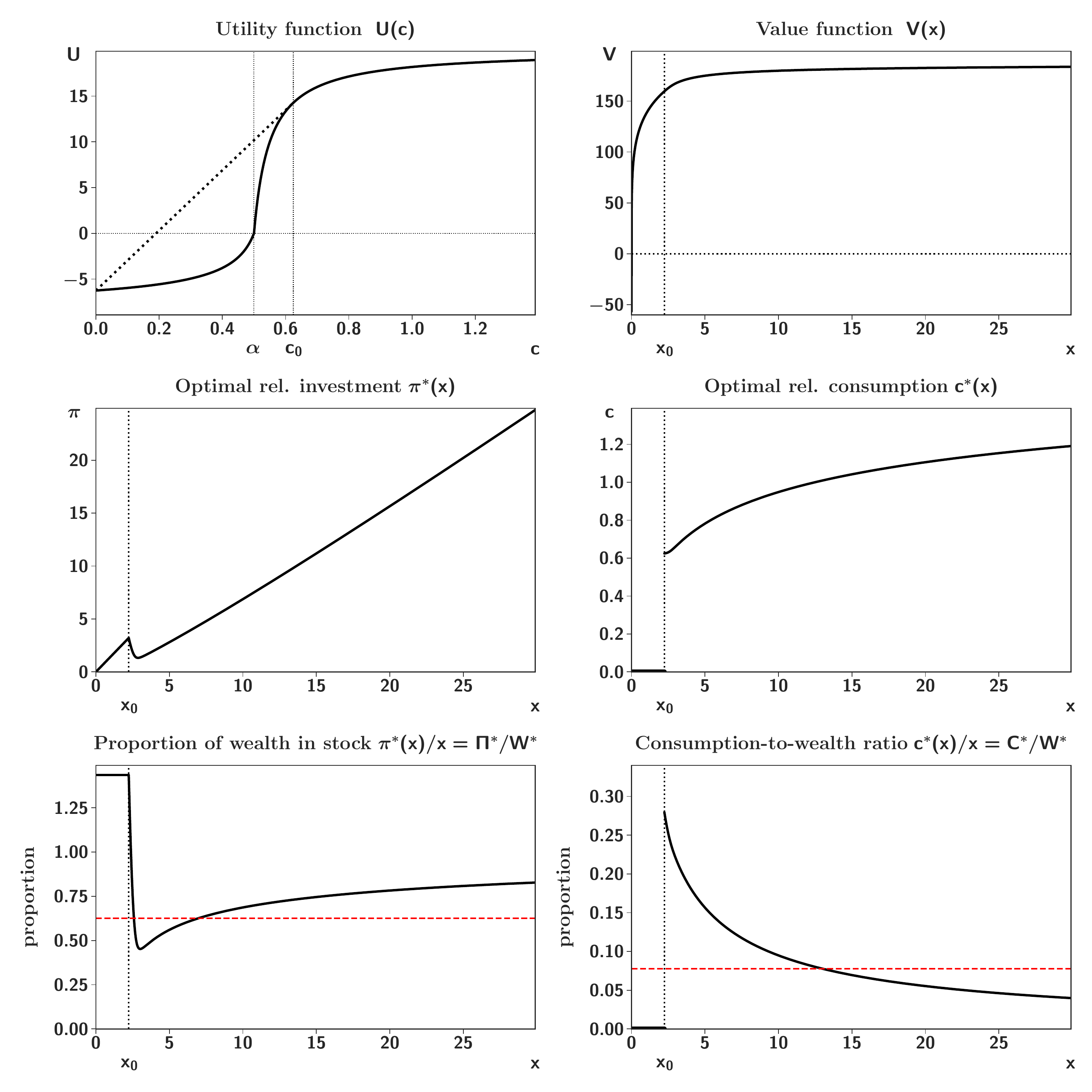}}
	}\vspace{1em}
	\caption{The value function and the optimal policies for the S-shaped power utility given by \cref{eq:U_power} and assuming the default value of the parameters in \cref{eq:default_param_utility} and \cref{eq:default_param_market}. The red dashed horizontal lines in the bottom plots represent the constant optimal portfolio weight $\frac{\mu}{(1-p)\sig^2}= 62.5\%$ and the constant optimal proportional consumption $\frac{1}{1-p}\left[\del - p \left(r +\frac{\mu^2}{2(1-p)\sig^2}\right)\right]\approx 7.8\%$ in the Merton \cite{Merton69} infinite horizon model.
		\label{fig:CRRA}}
\end{figure}

\noindent But, as $x\to x_0^+$, the agent increases her stock investment although her wealth to habit ratio is decreasing. In essence, the agent is placing a bet on the stocks to avoid further decrease of $x$, but only if $x$ is sufficiently close to the threshold $x_0$. See, also, the discussions for Figure \ref{fig:power_util_sens} (sensitivity with respect to $\kap$) and Figure \ref{fig:ABY22_Limit} for other explanations for this behavior of the optimal investment policy.

The two bottom plots of Figure \ref{fig:CRRA} provide an alternative way of illustrating the optimal investment and consumption policies. Note that, by \cref{eq:ratios}, we have $\Pi_t/W_t=\pi_t/X_t$ and $C_t/W_t=c_t/X_t$. Therefore, the ratio $\pi^*(x)/x$ is the optimal portfolio weight of the risky asset (that is, the proportion of wealth invested in the risky asset) if the wealth-to-habit ratio is $x$. Similarly, $c^*(x)/x$ is the optimal ratio of the (net) consumption rate to wealth for a wealth-to-habit ratio of $x$. The two bottom plots of Figure \ref{fig:CRRA} show the ratios $\pi^*(x)/x$ and $c^*(x)/x$ against $x$. These ratios enable us to compare our optimal policies with those in other studies. For instance, in the infinite-horizon optimal consumption problem of \cite{Merton69}, the optimal portfolio weights and the ratio of the consumption rate to wealth are constants given by \cref{eq:Merton_pi_c} below. These constants ratios are represented by the red horizontal dashed lines in the bottom two plots of  Figure \ref{fig:CRRA}.\footnote{Here, we assume a power utility of consumption in the Merton problem with constant relative risk-aversion $1-p=2$, subjective discount rate $\del=0.1$, excess risky mean return $\mu=0.05$, and volatility $\sig=0.2$.} As the bottom left plot indicates, the agent's stock investment is significantly higher than the Merton portfolio weight in the austerity region (i.e when $x<x_0$). As $x$ increases above $x_0$, the agent's optimal portfolio weight has a sharp decrease below Merton's weight, and then gradually increases above it. As the bottom right plot indicates, the agent takes extreme austerity measure by not consuming when $x<x_0$. As $x$ becomes larger than $x_0$, the agent first consumes at a rate above Merton's policy. But, once her wealth-to-habit ratio becomes sufficiently large, she consumes at a rate significantly lower than Merton's policy. These findings are expected and consistent with other studies involving consumption habit formation and loss aversion. In particular, habit formation makes the agent more reluctant to consume at a higher rate, since doing so would increase her habit and reduce future utility of consumption. Loss aversion, on the other hand, makes the agent wary of consumption below their loss threshold. Thus, the agent avoids consumption and builds up her wealth so that, in the future, she is able to consume above her loss threshold.

In the next two subsections, we provide sensitivity of the optimal policies to changes in several model parameters.

\subsection{Sensitivity to risk-aversion and loss-aversion parameters}\label{subsec:sens_RALA}
In this subsection, we investigate how the optimal policies, corresponding to the power S-shaped utility of the previous subsection, are affected if we change the risk-aversion parameter $p$, the loss-aversion parameter $\kap$, or the loss reference level $\al$. When investigating sensitivity of our results with respect to a certain parameter, we change the value of that parameter while keeping the remaining parameters at their default values in \cref{eq:default_param_utility} and \cref{eq:default_param_market}.

The top row of Figure \ref{fig:power_util_sens} illustrates sensitivity with respect to the risk tolerance parameter $p$. Note that $1-p$ is the (constant) relative risk aversion of the gain power utility in \cref{eq:U_power}. So, a higher value of $p$ means that the agent is less risk-averse (i.e. more risk tolerant). The plots show the default value of $p=-1$ by solid black curves, a more risk averse value $p=-2$ by the dashed blue curves, and the more risk-tolerant value of $p=0.0$ by the dash-dotted green curve. 
When $x$ is sufficiently larger than $x_0$ (that is, when the agent is sufficiently away from the austerity region), more risk averse agents (with lower value of $p$) generally invest less in the risky asset and consume less, both of which are reasonable and expected. The plots also show that more risk averse agents take stricter austerity measure by having a higher austerity threshold $x_0$.

The middle row of Figure \ref{fig:power_util_sens} shows sensitivity of the optimal policies with respect to the loss-aversion parameter $\kap$. By \cref{eq:U_power},  a larger value of $\kap$ indicates that the agent has a stronger aversion to consuming below the reference habit level. 
The solid black curves correspond to the default value of $\kap=1$, the dashed blue curves correspond to a slightly more loss averse agent with $\kap=2$, and the dash-dotted green curves correspond to an extremely more loss averse agent with $\kap=100$. 
The plots indicate that increasing $\kap$ has a mixed effect. On one hand, it decreases the austerity threshold $x_0$. On the other hand, the drop in the stock investment (due to austerity measure) is more pronounced for higher values of $\kap$, indicating that the agent is taking a more drastic austerity measure.

%
%

\begin{figure}[H]
	\centerline{
			\adjustbox{trim={0.0\width} {0.0\height} {0.0\width} {0.5\height},clip}
			{\includegraphics[scale=0.32, page=1]{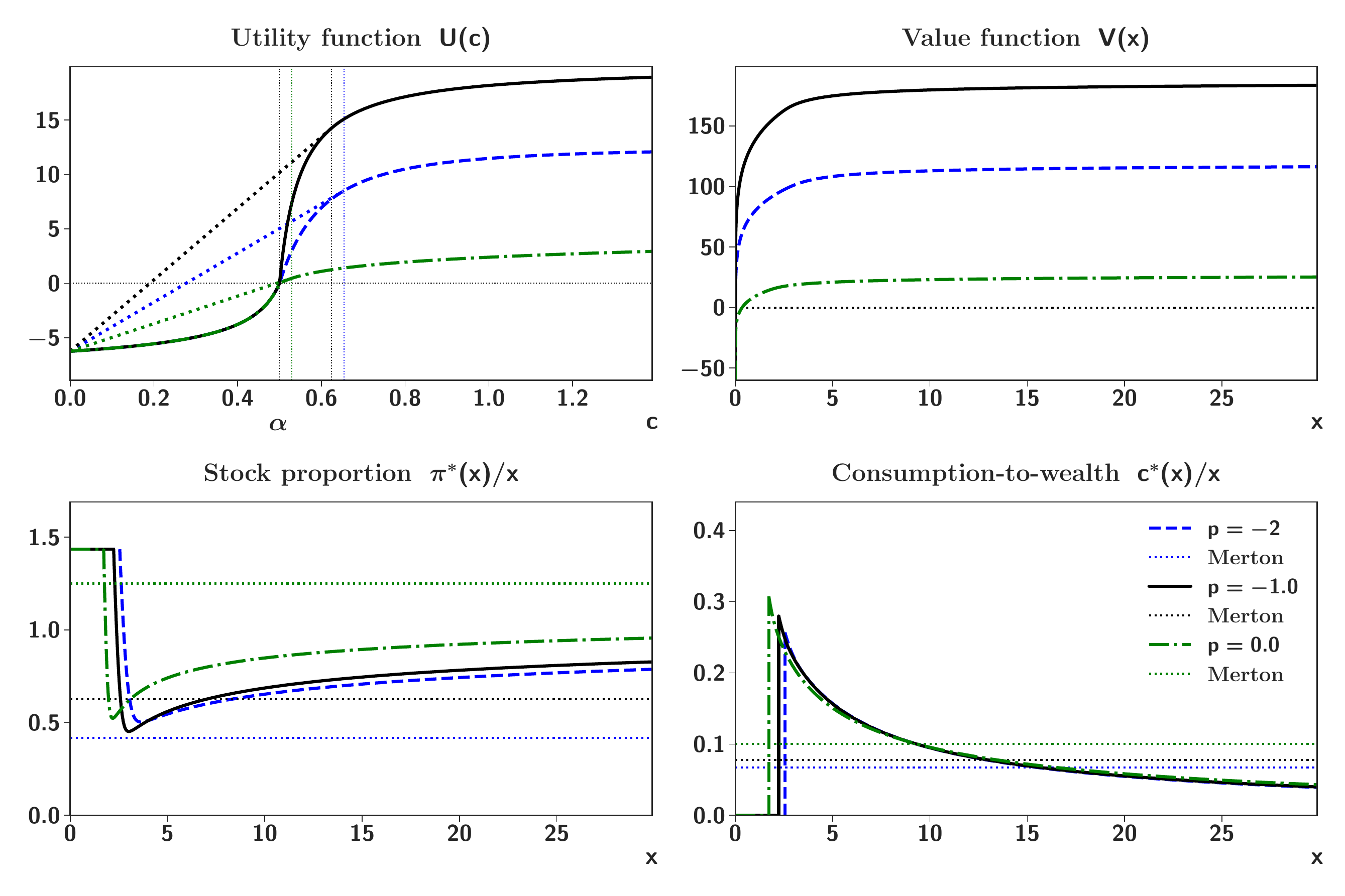}}
	}\vspace{1em}
	\centerline{
			\adjustbox{trim={0.0\width} {0.0\height} {0.0\width} {0.5\height},clip}
			{\includegraphics[scale=0.32, page=1]{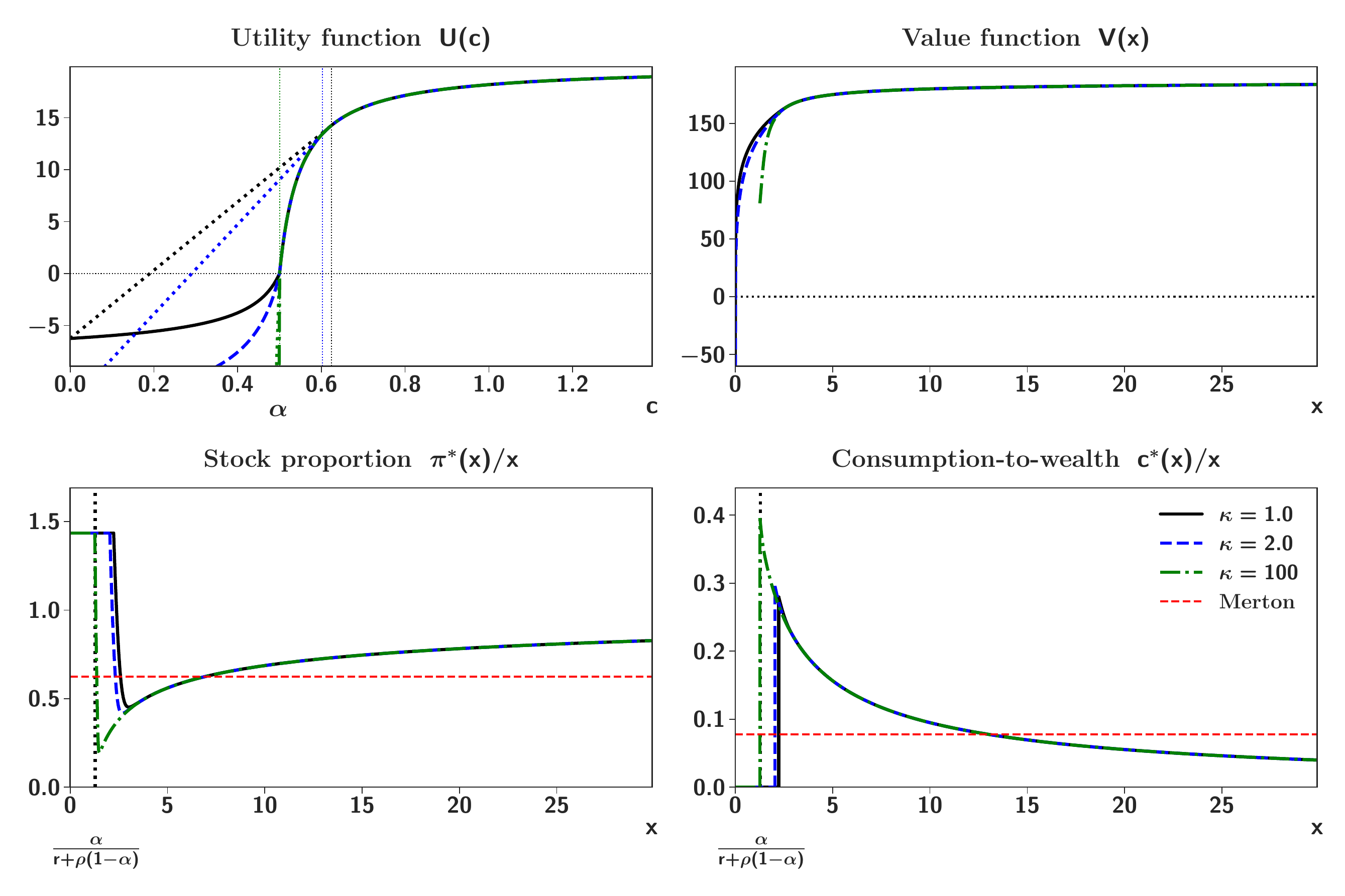}}
	}\vspace{1em}
	\centerline{
			\adjustbox{trim={0.0\width} {0.0\height} {0.0\width} {0.5\height},clip}
			{\includegraphics[scale=0.32, page=1]{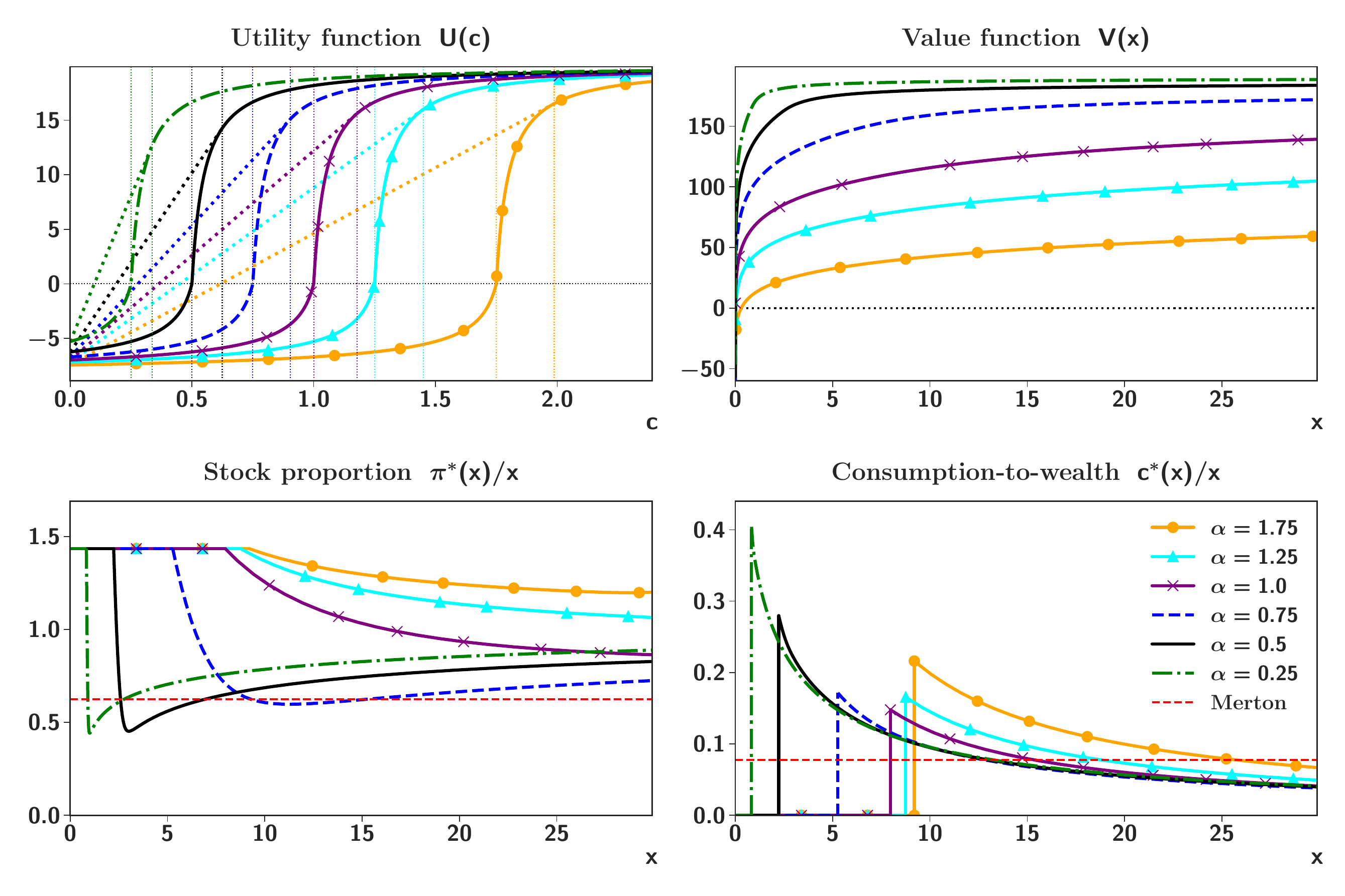}}
	}
	\caption{ Sensitivity of the optimal relative policies with respect to the parameters of the S-shaped power utility, namely, the risk tolerance parameter $p$ (the top row),  the loss aversion parameter $\kap$ (the middle row), and the loss reference $\al$ (the bottom row).
		\label{fig:power_util_sens}}
\end{figure}

The plot for the optimal investment (i.e. the middle left plot of Figure \ref{fig:power_util_sens}) indicates an interesting phenomenon in the case of extreme loss-aversion (i.e. the dash-dotted green curve). The portfolio weight of the risky asset approaches zero at the threshold $x=\al/\big(r+\rho(1-\al)\big)$. This behavior is explained as follows. For very larger values of $\kap$, the agent becomes extremely loss averse and would avoid consuming at a rate $C_t<\al H_t$. Thus, the agent (almost) adapts the constraint $C_t\ge \al H_t$ for all $t\ge0$. This is the habit-formation constraint in \cite{ABY22}. By Lemma 2.2 therein, under this consumption constraint and to avoid bankruptcy, the wealth-to-habit ratio must satisfy the no-bankruptcy constraint $X_t=W_t/H_t\ge\al/\big(r+\rho(1-\al)\big)$. To enforce this constraint, any admissible policy must invest fully in the riskless asset (i.e $\Pi_t=0$) whenever $X_t=\al/\big(r+\rho(1-\al)\big)$. For the case $\kap=100$ (i.e. the dash-dotted green curve), the agent almost adapts the habit-formation constraint $C_t\ge \al H_t$, $t\ge0$. As a result, she also almost enforces the the no-bankruptcy constraint $X_t=W_t/H_t\ge\xu$, $t\ge0$, by almost fully investing in the riskless asset at $x=\al/\big(r+\rho(1-\al)\big)$. In the Subsection \ref{subsec:limits}, we investigate how our model approximates the results of \cite{ABY22} under the habit-formation constraint. See Figure \ref{fig:ABY22_Limit} and its discussion.

The bottom row of Figure \ref{fig:power_util_sens} illustrates the effect of changing the habit reference point $\al$ on the optimal policies. By \cref{eq:U_power}, increasing $\al$ shifts the utility function to the right. 
That is, increasing $\al$ makes the agent more loss averse by increasing her loss region. The plots display three cases with $\alpha<1$ (black solid: $\alpha=0.5$; blue dashed: $\alpha=0.75$; green dash--dotted: $\alpha=0.25$) and three cases with $\alpha \ge 1$ (purple $\times$: $\alpha=1$; cyan $\blacktriangle$: $\alpha=1.25$; orange $\bullet$: $\alpha=1.75$). As expected, the austerity threshold $x_0$ increases with $\alpha$. Moreover, $x_0$ is substantially more sensitive to changes in $\alpha$ than to variations in risk tolerance $p$ or loss aversion $\kappa$.

Regarding the intensity of austerity behavior, two clear patterns emerge:
\begin{itemize}
	\item \textbf{When $\alpha<1$:}  
	The less loss-averse agent (green) begins consuming at a higher rate upon exiting austerity, but her consumption soon aligns with that of more loss-averse agents (black and blue). For large $x$, she invests more in the risky asset, whereas for small $x$ she invests less, reflecting the tighter austerity imposed on more loss-averse agents.
	\item \textbf{When $\alpha \ge 1$:}  
	The more loss-averse agents (cyan and orange) consume at higher rates outside the austerity region and invest more aggressively for all wealth levels. The stock proportion $\pi^*(x)/x$ is decreasing for $x>x_0$. These patterns indicate stricter austerity and more pronounced consumption deferral under high aspiration levels.
\end{itemize}

\subsection{Sensitivity to the market return, habit persistent, and utility discounting}\label{subsec:Sens_mu_rho_del}
In Figure \ref{fig:power_market_HF_sens}, we investigate the effect of changing other model parameters, namely, the market excess expected return $\mu$, the habit formation persistence $\rho$, and the subjective discount rate $\del$. 
As in the previous subsection, we only change one parameter while keeping the remaining parameters at their default values in \cref{eq:default_param_utility} and \cref{eq:default_param_market}.

The top row of Figure \ref{fig:power_market_HF_sens} shows sensitivity of 
the optimal policies with respect to the excess expected return $\mu$. The solid black curves represent the default value of $\mu=5\%$, the dashed blue curves represent the case of a more profitable risky investment with $\mu=10\%$, and the dash-dotted green curves represent a less profitable risky investment with $\mu=2\%$.\footnote{Recall that the expected stock return is $\mu+r$ in our model, see \cref{eq:S-SDE}. Since $r=4\%$, these values correspond to expected stock returns of $9\%$, $14\%$, and $6\%$, respectively.} 
The top plot of Figure \ref{fig:power_market_HF_sens} indicate that the agent invests more in the risky asset and consumes more when the risky asset is more profitable. The austerity threshold $x_0$ is decreasing in $\mu$, meaning that the loss averse agent requires a lower wealth threshold when the risky asset is more profitable, which is reasonable. 

The plots in the middle row of Figure \ref{fig:power_market_HF_sens} illustrate sensitivity of 
the optimal policies with respect to the habit formation persistence parameter $\rho$. As pointed out after \cref{eq:H}, a larger value of $\rho$ makes the agent's habit more sensitive to her current consumption, while a smaller value makes the habit process more persistent by assigning higher weights to past consumption rates. In the middle row plots of Figure \ref{fig:power_market_HF_sens}, the solid black curves represent the default value of $\rho=0.7$, the dashed blue curves represents the case of a more persistent habit process with $\rho=0.2$, and the dash-dotted green line represents a more transient (i.e. more sensitive to current consumption) habit process with $\rho=1.2$. 
For smaller values of $\rho$, it is harder for the agent to change her habit process. Thus, the effect of loss aversion and habit formation become more pronounced. That is, for smaller values of $\rho$, the

%
%

\begin{figure}[H]
	\centerline{
			\adjustbox{trim={0.0\width} {0.0\height} {0.0\width} {0.5\height},clip}
			{\includegraphics[scale=0.32, page=1]{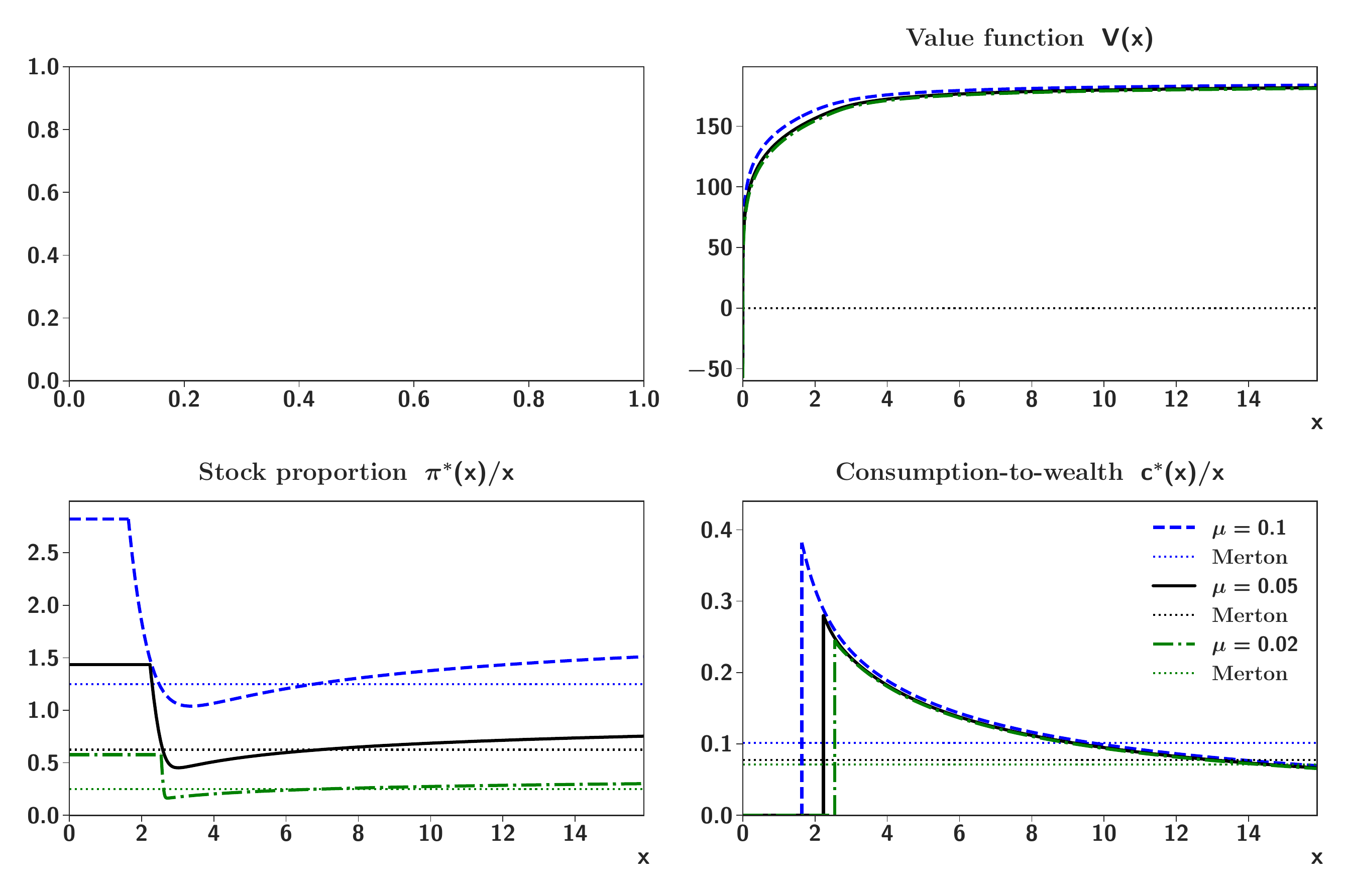}}
	}\vspace{1em}
	\centerline{
			\adjustbox{trim={0.0\width} {0.0\height} {0.0\width} {0.5\height},clip}
			{\includegraphics[scale=0.32, page=1]{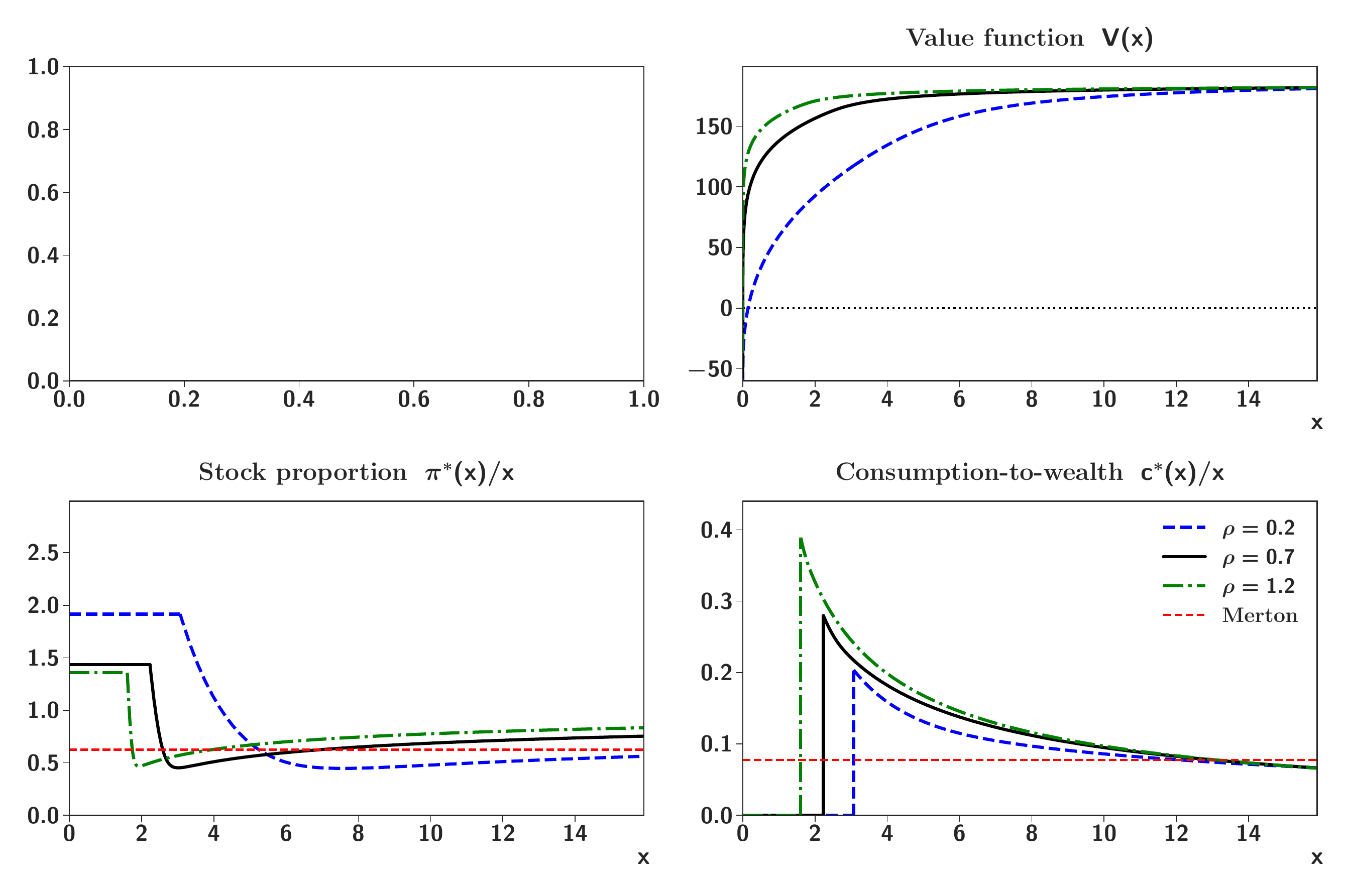}}
	}\vspace{1em}
	\centerline{
			\adjustbox{trim={0.0\width} {0.0\height} {0.0\width} {0.5\height},clip}
			{\includegraphics[scale=0.32, page=1]{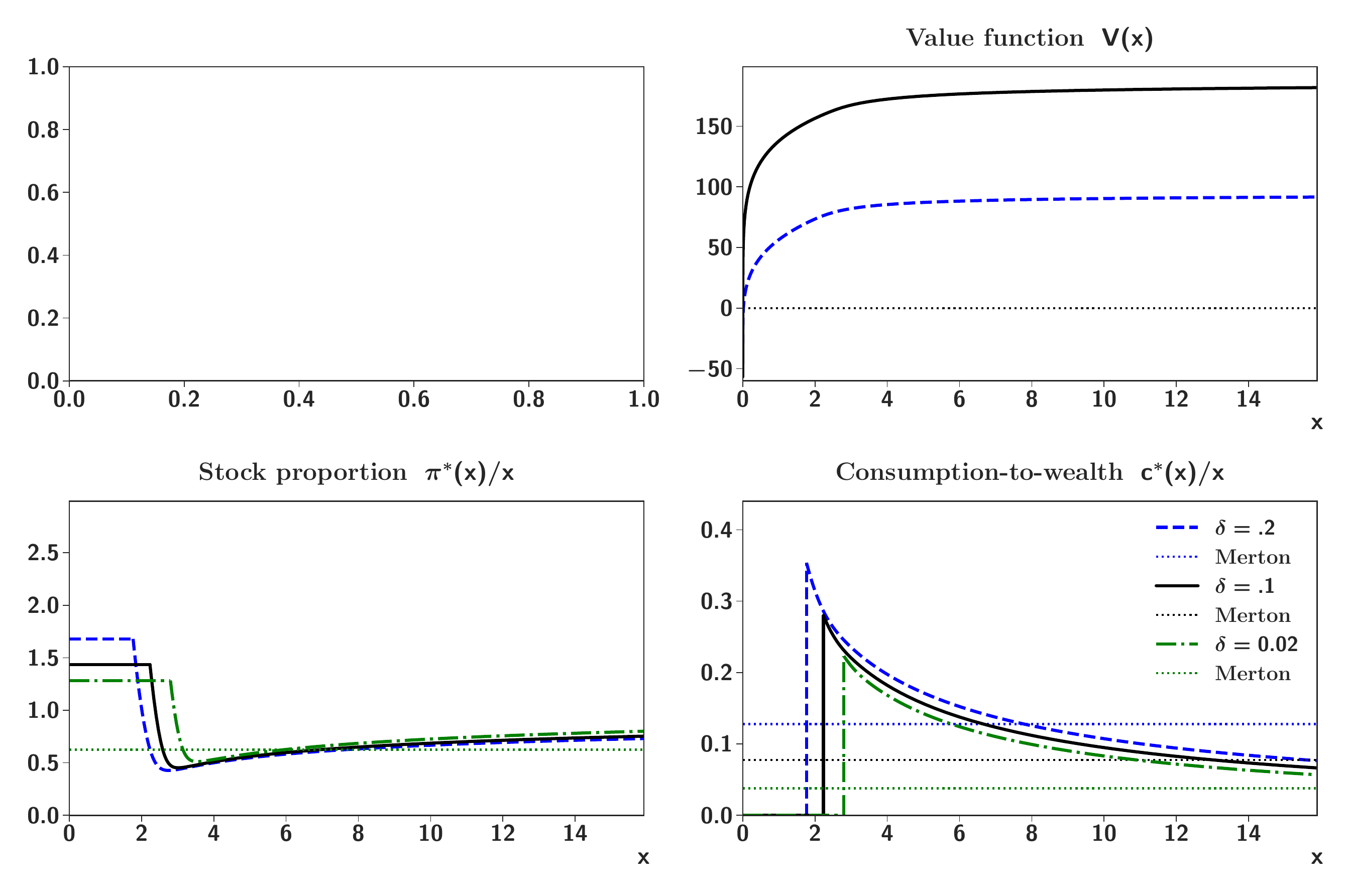}}
	}
	\caption{Sensitivity of the value function and optimal relative policies with respect to the excess return of the risky asset $\mu$ (the top row), the habit formation persistence rate $\rho$ (the middle row), and  the subjective utility discount rate $\del$ (the bottom row).
		\label{fig:power_market_HF_sens}}
\end{figure}

\noindent   agent takes stricter austerity measure in the loss region (i.e. when $x$ is near $x_0$). For smaller values of $\rho$, she also reduces her consumption and stock investment for sufficiently large values of $x$ in the gain region (i.e. for sufficiently large $x$ larger than $x_0$). 

Finally, the plots in the bottom row of Figure \ref{fig:power_market_HF_sens} show the effect of changing the subjective discount rate $\del$. By \cref{eq:VF}, a larger value of $\del$ indicates that the agent is more impatient in that she prefers consuming earlier rather than later. In the plots, the solid black curves represent the default value of $\del=0.1$, the dashed blue curves represents the case of a more impatient agent with larger $\del=0.2$, and the dash-dotted green curves represent a more patient agent with smaller $\del=0.02$. 
The bottom left plot of Figure \ref{fig:power_market_HF_sens} indicates that, during austerity (i.e. when $x$ is small), the more impatient agent invests more in the risky asset and have a smaller austerity threshold $x_0$, which is reasonable. As $x$ becomes larger, the more impatient agent more quickly adjust her investment pattern from the austerity measure, and will eventually (i.e. as $x$ gets larger) invest more-or-less the same amount as a more patient agent. As the bottom right plot of Figure \ref{fig:power_market_HF_sens} indicates, the more impatient agent consumes more than a more patient agent. These are also expected behavior.

%
%
\subsection{Connections with existing literature}\label{subsec:limits}

The goal of this subsection is to reproduce numerical experiments in other papers using our model. 
Doing so serves two purposes. Firstly, it provides an alternative way of validating our results. Secondly, it showcases the flexibility of our model as it can encompass several other studies as limiting cases. 

We provide connection between our model and three other infinite horizon optimal consumption and investment models, namely,
\begin{enumerate}
	\item[(a)] The classical infinite horizon model in \cite{Merton69} under a power utility.
	
	
	\item[(b)] The multiplicative habit formation model with (strictly concave) power utility in \cite{Rogers2013}.
	
	\item[(c)] The optimal policies in \cite{ABY22} under habit formation constraint. 
\end{enumerate}
These models are limiting cases of our model, in that they are obtained by letting specific parameters approach certain values in our model.

Since our goal is to reproduce the numerical experiments of other papers, we have to use different model parameters than the one used in Subsection \ref{subsec:power}.

We start with the classical infinite horizon optimal investment and consumption problem in \cite{Merton69}, namely,
\begin{align}\label{eq:Merton_VF}
	V_M(w) = \sup_{(\Pi, C)\in\Ac_0(w)} \Eb\left[\int_0^{+\infty}\ee^{-\del t} \frac{C_t^p}{p} \dd t\right],\quad w>0.
\end{align}
In this model, the value function is
\begin{align*}
	V_M(w) = \gam_M^{p-1} \frac{w^p}{p},\quad w>0,
\end{align*}
and the feedback optimal policies are $\big\{\Pi_M(W^*_t)\big\}_{t\ge0}$ and $\big\{C_M(W^*_t)\big\}_{t\ge0}$, in which
\begin{align}\label{eq:Merton_pi_c}
	\Pi_M(w):= \frac{\mu}{(1-p)\sig^2} w,\quad \text{and}\quad C_M(w):= \gam_M w,
\end{align}
with the constant $\gam_M$ given by
\begin{align*}
	\gam_M := \frac{1}{1-p}\left[\del - p \left(r +\frac{\mu^2}{2(1-p)\sig^2}\right)\right].
\end{align*}
Consider our model with the power S-shaped utility \cref{eq:U_power}. As $\al,\eps_1\to0^+$ for the utility function $U(c)$ in \cref{eq:U_power}, the function $U(c)+\eps_1^p/p$ becomes the power utility function in \cref{eq:Merton_VF}. Furthermore, by \cref{eq:H}, the habit process becomes constant (i.e. $H_t\to h$) as $\rho\to0^+$. It then follows that for small values of $\rho$, $\al$, and $\eps_1$, the optimal stock proportion $\pi^*(x)/x$ and the optimal proportional consumption $c^*(x)/x$ in our model should approach, respectively, the constants $\frac{\mu}{(1-p)\sig^2}$ and $\gam_M$ in Merton's model. This convergence is observed in Figure \ref{fig:Merton_Limit}, which shows Merton's policies vs. ours with $\al=\rho=\eps_1=0.01$.

%
%
\begin{figure}[htbp]
	\centerline{
			\adjustbox{trim={0.0\width} {0.025\height} {0.0\width} {0.5\height},clip}
			{\includegraphics[scale=0.32, page=1]{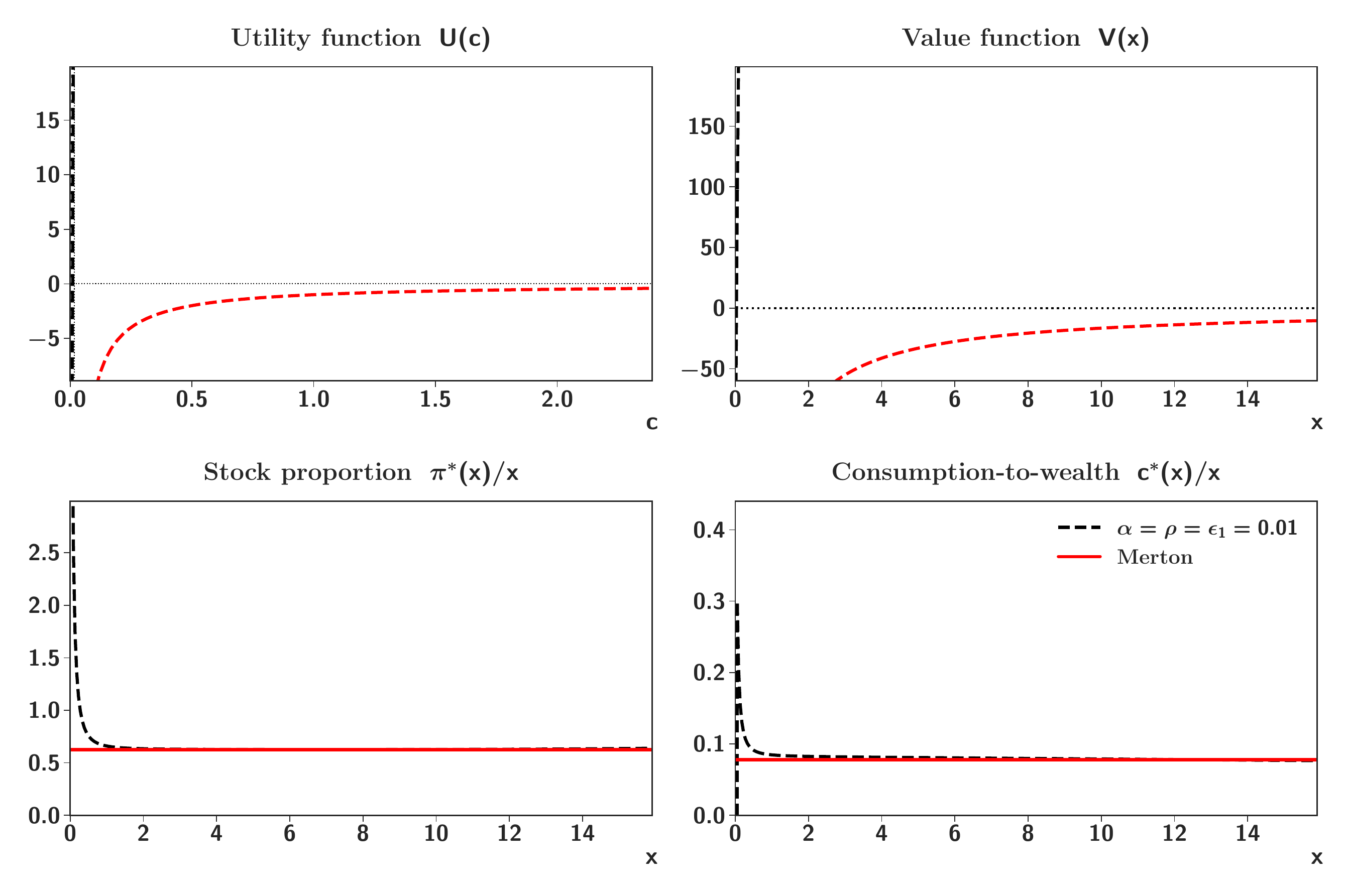}}
	}
	\caption{
		By letting $\al\to0^+$ and $\rho\to 0^+$, our model converges to the classical infinite horizon optimal consumption model of \cite{Merton69}.
		\vspace{1.5em}
		\label{fig:Merton_Limit}}
	\centerline{
			\adjustbox{trim={0.0\width} {0.025\height} {0.0\width} {0.5\height},clip}
			{\includegraphics[scale=0.32, page=1]{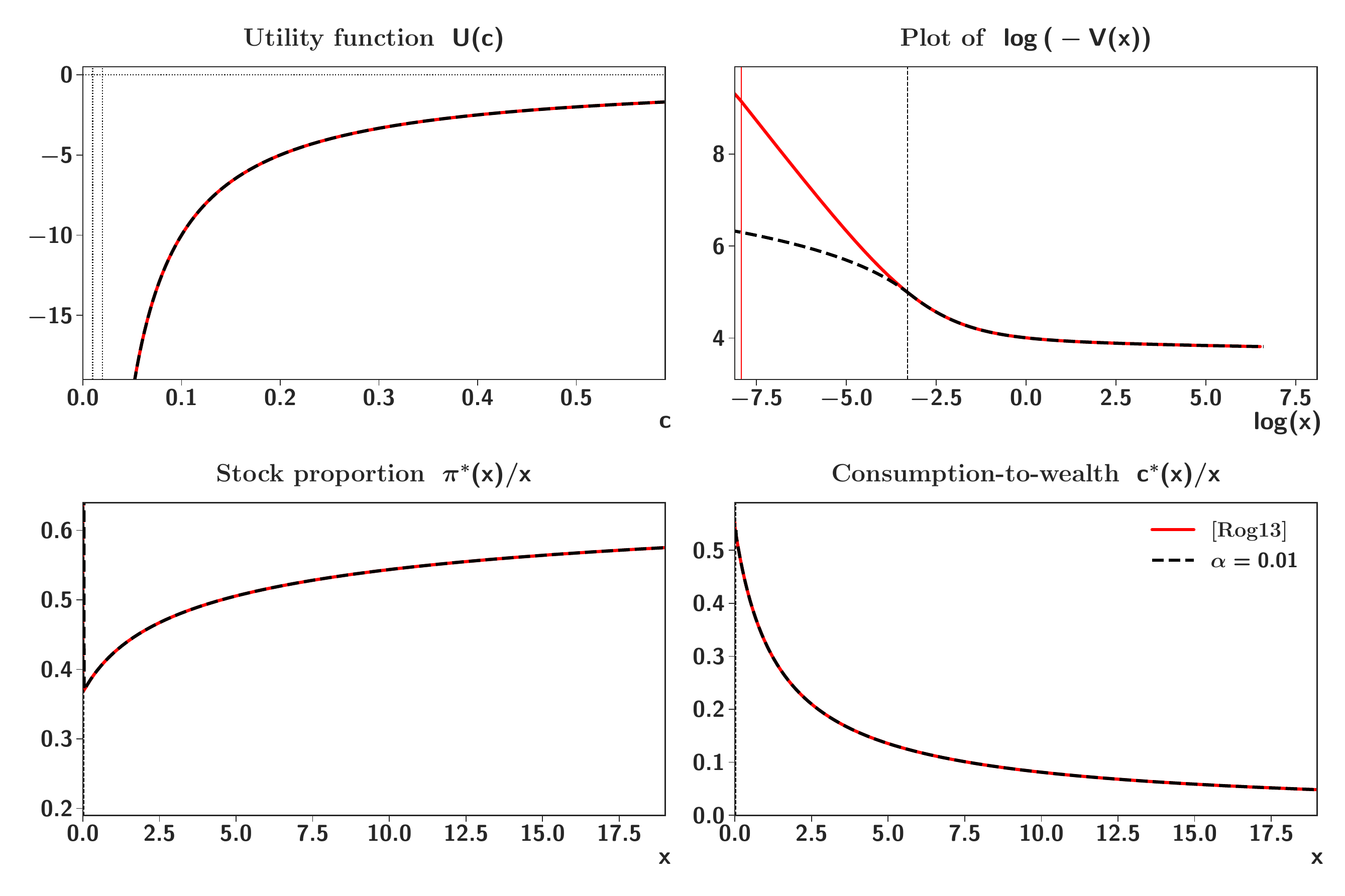}}
	}
	\caption{
		By letting $\alpha \to 0+$, our model converges to multiplicative habit formation proposed in Section 2.3 of \cite{Rogers2013}.
		\vspace{1.5em}
		\label{fig:MultipHF_limit}}
	\centerline{
			\adjustbox{trim={0.0\width} {0.025\height} {0.0\width} {0.5\height},clip}
			{\includegraphics[scale=0.32, page=1]{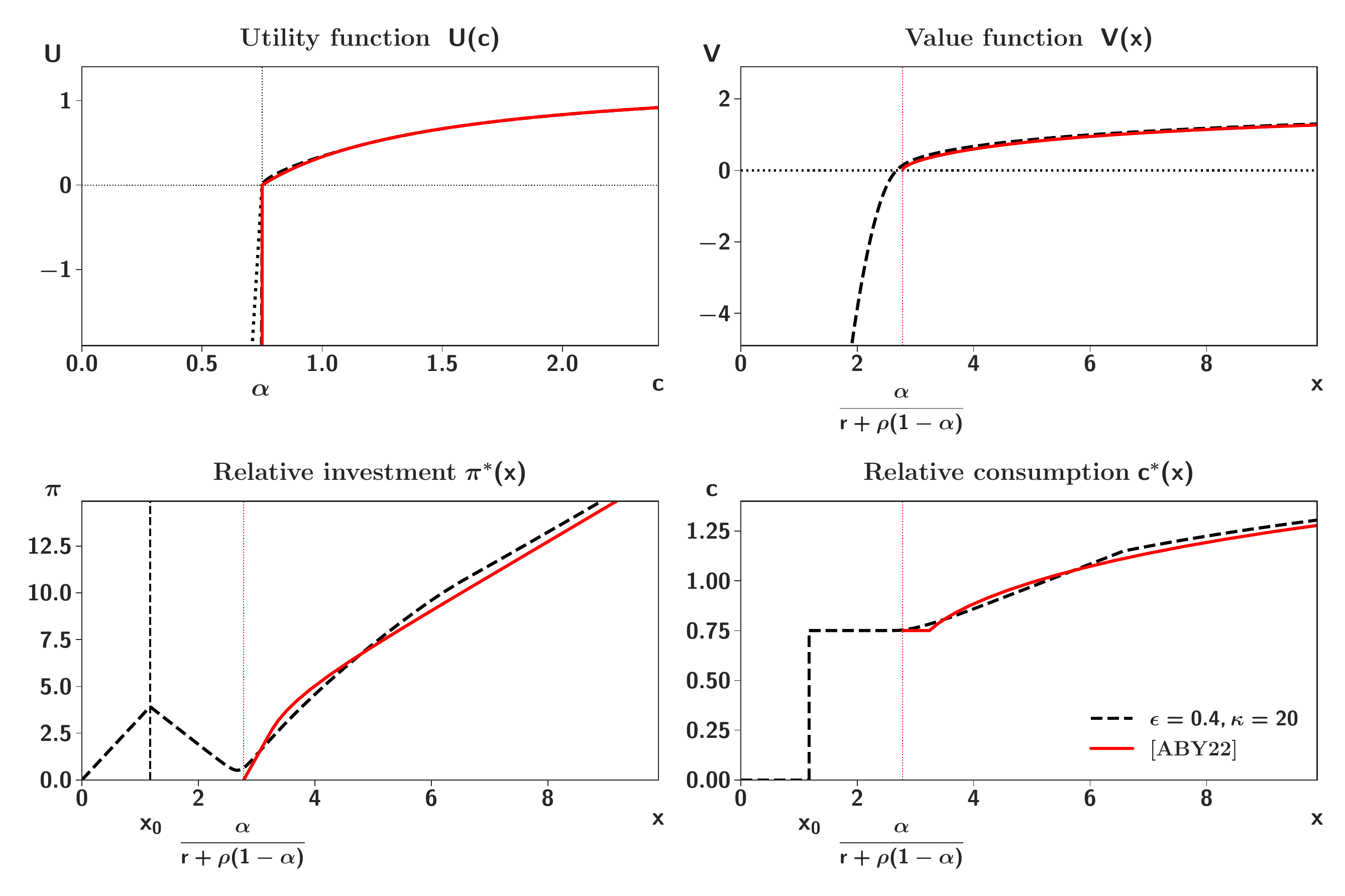}}
	}
	\caption{
		By letting $\kappa\to \infty$ and $\epsilon\to 0$, our model converges to habit formation constraint model in \cite{ABY22}.
		\label{fig:ABY22_Limit}}
\end{figure}


We now turn to a multiplicative habit formation model without loss-aversion, proposed in Section 2.3 of \cite{Rogers2013}. After a dimension reduction similar to that in \cref{eq:ratios}, they numerically solve the problem
\begin{align*}
	V(x) = \sup_{(\pi, c)\in\Acr(x)} 
	\Eb\left[\int_0^{+\infty}-\ee^{-\del t} c_t^{-1}\dd t\right].
\end{align*}
In \cite{Rogers2013}, they utilized two different numerical methods and achieved consistent results, as shown in Fig. 2.2 on page 36 of \cite{Rogers2013}. For the comparison, let us set $p=-1$, $\delta=0.02$, $\sigma=0.35$, $r=0.05$, and $\mu+r=0.14$. Figure \ref{fig:MultipHF_limit} replicates the plots in \cite{Rogers2013}, including the utility function, the logarithm of the minus value function, the optimal stock proportion, and the optimal consumption-to-wealth ratio. As Figure \ref{fig:MultipHF_limit} clearly illustrates, setting $\al=0.01$ in our model (with the above mentioned values for other parameters) yields feedback optimal stock proportion $\pi^*(x)/x$ and consumption-to-wealth ratio $c^*(x)/x$ that are almost indistinguishable with the curves shown in Fig. 2.2 on page 36 of \cite{Rogers2013}. 

In the rest of this subsection, we compare with the recent paper \cite{ABY22} with the habit formation constraint, which solves the following problem
\begin{align*}
	V(x) = \sup_{(\pi, c)\in\Acr(x)} \left\{
	\Eb\left[\int_0^{+\infty}\ee^{-\del t} \frac{c_t^p}{p}\dd t\right]
	\,:\,c_t\ge\al \text{ and } X_t\ge\xu \text{ for all } t\ge0
	\right\},\quad x\ge \xu.
\end{align*}
Here, $\xu :=  \frac{\al}{r+\rho(1-\al)}$ and $p<0$, and the habit formation constraint $C_t\ge \al H_t$ is enforced for all $t\ge0$. Under this constraint and to avoid bankruptcy, the wealth-to-habit ratio must be above the minimum bound $X_t:=W_t/H_t\ge \xu$, $t\ge0$; see Lemma 2.2 in \cite{ABY22}.

In their numerical example, they choose $p=-1$. Thus, their utility function is
\begin{align*}
	U_0(c) &= 
	\begin{cases}
		\displaystyle\frac{1}{\al}-\frac{1}{c};\quad c\ge \al\\
		-\infty;\quad 0<c<\al,
	\end{cases}
\end{align*}
in which we have shift their utility by the constant $1/\al$ to match our convention of $U(\al)=0$. To approximate the utility function $U(\cdot)$, we use the following three-piece S-shaped utility function parameterized by $\kap,\eps>0$,
\begin{align}\label{eq:ABY22_approx_utility}
	U_\eps(c) &:= 
	\begin{cases}
		\displaystyle\frac{1}{\al}-\frac{1}{c};\quad c\ge \al+\eps,\\[1ex]
		\displaystyle \frac{1}{\al(\al+\eps)}\eps^{\frac{\eps}{\al+\eps}} (c-\al)^\frac{\al}{\al+\eps};\quad \al\le c < \al+\eps,\\[1ex]
		\displaystyle2 \kap (\al-c)^{0.5};\quad 0<c<\al.\\
	\end{cases}
\end{align}
In particular, $U_\eps(\cdot)\to U_0(\cdot)$ as $\kap\to+\infty$ and $\eps\to 0$. Note also that, for any $\eps>0$, $U_\eps(\cdot)$ satisfies \cref{eq:U}, \cref{eq:Uc}, Assumption \ref{assumption:moderateLA}, and Assumption \ref{assum:Growth}.

Figure \ref{fig:ABY22_Limit} illustrates the optimal relative consumption and investment policies in Figure 3 \cite{ABY22} and the corresponding optimal policies in our model with values of $\epsilon=0.4$ and $\kappa=20$ in \cref{eq:ABY22_approx_utility}. To match the parameter values in \cite{ABY22}, we have also set $\delta=0.3$, $\al=0.75$, $\sigma=0.2$, $r=0.02$, and $\mu+r=0.12$. As the plots indicate, our policy closely approximate those in \cite{ABY22}. Finally, note that \cite{ABY22} did not study the case $0\le x\le \al/\big(r+\rho(1-\al)\big)$ in which bankruptcy is unavoidable. Our model, however, can approximate the optimal policy for the extension of \cite{ABY22} model to the case $x\in(0,\xu)$.

%
%
\subsection{Non-power S-shaped utilities}\label{subsec:non-power}

Our previous examples mainly focus on the commonly used power S-shaped utility functions and illustrate some optimal portfolio and consumption behavior induced by the loss-aversion and the habit formation. However, our theoretical characterization of the optimal relative policies and the associated free boundary problems in Theorem 
\ref{thm:verification_moderate} are applicable to general S-shaped utilities as long as Assumptions \ref{assumption:moderateLA} and \ref{assum:Growth} are satisfied. For instance, we may consider the following two examples of non-power S-shaped utility functions that appeared in the literature:

(i) The Exponential S-shaped utility, for $q\ge p>0$, $\kap\ge1$,
\begin{align*}
	U(c) = 
	\begin{cases}
		\displaystyle 
		1-\ee^{-p(c-\al)},\quad c>\al,\\[1em]
		\displaystyle 
		\kap\left(\ee^{-q(c-\al)}-1\right),\quad 0\le c\le \al.
	\end{cases}
\end{align*}

%
%

\begin{figure}[tbp]
	\centerline{
			\adjustbox{trim={0.0\width} {0.02\height} {0.0\width} {0.0\height},clip}
			{\includegraphics[scale=0.35, page=1]{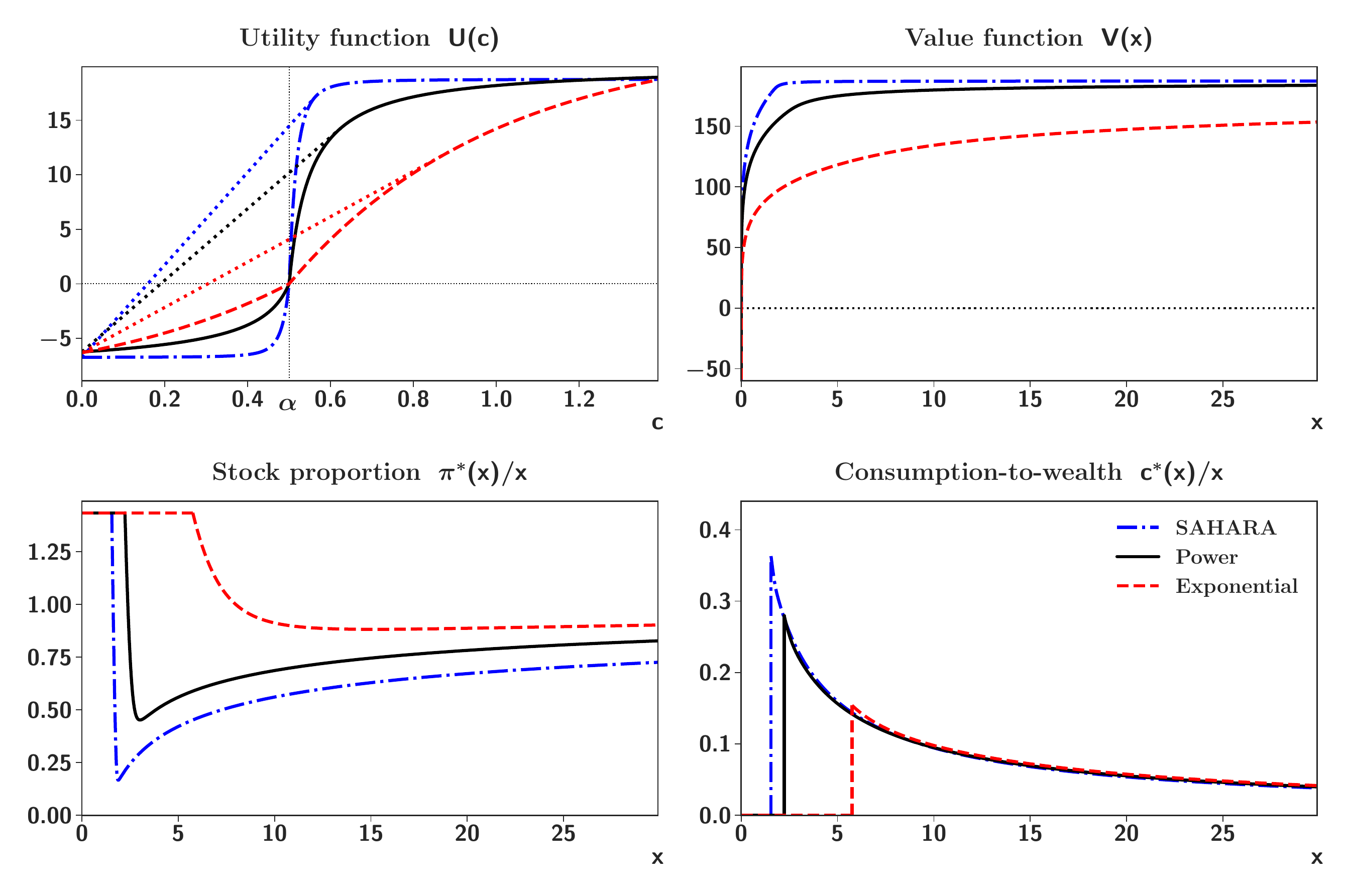}}
	}
	\caption{
		The value functions and the optimal feedback controls for different types of S-shaped utility functions.
		\label{fig:non_power}}
\end{figure}

(ii) The Symmetric Asymptotic Hyperbolic Absolute Risk Aversion (SAHARA) S-shaped utility 
\begin{align*}
	U(c) = 
	\begin{cases}
		U^{\gamma_1, \beta_1}_+(c-\al),\quad c>\al,\\
		-U^{\gamma_2, \beta_2}_-(\al-c),\quad 0\le c\le \al,
	\end{cases}
\end{align*}
where the utility function $U^{\gamma,\beta}(x)$, introduced by \cite{CPV2011}, is given by
\begin{align}\label{eq:SAHARA}
	U^{\gam,\beta}(x) :=
	\begin{cases}
		\frac{1}{1-\gam^2}\left(x+\gam\sqrt{x^2+\beta^2}\right)
		\left(x+\sqrt{x^2+\beta^2}\right)^{-\gam},
		&\quad \gam>0, \gam\ne1,\\[1ex]
		\frac{1}{2}\log\left(x+\sqrt{x^2+\beta^2}\right) + \frac{x}{2\left(x+\sqrt{x^2+\beta^2}\right)},
		&\quad \gam=1,
	\end{cases}
\end{align}
for $x\in\Rb$, in which $\gam>0$ is the risk-aversion parameter and $\beta>0$ is the scaling factor. They are characterized by their absolute risk-aversion function $-(U^{\gam,\beta})^{\prime\prime}(x)/(U^{\gam,\beta})'(x)=\gam/\sqrt{x^2+\beta^2}$, $x>0$. 
The power and logarithmic utility functions can be expressed as the limit of the SAHARA utility function:
\begin{itemize}
	\item For $\gam>0$, $\gam\ne1$, and $\beta\to 0^+$, we obtain the power utility function $U^{\gam,0^+}(x)= \frac{2^{-\gam} x^{1-\gam}}{1-\gam}$, $x>0$.
	\item For $\gam=1$ and $\beta\to 0^+$, we obtain the logarithmic utility function $U^{1,0^+}(x)= 0.5\log2x + 0.25$, $x>0$.
\end{itemize}

To better exemplify the generality of our theoretical findings, we also plot the value function, the optimal stock proportion $\pi^*(x)/x$ and the optimal consumption-to-wealth $c^*(x)/x$ in Figure \ref{fig:non_power} for the above two examples of S-shaped utility functions by employing the piecewise feedback functions in \cref{pi_optimal_moderate} and \cref{c_optimal_moderate} and the solutions $\varphi(y)$ and $\psi(y)$ of the associated free boundary problems in Theorem \ref{thm:DualFBP-moderate}.

%
%
\section{Technical Proofs}\label{sec:techproofs}
This sections includes the technical proofs of several results in the earlier sections.

\subsection{Solution to an auxiliary boundary value problem}\label{app:Aux}

As a preparation for the proof of Theorem \ref{thm:DualFBP-moderate}, this subsection first investigates an auxiliary system of ODEs with certain boundary value conditions. In particular, let us consider the boundary value problem
\begin{align}\label{eq:FBP_ybar}
	\begin{cases}\displaystyle
		\varphi'(y) = g_1\big(y,\varphi(y),\psi(y)\big),\quad y<\yb,\\
		\psi'(y) =  g_2\big(y,\varphi(y),\psi(y)\big),\quad y<\yb,\\
		\varphi(\bar{y}) = \phi_0,\quad\psi(\bar{y})= \Psi(\bar{y}),
	\end{cases}
\end{align}
for a given boundary $\bar{y}\in (\gamma\phi_0,\phi_0)$. 

Throughout this subsection, we adopt the notations and assumptions in Theorem \ref{thm:DualFBP-moderate}. Let us define $\gam := \frac{\lam}{\lam-1}\in(0,1)$,
\begin{align}\label{eq:Dc}
	\Dc := \{(y,\varphi,\psi):y>0,\varphi>0,\psi\in (0,1)\},
\end{align}
and
\begin{align}\label{Psi_def}
	\Psi(y) := 
	\frac{2\sig^2}{\phi_0\mu^2}\left(\frac{\del}{ \lam} +r + \rho-\del\right) (y-\phi_0)=
	\frac{1}{\phi_0}(\lambda-1)(y-\phi_0),\quad y>0.
\end{align}
To derive the second equation in \cref{Psi_def}, we have used the fact that $\lambda$ satisfies
\begin{align*}
	\frac{\mu^2}{2\sigma^2}\lambda^2-\left(\frac{\mu^2}{2\sigma^2}+r+\rho-\delta\right)\lambda -\delta = 0,
\end{align*}
as pointed out in Remark \ref{rem:lam}.
Note also that 
$\Psi(\gamma\phi_0)=1$, $\Psi(\phi_0)=0$.
Finally, for $(y,\varphi,\psi)\in \Dc$, we define
\begin{align}
	\label{eq:g1}
	g_1(y,\varphi,\psi) := \frac{\varphi}{y}(1-\psi),
\end{align}
and
\begin{align}\label{eq:g2}
	g_2(y,\varphi,\psi) :=
	\textstyle
	-\frac{2\rho\sig^2}{\mu^2} 
	\left[
	\frac{1-\psi}{y}
	\left(\frac{\mu^2}{2\rho\sig^2}\psi - (U_+')^{-1}\big(\varphi\big)+\frac{r-\del}{\rho}+1-\al\right)
	-\frac{r+\rho}{\rho \varphi} + \frac{\del}{\rho y}
	\right].
\end{align}

Because the boundary value conditions in \cref{eq:FBP_ybar} are in the interior of $\Dc$, and $g_1$ and $g_2$ are locally Lipschitz inside $\Dc$, \cref{eq:FBP_ybar} is locally solvable inside $\Dc$. In particular, there exists a function $\epsilon:(\gamma\phi_0,\phi_0)\to [0,\infty)$ with $\epsilon(y)<y$, such that $(\epsilon(\bar{y}),\bar{y}]$ is the maximal interval in which the solution of \cref{eq:FBP_ybar} exists inside $\Dc$. 

The next result provides further properties of the solution of \cref{eq:FBP_ybar} and, specifically, its dependence on $\yb$. Figure \ref{fig:phi_psi_moderate} illustrates the solutions of \cref{eq:FBP_ybar} for various values of $\yb$, and it is helpful to refer to this figure when  reading the statement of the lemma and its proof.

\begin{lemma}\label{lemma:exit_FBP_moderate}
	Given a $\yb\in (\gamma\phi_0,\phi_0)$, let $\big(\varphi_{\bar{y}}(\cdot),\psi_{\bar{y}}(\cdot)\big)$ be the local solution of \cref{eq:FBP_ybar} and $\epsilon(\bar{y})$ be the left endpoint of its maximal existence interval. Then, we have:\vspace{1ex}
	
	\noindent$(i)$ If $\gam\phi_0< \yb' < \yb<\phi_0$ and $\eps(\yb)<\yb'$, then $\varphi_{\yb'}(y)>\varphi_{\bar{y}}(y)$ and $\psi_{\bar{y}'}(y)>\psi_{\bar{y}}(y)$ for values of $y$ at which both solutions exist (i.e. for $\max\{\epsilon(\bar{y}),\epsilon(\bar{y}')\}< y\le\bar{y}'$).\vspace{1ex}
	
	\noindent$(ii)$ If $\epsilon(\bar{y})>0$, then $\big\{\big(y,\varphi_{\bar{y}}(y),\psi_{\bar{y}}(y)\big):\eps(\yb)<y\le\yb\big\}$ exits $\Dc$ either through
	$
	\overline{\Dc}_0 := (0,\phi_0)^2\times \{0\}
	$, 
	or through 
	$
	\overline{\Dc}_1:= (0,\gam\phi_0)\times (0,\phi_0)\times \{1\}.
	$
	\vspace{1ex}
	
	\noindent$(iii)$ For $\bar{y}$ sufficiently close to $\gam\phi_0$ (respectively, $\phi_0$), $\big\{\big(y,\varphi_{\bar{y}}(y),\psi_{\bar{y}}(y)\big):\eps(\yb)<y\le\yb\big\}$ exits $\Dc$ through $\overline{\Dc}_1$ (respectively, through $\overline{\Dc}_0$).\qed
\end{lemma}

\begin{proof}
	\noindent$(i)$ Define $\big(\overline{\varphi}(y),\overline{\psi}(y)\big):=\big(\phi_0,\Psi(y)\big)$, $y\in(0,\yb)$, with $\Psi(\cdot)$ given by \cref{Psi_def}. For $y\in(0,\yb)$, we have
	\begin{align*}
		&\overline{\psi}'(y)-g_2\big(y,\overline{\varphi}(y),\overline{\psi}(y)\big)\\*
		&=\textstyle\Psi'(y)+\frac{2\sigma^2}{\mu^2}\left[\frac{1-\Psi(y)}{y}\left(\frac{\mu^2}{2\sigma^2}\Psi(y)-\rho (U_+')^{-1}(\phi_0)+r-\delta+\rho-\rho\alpha\right)-\frac{r+\delta}{\phi_0}+\frac{\delta}{y}        \right]\\
		&=\textstyle\frac{2\sigma^2}{\mu^2}\left[ \left( \frac{\delta}{\lambda}-\delta\right)\frac{1}{\phi_0}+\frac{\delta}{y}  +\frac{1-\Psi(y)}{y}\left(-\rho (U_+')^{-1}(\phi_0)-\rho\alpha+y\Psi'(y)-\frac{\delta}{\lambda}       \right) \right]\\
		&\textstyle\leq \frac{2\sigma^2}{\mu^2}\left[ \delta\left(\frac{1}{y}-\frac{1}{\phi_0}\right)\left(1-\frac{1}{\lambda}\right)+\frac{\Psi(y)\delta}{\lambda y}      \right]
		=0.
	\end{align*}
	We have used $\Psi'(y)=\frac{2\sigma^2}{\phi_0\mu^2}\left( \frac{\delta}{\lambda}+r+\rho-\delta\right)$ and 
	$\Psi(y)=y\Psi'(y)-\frac{2\sigma^2}{\mu^2}\left(\frac{\delta}{\lambda}+r+\rho-\delta\right)$ (c.f. \cref{Psi_def}) for the second step. The third step follows from
	$-\rho (U_+')^{-1}(\phi_0)-\rho\alpha+y\Psi'(y)	=-\rho c_0 + \frac{y}{\phi_0}(\lam-1)<0$, in which we have used $\phi_0:= U_+'(c_0-\al)$ from Lemma \ref{lem:Ut} and that $\lam<0$ by \cref{eq:lam_bounds}. The last equality directly follows from \cref{Psi_def}. By \cref{eq:g1}, we also have that
	$\overline{\varphi}'(y)-g_1(y,\overline{\varphi}(y),\overline{\psi}(y))=-\frac{1}{y}\overline{\varphi}(y)(1-\overline{\psi}(y))\leq 0$. 
	We have thus shown that
	\begin{align*}
		\begin{cases}
			\overline{\varphi}'(y)\le g_1(y,\overline{\varphi}(y),\overline{\psi}(y)),\\
			\overline{\psi}'(y)\le g_2\big(y,\overline{\varphi}(y),\overline{\psi}(y)),
		\end{cases}
	\end{align*}
	for $y\in(0,\yb)$. Since $\overline{\varphi}(\cdot)$ and $\overline{\psi}(\cdot)$ (trivially) satisfy the terminal conditions of \cref{eq:FBP_ybar}, the comparison theorem for the system of differential equation in \cref{eq:FBP_ybar} (see Lemma B.2 of \cite{ABY22}) then yields that $\psi_{\bar{y}}(y)\leq \Psi(y)$ for $y\in\big(\eps(\yb),\yb)$. Therefore,  
	$\psi_{\bar{y}'}(\bar{y}') = \Psi(\bar{y}') \ge \psi_{\bar{y}}(\bar{y}')$. 
	Statement $(i)$ then follows by comparing (e.g., using Lemma B.2 of \cite{ABY22}) $(\phi_{\yb'},\psi_{\yb'})$ and $(\phi_{\yb},\psi_{\yb})$ over the domain $\big(\max\{\epsilon(\bar{y}),\epsilon(\bar{y}')\},\bar{y}'\big]$. That the inequalities are strict follows from uniqueness of the solution of \cref{eq:FBP_ybar}.\vspace{1em}
	
	%
	%
	
	\begin{figure}[t]
		\centerline{
				\adjustbox{trim={0.0\width} {0.0\height} {0.0\width} {0.0\height},clip}
				{\includegraphics[scale=0.4, page=1]{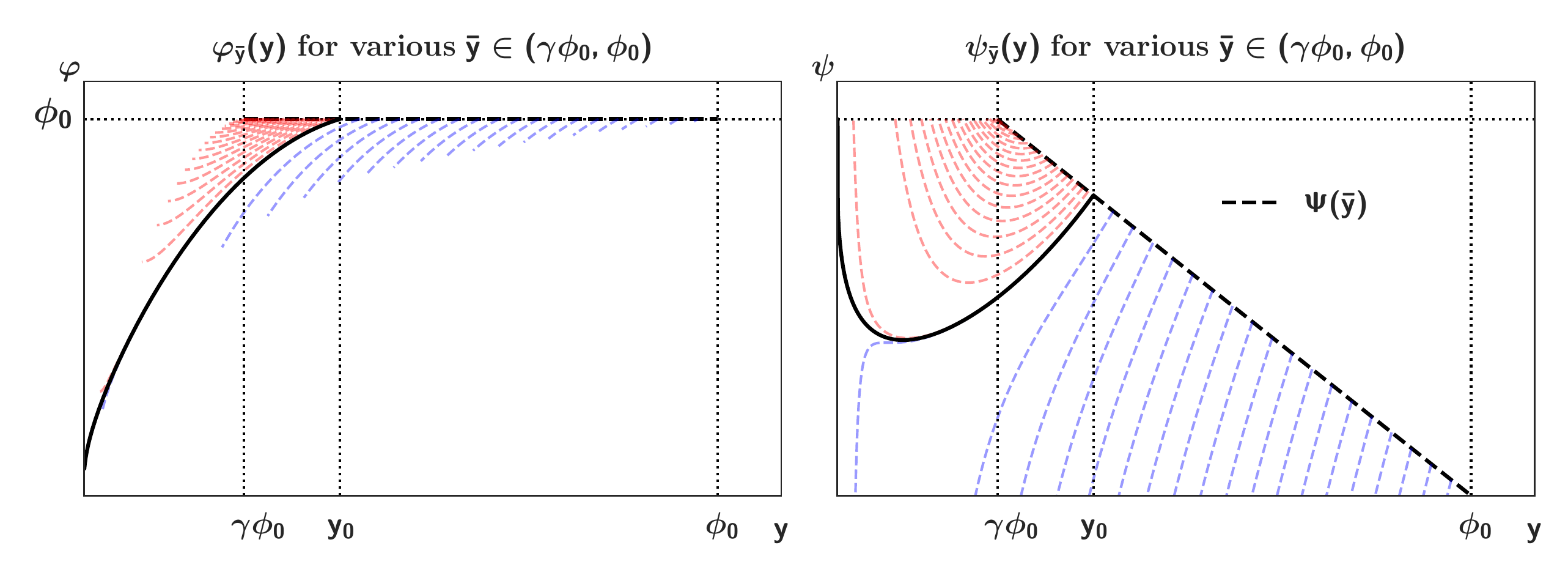}}
		}
		\caption{The solutions $(\varphi_{\yb},\psi_{\yb})$ of the terminal value problem \cref{eq:FBP_ybar} for various values of $\yb\in(\gam \phi_0, \phi_0)$. Each solution is shown up to its exit from the domain $\Dc$ of \cref{eq:Dc}, that is, for values of $y\le \yb$ satisfying $\big(y,\varphi_{\yb}(y),\psi_{\yb}(y)\big)\in\Dc$. 
			Blue dashed curves indicate solutions that exit from the boundary $\overline{\Dc}_0 := (0,\phi_0)^2\times \{0\}$ (represented by the horizontal line $\psi=0$ in the plot on the right side), while red dashed curves are solutions exiting from the boundary $\overline{\Dc}_1:= (0,\gam\phi_0)\times (0,\phi_0)\times \{1\}$ (i.e. the horizontal line $\psi=1$ in the plot on the right side). For each $\yb$, the value of $\eps(\yb)$ is the value of $y$ at which $\psi_{\yb}(y)=0$ (for the blue curves) or $\psi_{\yb}(y)=1$ (for the red curves). The interval $[\eps(\yb),\yb]$ is the domain of the solution $(\varphi_{\yb},\psi_{\yb})$ inside $\Dc$. Finally, $y_0$ is the value (of $\yb$) such that $\eps(y_0)=0$. In particular, $\varphi=\varphi_{y_0}$ and $\psi=\psi_{y_0}$ are the solution of \cref{eq:ODESys-moderate} and \cref{eq:FB-moderate} in Theorem \ref{thm:DualFBP-moderate}.These functions are shown by the solid black lines in the plots.
			\label{fig:phi_psi_moderate}}
	\end{figure}

	\noindent$(ii)$ Because $\epsilon(\bar{y})>0$, we have that $(\varphi_{\bar{y}},\psi_{\bar{y}})$ exits $\Dc$ at the point $(\epsilon(\bar{y}),\phih,\psih)\in (0,\yb)\times [0,\phi_0)\times [0,1]$ and that at least one of the following holds (a) $\phih=0$;  (b) $\psih=0$; or (c) $\psih=1$. See Figure \ref{fig:phi_psi_moderate} for an illustration. We first claim that $\phih>0$ (that is, case (a) is impossible). Assume on the contrary that $\phih:=\varphi_{\yb}(\epsilon(\yb)+)=0$. Define $\big(\phit(y),\psit(y)\big):=(y,0)$, $y>0$. For $y\in (\epsilon(\bar{y}),\bar{y}]$, we then have
	\[
		\phit'(y)=1=\frac{1}{y}\phit(y)(1-\psit(y))=g_1(y,\phit(y),\psit(y)),
	\]
	and 
	\begin{align*}\textstyle
		&\psit'(y)-g_2(y,\phit(y),\psit(y))
		=\textstyle
		\frac{2\rho\sigma^2}{\mu^2}
		\left[
		\frac{1}{y}\left(
		-(U_+')^{-1}(y)+\frac{r-\delta}{\rho}+1-\alpha
		\right)
		-\frac{r+\rho}{\rho y}+\frac{\delta}{\rho y}
		\right]\\*
		&=\textstyle\frac{2\rho\sigma^2}{y\mu^2}\left[ -(U_+')^{-1}(y)-\alpha\right]
		< 0 =\psi_{\bar{y}}'(y)-g_2\big(y,\varphi_{\bar{y}}(y),\psi_{\bar{y}}(y)\big).
	\end{align*}
	Because $\varphi_{\bar{y}}(\epsilon(\bar{y})+)=0<\phit(\epsilon(\bar{y})+)$, a comparison argument similar to the one in the proof of part $(i)$ 
	yields that $\varphi_{\bar{y}}(\bar{y})\leq \phit(\bar{y})=\bar{y}<\phi_0$. This assertion, however, contradicts the boundary condition $\varphi_{\bar{y}}(\bar{y})=\phi_0$. Hence, we must have $\phih=\varphi_{\yb}(\epsilon(\yb)+)>0$, as claimed.
	
	To complete the proof of part $(ii)$, it only remains to show that if $\psi_1=1$, then $\epsilon(\bar{y})<\gamma\phi_0$. From the proof of part $(i)$, we have that $\psi_{\bar{y}}(y)\leq \Psi(y)$ for $y\in(\eps(\yb),\yb)$, implying $\epsilon(\bar{y})\leq\gamma\phi_0$ when $\psi_1=1$. Finally, we show that $\epsilon(\bar{y})\ne \gamma\phi_0$. If $\epsilon(\bar{y})\ne \gamma\phi_0$, then for any $\yb'\in(\gamma\phi_0, \bar{y})$, part $(i)$ yields that $(\varphi_{\yb'},\psi_{\yb'})$ exits $\Dc$ through $\overline{\Dc}_1$ and $\epsilon(\yb')\geq \epsilon(\bar{y})$. It then follows that $\epsilon(\bar{y})=\epsilon(\yb')=\gamma\phi_0$, contradicting the uniqueness of the solution of \cref{eq:FBP_ybar}.\vspace{1em}
	
	\noindent $(iii)$  $g_1$ and $g_2$ are Lipschitz in a neighborhood $\Nc$ of the point $(\gamma\phi_0,\phi_0,1)$. Thus, for $\yb$ sufficiently close to $\gam\phi_0$, the solution $(\varphi_{\bar{y}},\psi_{\bar{y}})$ exists in $\Nc$. In view that
	\begin{align*}
		\textstyle
		g_2(\gamma\phi_0,\phi_0,1) = -\frac{2\sigma^2}{\mu^2}\left[ \frac{\delta}{\gamma\phi_0}-\frac{r+\rho}{\phi_0}    \right]
		=\frac{2\sigma^2}{\mu^2\phi_0}\left[r+\rho-\delta+\frac{\delta}{\lambda}\right]
		=\Psi'(\gamma\phi_0)<0,
	\end{align*}
	and the fact that $(\varphi_{\bar{y}},\psi_{\bar{y}})$ continuously depends on $\bar{y}$ when $\bar{y}$ is close to $\gamma\phi_0$, it holds that $(\varphi_{\bar{y}},\psi_{\bar{y}})$ must exit $\Dc$ through $\overline{\Dc}_1$ as $\yb\to \gam\phi_0^+$. The statement for $\yb\to\phi_0^-$ can be proved in a similar fashion.
\end{proof}

%
%

\subsection{Proof of Theorem \ref{thm:DualFBP-moderate}}
\label{proof:moderate_dual_FBP}

Based on the results from the previous subsection, we are ready to complete the proof of Theorem \ref{thm:DualFBP-moderate}.

\noindent \textbf{\emph{Proof of Theorem \ref{thm:DualFBP-moderate}.(i).}}
Let
\begin{align}\label{eq:y0_cand}
	y_0 := \sup\Big\{\bar{y}\in (\gamma\phi_0,\phi_0): (\varphi_{\bar{y}'},\psi_{\bar{y}'}) {\rm\ exits\ }\Dc {\rm\ through\ }\overline{\Dc}_1{\rm\ for\ any\ }\bar{y}'\in (\gamma\phi_0,\bar{y}) \Big\},
\end{align}
and note that $y_0\in(\gamma\phi_0,\phi_0)$ by Lemma \ref{lemma:exit_FBP_moderate}-(iii). To prove the first statement in Theorem \ref{thm:DualFBP-moderate}.(i) (that is, existence of $y_0$, $\varphi(y)$, and $\psi(y)$), we show that $\eps(y_0)=0$, with $\eps(\cdot)$ defined right before Lemma \ref{lemma:exit_FBP_moderate}. 

Suppose in contrary that $\epsilon(y_0)>0$. By \cref{eq:y0_cand}, there is an increasing sequence $\{\yb_n\}_{n=1}^\infty\to y_0$ such that $(\varphi_{\bar{y}_n},\psi_{\bar{y}_n})$ exits $\Dc$ through $\overline{\Dc}_1$ for all $n$, which implies $\psi_{\bar{y}_n}(y)\geq 0$ for $y\in (\epsilon(\bar{y}_n),\bar{y}_n)$. From the continuous dependence of the solution of \cref{eq:FBP_ybar} on its terminal conditions, we deduce that $\psi_{y_0}(y)\geq 0$, $y\in (\epsilon(y_0),y_0)$. We then reach the following dichotomy between cases (a) and (b) below.  To prove that $\eps(y_0)=0$, we will show that each case leads to a contradiction. See Figure \ref{proof_figure_moderate} for a visualization.\vspace{1ex}

\begin{figure}[tbp]
	\begin{subfigure}{0.5\textwidth}
		\centering
			\adjustbox{trim={0.1\width} {0.075\height} {0.2\width} {0.0\height},clip}
			{\includegraphics[scale=0.2, page=1]{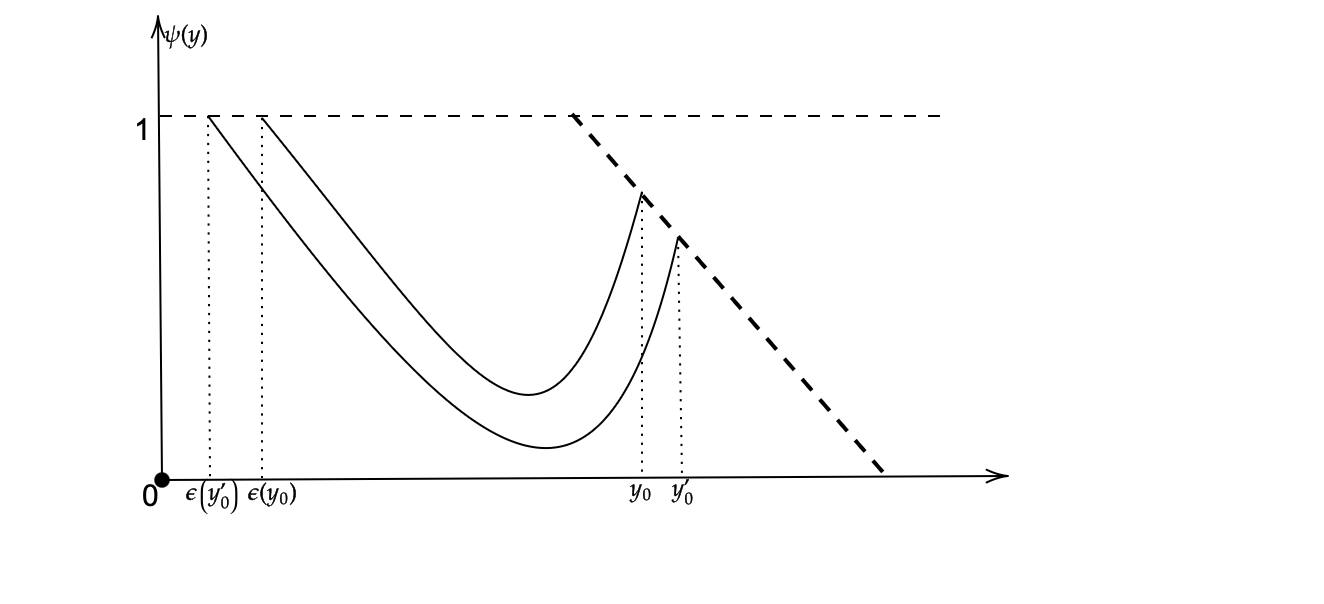}}
		\caption{Case (a).}
	\end{subfigure}%
	\begin{subfigure}{0.5\textwidth}
		\centering
			\adjustbox{trim={0.1\width} {0.075\height} {0.2\width} {0.0\height},clip}
			{\includegraphics[scale=0.2, page=1]{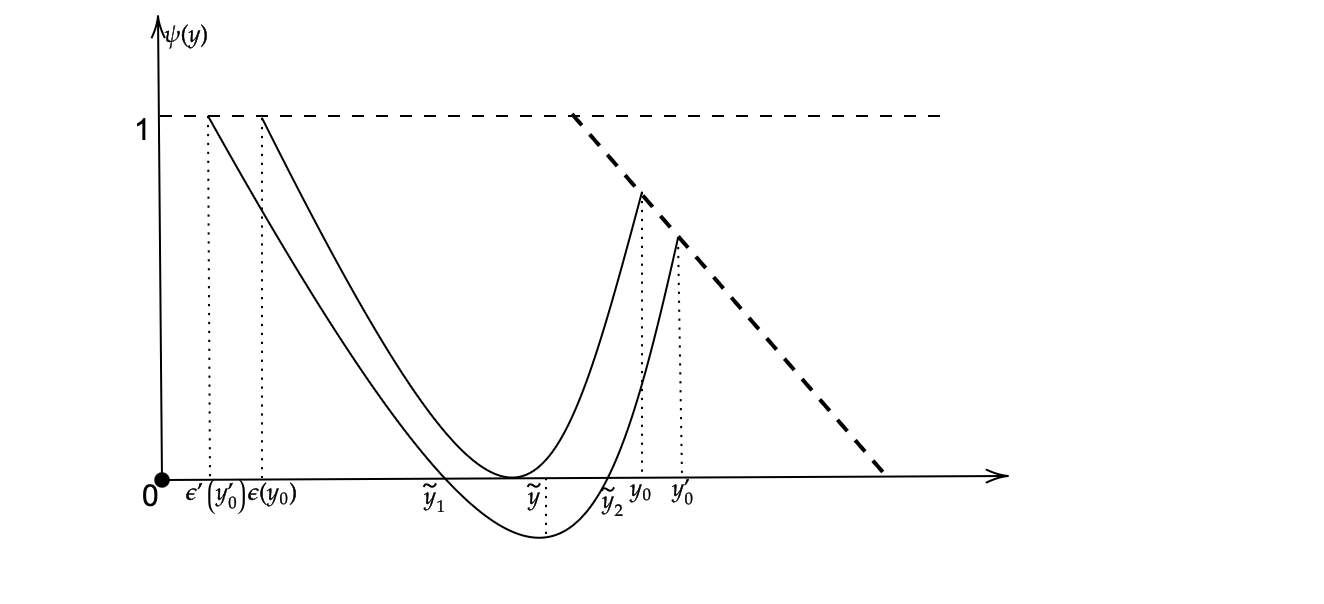}}
		\caption{Case (b).}
	\end{subfigure}
	\caption{Illustrations for the proof of Theorem \ref{thm:DualFBP-moderate}-(i). Note that these plots are used in a proof-by-contradiction. They represent cases that cannot be true.}
	\label{proof_figure_moderate}
\end{figure}
\noindent \emph{In case (a)}, we have $\min_y\big\{\psi_{y_0}(y): \epsilon(y_0)\le y\le y_0\big\}>0$. Because the solution of \cref{eq:FBP_ybar} continuously depends on the boundary conditions, there exists a $y_0'>y_0$ sufficiently close to $y_0$ such that 
$\min_y\big\{\psi_{y_0'}(y):\epsilon(y_0')\le y\le y_0'\big\}>0$. Therefore, $ (\varphi_{y_0'},\psi_{y_0'})$ exits $\Dc$ through $\overline{\Dc}_1$. For any $\bar{y}\in (y_0,y_0')$, Lemma \ref{lemma:exit_FBP_moderate}-(i) then yields that $(\varphi_{\bar{y}},\psi_{\bar{y}})$ also exits $\Dc$ through $\overline{\Dc}_1$. The last statement is in contradiction with the definition of $y_0$.
\vspace{1ex}

\noindent\emph{In case (b)}, we have $\min_y\big\{\psi_{y_0}(y): \epsilon(y_0)\le y\le y_0\big\}=0$. Define $\Dc_{\infty}:=\{y>0,\varphi>0,\psi<1\}\supset\Dc$. As $g_1$ and $g_2$ (given by \cref{eq:g1} and \cref{eq:g2}, respectively) are locally Lipschitz in $\Dc_\infty$, we have that \cref{eq:FBP_ybar} has a unique solution $(\varphi_{\yb},\psi_{\yb})$ that extends to the boundary of $\Dc_\infty$, and that is continuously dependent on $\yb$. Note that $\psi_{\yb}$ is now allowed to take negative values. Denote by $\epst(\yb)\in(0,\yb)$ the infimum of the maximal interval on which the solution of \cref{eq:FBP_ybar} exists in $\Dc_{\infty}$, noting that $\psi_{\bar{y}}$ may now explode to $-\infty$. By continuous dependence on the boundary conditions, there exists $y_0'>y_0$ sufficiently close to $y_0$ such that $0<\epst(y_0')<\epsilon(y_0)$, $\varphi_{y_0'}(\epst(y_0')+)>0$, $\psi_{y_0'}(y_0')=\Psi(y_0')>0$, and $\psi_{y_0'}(\epst(y_0')+)=1$. Because $\min_y\big\{\psi_{y_0}(y): \epsilon(y_0)\le y\le y_0\big\}=0$, Lemma \ref{lemma:exit_FBP_moderate}-(i) yields that there exists $\tilde{y}\in (\epst(y_0'),y_0')$ such that $\psi_{y_0'}(\tilde{y})<0$. Furthermore, due to $\psi_{y_0'}(y_0')=\Psi(y_0')>0$ and $\psi_{y_0'}(\epst(y_0')+)=1>0$, there exist $\tilde{y}_1$ and $\tilde{y}_2$ such that $\epst(y_0')<\tilde{y}_1<\tilde{y}<\tilde{y}_2<y_0'$,  $\psi_{y_0'}(\tilde{y}_1)=\psi_{y_0'}(\tilde{y}_2)=0$, and $\psi_{y_0'}(y)<0$ for $y\in (\tilde{y}_1,\tilde{y}_2)$, as illustrated in Figure \ref{proof_figure_moderate}.(b). In view that $\psi_{y_0'}(y)<0$ for $y\in (\tilde{y}_1,\tilde{y}_2)$ and $\psi_{y_0'}(\tilde{y}_1)=\psi_{y_0'}(\tilde{y}_2)=1$, we must have
\begin{align}\label{eq:psiprime1}
	\begin{cases}
		\psi'_{y_0'}(\tilde{y}_2)\geq 0,\\[1ex]
		\psi'_{y_0'}(\tilde{y}_1)\leq 0.
	\end{cases}
\end{align}
Setting $\yb=\yt_i$, $i\in\{1,2\}$, in \cref{eq:FBP_ybar} then yields
\begin{align}\label{eq:psiprime2}
	\psi'_{y_0'}(\tilde{y}_i)=\textstyle\frac{2\rho}{\tilde{y}_i\mu^2}\left[\textstyle\alpha+(U_+')^{-1}(\varphi_{y_0'}(\tilde{y}_i))-\frac{r+\rho}{\rho}+\frac{r+\rho}{\rho}\frac{\tilde{y}_i}{\varphi_{y_0'}(\tilde{y}_i)}\right],\quad i\in\{1,2\}.
\end{align}
From the first equation in \cref{eq:psiprime1}, it follows that
\begin{align}\label{eq:caseb_aux1}
	\frac{\tilde{y}_2}{\varphi_{y_0'}(\tilde{y}_2)}\geq 1-\frac{\rho}{r+\rho}\left(\alpha+(U_+')^{-1}(\varphi_{y_0'}\big(\tilde{y}_2)\big)\right).
\end{align}
By \cref{eq:FBP_ybar}, we have $\varphi_{y_0'}'(y)=g_1\big(y,\varphi_{y_0'}(y),\psi_{y_0'}(y)\big)> \varphi_{y_0'}(y)/y$ for $y\in (\tilde{y}_1,\tilde{y}_2)$. Therefore,
\begin{align}\label{eq:caseb_aux2}
	\textstyle
	\frac{\dd}{\dd y}\left(\frac{y}{\varphi_{y_0'}(y)}\right) = \frac{\varphi_{y_0'}(y)-y\varphi_{y_0'}'(y)}{\varphi_{y_0'}(y)^2}< 0,\quad y\in(\yt_1,\yt_2).
\end{align}
By using \cref{eq:caseb_aux2} and then \cref{eq:caseb_aux1}, we can obtain that
\begin{align*}
	\frac{\tilde{y}_1}{\varphi_{y_0'}(\tilde{y}_1)}> \frac{\tilde{y}_2}{\varphi_{y_0'}(\tilde{y}_2)}&\geq 1-\frac{\rho}{r+\rho}(\alpha+(U_+')^{-1}(\varphi_{y_0'}(\tilde{y}_2)))
	\geq 1-\frac{\rho}{r+\rho}(\alpha+(U_+')^{-1}(\varphi_{y_0'}(\tilde{y}_1))),
\end{align*}
in which the last inequality holds because $\varphi_{y_0'}$ in increasing by \cref{eq:FBP_ybar} and $U_+'$ (and thus $(U_+')^{-1}$) is decreasing by \cref{eq:Uc}. From \cref{eq:psiprime2}, we then obtain that
\begin{align*}
	\psi'_{y_0'}(\tilde{y}_1)=\textstyle\frac{2\rho}{\tilde{y}_1\mu^2}\left[\textstyle\alpha+(U_+')^{-1}(\varphi_{y_0'}(\tilde{y}_1))-\frac{r+\rho}{\rho}+\frac{r+\rho}{\rho}\frac{\tilde{y}_1}{\varphi_{y_0'}(\tilde{y}_1)}\right]>0.
\end{align*}
The last inequality contradicts the second equation in \cref{eq:psiprime1}.

We have thus shown that both cases $(a)$ and $(b)$ are impossible. Therefore, $\eps(y_0)=0$. By setting $\yb=y_0$ in the boundary value problem \cref{eq:FBP_ybar}, we obtain that $y_0$, $\phi=\phi_{y_0}$, and $\psi=\psi_{y_0}$ satisfy \cref{eq:ODESys-moderate} and \cref{eq:FB-moderate}.

It only remains to show that if $\lim_{y\to 0+} \varphi(y)=0$, then $\lim_{y\to 0+}\psi(y)=1$. Assume that $\lim_{y\to 0+} \varphi(y)=0$. By considering the change of variable $z=\log y$, $\phit(z)=\varphi(e^z)$, and $\psit(z)=\psi(e^z)$, we transform \cref{eq:ODESys-moderate} into the following system:
\begin{equation}\label{eq:ODESys-moderate-logtransform}
	\left\{	
	\begin{aligned}
		\phit'(z) &=\phit(z)\big(1-\psit(z)\big),\vspace{1ex}\\
		\psit'(z) &= \psit(z)^2-\kappa \rho (\psit(z)-1)(U_+')^{-1}\big(\phit(z)\big)+\{\kappa[r-\delta+\rho(1-\alpha)]-1\}\psit(z)
		\\&\quad{}
		-\kappa[r+\rho(1-\alpha)]+\kappa \frac{(r+\rho)e^z}{\phit(z)},
	\end{aligned}
	\right.
\end{equation}
where $\kappa = \frac{2\sigma^2}{\mu^2}$. From Theorem \ref{thm:DualFBP-moderate}, we know that $\psit(z)\in [0,1]$ for any $z<z_0:=\log y_0$. 

To complete the proof of Theorem \ref{thm:DualFBP-moderate}.(i), we will consider the only two possible cases, which we refer to as the ``monotone case'' and ``oscillatory case,'' respectively. In each case, we will prove that $\psi(0+)=\psit(-\infty)=1$.\vspace{1ex}

\noindent \textbf{Case 1 (the monotone case):} Assume that $\psit(z)$ is monotonic in a neighborhood of $-\infty$. In other words, assume that there exists a $\underline{z}$ such that either $\tilde{\psi'}(z)\geq 0$ or $\tilde{\psi'}(z)\leq 0$ for any $z<\underline{z}$ (that is, $\psit$ is monotone on $(-\infty,\underline{z}]$). In this case, $c:=\lim_{z\to -\infty}\psit(z)\in [0,1]$ is well-defined. Suppose $c<1$. From the second equation of \cref{eq:ODESys-moderate-logtransform}, we obtain that
\begin{align*}
	&\liminf_{z\to -\infty} \psit'(z)\\
	&\geq -\kappa \rho (c-1)\liminf_{z\to -\infty}(U_+')^{-1}\big(\phit(z)\big) -|\kappa[r-\delta+\rho(1-\alpha)]-1|
	-\kappa[r+\rho(1-\alpha)]
	= \infty.
\end{align*}
Therefore, there exist constants $C>0$ and $z_C<z_0$, such that $\psit'(z)\geq C$ for $z<z_C$. We then must have that $\psi(z)\leq C(z-z_C)+\psi(z_C)\to -\infty$ as $z\to -\infty$. The last statement contradicts with $\lim_{z\to -\infty}\psi(z)=c\geq 0$.\vspace{1em} 

\noindent \textbf{Case 2 (the oscillatory case):} 
Assume that $\psit(z)$ is not monotonic in a neighborhood of $-\infty$. That is, for any $\zu\le z_0$, there always exist $z_1,z_2<\zu$ such that $\psit'(z_1)<0$ and $\psit'(z_2)>0$. Throughout the proof of this case, let $\zt_0$ be the largest stationary point of $\psit(t)$, i.e., $\zt_0:=\max\{z\in(-\infty,z_0]\,:\, \psit'(z)=0\}$. Note that $\zt_0\in(-\infty, z_0]$, by the standing assumption of Case 2 and continuity of $\psit'(z)$.

This case is more intricate because it is unclear beforehand if $\lim_{z\to -\infty}\psit(z)$ exists. Indeed, due to the fact that $\psit(z)$ is not monotone in a neighborhood of $-\infty$, we need to specifically exclude the case $\liminf_{z\to -\infty}\psit(z)\ne \limsup_{z\to -\infty}\psit(z)$. We have $\limsup_{z\to -\infty}\psit(z)\le 1$, because $\psit(z)\le 1$ for $z\le z_0$. Thus, to show that $\lim_{z\to -\infty}\psit(z)=1$ in Case 2, it suffices to show that $\liminf_{z\to -\infty}\psit(z)\ge 1$. This will be our goal for the remainder of the proof of Theorem \ref{thm:DualFBP-moderate}.(i).

Let $\{z_n\}_{n=1}^{\infty}\subset(-\infty,\zt_0]$ be an arbitrary decreasing sequence satisfying $\lim_{n\to\infty}z_n=-\infty$. We will eventually show that $\lim_{n\to\infty}\psit(z_n)\ge1$, which implies that $\liminf_{z\to -\infty}\psit(z)\ge 1$ (since the sequence $\{z_n\}_{n=1}^{\infty}$ is arbitrary).

For each such sequence, we define a related sequence $\{z^*_n\}_{n=1}^{\infty}$  by
\begin{equation}\label{app:def-zstar}
	z^*_n =
	\begin{cases}
		\sup\Big\{z\in(z_n,\zt_0] \,:\, \psit'(z')<0\; {for}\, z'\in[z_n,z]\Big\}; &\quad \text{if}\;\psit'(z_n)<0,\\[1ex]
		z_n; &\quad \text{if}\;\psit'(z_n)=0,\\[1ex]
		\inf\Big\{z\in(-\infty, z_n) \,:\, \psit'(z')>0\; {for}\, z'\in[z,z_n]\Big\}; &\quad \text{if}\;\psit'(z_n)>0.\\[1ex]
	\end{cases}
\end{equation}
Roughly speaking, if $\psit(\cdot)$ is decreasing in a neighborhood of $z_n$, then $z^*_n$ is the closest stationary point of $\psit(\cdot)$ (i.e. $\psit'(z^*_n)=0$) that is larger than $z_n$. Similarly, if $\psit(\cdot)$ is increasing in a neighborhood of $z_n$, then $z^*_n$ is the closest stationary point of $\psit(\cdot)$ that is smaller than $z_n$.

Note that $z^*_n\in(-\infty,\zt_0]$ for all $n\ge1$. That $z^*_n\le \zt_0$ is clear from \cref{app:def-zstar}. If $z^*_{n_0}=-\infty$ for some $n_0\ge1$, then we must have $\psit'(z)>0$ on $(-\infty,z_{n_0}]$ by the third expression of \cref{app:def-zstar}. This contradicts our main assumption in Case 2 (that is, $\psit(z)$ is not monotonic in a neighborhood of $-\infty$). So, we must have $z^*_n\in(-\infty,\zt_0]$.

We also claim that $\{z^*_n\}_{n=1}^{\infty}$ is a non-increasing sequence.
To show this, 
take arbitrary indices $j>i\ge 1$. Note that $\{z_n\}_{n=1}^{\infty}$ is decreasing by assumption, we must have $z_j<z_i$. There are nine possibilities to consider. In each case below, we show that $z^*_j\le z^*_i$:
\begin{itemize}
	\item If $\psit'(z_j)< 0$ and $\psit'(z_i)< 0$, then
	\begin{align*}
		z^*_j &= \sup\Big\{z\in(z_j,\zt_0] \,:\, \psit'(z')<0\; {for}\, z'\in[z_j,z]\Big\}\\
		&\le \sup\Big\{z\in(z_i,\zt_0] \,:\, \psit'(z')<0\; {for}\, z'\in[z_i,z]\Big\}=z^*_i.
	\end{align*}
	
	\item If $\psit'(z_j)< 0$ and $\psit'(z_i) = 0$, then $z^*_j = \sup\Big\{z\in(z_j,\zt_0] \,:\, \psit'(z')<0\; {for}\, z'\in[z_j,z]\Big\} \le z_i=z^*_i$.
	
	\item If $\psit'(z_j)< 0$ and $\psit'(z_i) > 0$, then
	\begin{align*}
		z^*_j &= \sup\Big\{z\in(z_j,\zt_0] \,:\, \psit'(z')<0\; {for}\, z'\in[z_j,z]\Big\}\\
		&\le \inf\Big\{z\in(-\infty, z_i) \,:\, \psit'(z')>0\; {for}\, z'\in[z,z_i]\Big\}=z^*_i.
	\end{align*}
	
	\item If $\psit'(z_j)= 0$ and $\psit'(z_i) < 0$, then $z^*_j = z_j < z_i \le \sup\Big\{z\in(z_i,\zt_0] \,:\, \psit'(z')<0\; {for}\, z'\in[z_i,z]\Big\}=z^*_i$.
	
	\item If $\psit'(z_j)= 0$ and $\psit'(z_i) = 0$, then $z^*_j=z_j<z_i=z^*_i$.
	
	\item If $\psit'(z_j)= 0$ and $\psit'(z_i) > 0$, then
	\begin{align*}
		z^*_j = z_j \le \inf\Big\{z\in(-\infty, z_i) \,:\, \psit'(z')>0\; {for}\, z'\in[z,z_i]\Big\}=z^*_i.
	\end{align*}
	
	\item If $\psit'(z_j)> 0$ and $\psit'(z_i) < 0$, then
	\begin{align*}
		z^*_j &= \inf\Big\{z\in(-\infty, z_j) \,:\, \psit'(z')>0\; {for}\, z'\in[z,z_j]\Big\}\le z_j< z_i\\
		&\le \sup\Big\{z\in(z_i,\zt_0] \,:\, \psit'(z')<0\; {for}\, z'\in[z_i,z]\Big\}=z^*_i.
	\end{align*}
	
	\item If $\psit'(z_j)> 0$ and $\psit'(z_i) = 0$, then $z^*_j = \inf\Big\{z\in(-\infty, z_j) \,:\, \psit'(z')>0\; {for}\, z'\in[z,z_j]\Big\}\le z_j< z_i=z^*_i$.
	
	\item If $\psit'(z_j)> 0$ and $\psit'(z_i) > 0$, then
	\begin{align*}
		z^*_j &= \inf\Big\{z\in(-\infty, z_j) \,:\, \psit'(z')>0\; {for}\, z'\in[z,z_j]\Big\}\\
		&\le \inf\Big\{z\in(-\infty, z_i) \,:\, \psit'(z')>0\; {for}\, z'\in[z,z_i]\Big\}=z^*_i.
	\end{align*}
\end{itemize}

Finally, we claim that $\lim_{n\to\infty}z^*_n = -\infty$. Assume otherwise, as it has already been shown that $\{z^*_n\}_{n=1}^{\infty}$ is non-increasing, the only possibility is that $z^*_{\infty}:=\lim_{n\to\infty} z^*_n \in (-\infty,\zt_0]$. In view that $\lim_{n\to\infty} z_n=-\infty$ by assumption, there exist an index $N$ such that $z_n<z^*_{\infty}$ for all $n>N$. Therefore, we must have $z^*_n \ge z^*_{\infty} > z_n$ for all $n>N$. In other words, for $n>N$, $z^*_n$ is given by the top expression in \cref{app:def-zstar}. In particular, we must have $\psit(z)<0$ for any $z\in(z_n,z^*_{\infty})$ and any $n>N$. Note that $\lim_{n\to\infty} z_n=-\infty$ (by assumption), which implies that $\psit(z)<0$ for all $z\in(-\infty, z^*_{\infty})$, contradicting the standing assumption of Case 2 (that is, $\psit(z)$ is not monotonic in a neighborhood of $-\infty$).

So far, we have shown that $\{z^*_n\}_{n=1}^\infty$ is a non-increasing sequence such that $z^*_n\to-\infty$. From the definition of $z^*_n$ in \cref{app:def-zstar} and the continuity of $\psit(z)$, we also conclude that $\psit'(z^*_n)=0$ for all $n\ge1$.

We next claim that $\lim_{n\to\infty} \psit(z^*_n)=1$. As $\psit'(z^*_n)=0$ for all $n\ge1$, the second equation in \cref{eq:ODESys-moderate-logtransform} yields
\begin{align}\label{eq:psit_quadratic_eq}
	\psit(z^*_n)^2+b_n \psit(z^*_n) +c_n=0,\quad n\ge1,
\end{align}
in which  
$b_n := \kappa[r-\delta+\rho(1-\alpha)]-1-\kappa \rho (U_+')^{-1}\big(\phit(z^*_n)\big)$ and 
$c_n := \kappa \rho (U_+')^{-1}\big(\phit(z^*_n)\big)-\kappa[r+\rho(1-\alpha)]+\kappa \frac{(r+\rho)e^{z^*_n}}{\phit(z^*_n)}$. 
Thanks to the fact that $\lim_{n\to\infty}(U_+')^{-1}\big(\phit(z^*_n)\big)=(U_+')^{-1}\big(\phit(-\infty)\big)=(U_+')^{-1}\big(\varphi(0^+)\big)=+\infty$,
we have $b_n\to -\infty$ and $c_n\to +\infty$. Moreover, $b_n/c_n\to -1$ as $n\to\infty$. Therefore, if follows from \cref{eq:psit_quadratic_eq} that
\begin{equation*}
	\psit(z^*_n) \geq
	\frac{-b_n-\sqrt{b_n^2-4 c_n}}{2}
	= \frac{2 c_n}{-b_n+\sqrt{b_n^2-4 c_n}}
	= \frac{2}{-\frac{b_n}{c_n}+\sqrt{\left(\frac{b_n}{c_n}\right)^2-\frac{4}{c_n}}}
	\to 1,
\end{equation*}
as $n\to \infty$. As $\psit(z)\le 1$ for all $z\in(-\infty,z_0]$, we must have $\lim_{n\to\infty} \psit(z^*_n)=1$, as claimed.

Finally, by the definition of $z^*_n$ in \cref{app:def-zstar}, we have $\psit(z_n)\ge \psit(z^*_n)$ for all $n\ge1$. Therefore, $\lim_{n\to \infty}\psit(z_n)\ge \lim_{n\to \infty}\psit(z^*_n)=1$. Because $\{z_n\}_{n=1}^{\infty}$ was arbitrarily chosen, we conclude that $\liminf_{z\to -\infty}\psit(z)\ge 1$. As it has been argued that $\limsup_{z\to -\infty}\psit(z)\le 1$ in view of $\psit(z)\le 1$ for $z\le z_0$. We have thus shown that $\lim_{z\to -\infty}\psit(z)=1$ in Case 2, which completes the proof of Theorem \ref{thm:DualFBP-moderate}.(i).

\noindent \textbf{\emph{Proof of Theorem \ref{thm:DualFBP-moderate}.(ii).}} 
By \cref{eq:G}, we have $G'(\phi)=-\alpha-(U_+')^{-1}(\phi)$. Differentiating  \cref{eq:u_phi_psi_moderate} with respect to $y$ then yields
\begin{align}
	\nonumber
	u'(y)&=\textstyle
	\frac{1}{\delta}\left[ \frac{\mu^2}{2\rho\sigma^2}(\varphi'(y)\psi(y)+\varphi(y)\psi'(y))-\varphi'(y)(\alpha+(U_+')^{-1}(\varphi(y)))+\frac{\delta-r-\rho}{\rho}(1-\varphi'(y) )      \right]\\
	\nonumber
	&=\textstyle
	\frac{1}{\delta}\left[  -\frac{\varphi(y)(1-\psi(y))}{y}\left(\frac{r-\delta+\rho}{\rho}-\alpha-(U_+')^{-1}(\varphi(y))    \right)+\frac{r+\rho}{\rho}-\frac{\delta\varphi(y)}{\rho y}                         \right]\\
	&\qquad\textstyle{}+\frac{1}{\delta}\left[-\varphi'(y)(\alpha+(U_+')^{-1}(\varphi(y)))+\frac{\delta-r-\rho}{\rho}-\varphi'(y)\frac{\delta-r-\rho}{\rho}\right]=\frac{y-\varphi(y)}{\rho y}, \label{uprime-phi}
\end{align} 
which is \cref{up_phi_main}. Setting $y=y_0$ yields $u'(y_0)=\big(y_0-\varphi(y_0)\big)/(\rho y_0)$. As $\varphi(y_0)=\phi_0$ by \cref{eq:FB-moderate}, we obtain the first boundary condition in \cref{eq:modFBPDual-reduced}. The second boundary condition follows from \cref{eq:u_phi_psi_moderate},
\begin{align*}
	u(y_0)&=\frac{1}{\delta}\left[ \left(\frac{\delta}{\rho\lambda}+\frac{r+\rho-\delta}{\rho}    \right)(y_0-\phi_0)+G(\phi_0)+\frac{\delta-r-\rho}{\rho}(y_0-\phi_0) \right]
	=\frac{y_0-\phi_0}{\rho \lambda}+\frac{U(0)}{\delta},
\end{align*}
in which we used $G(\phi_0)=U(0)$ that was shown in \cref{eq:Psi-phi0}. 

To show that the differential equation in \cref{eq:FB-moderate} also holds, we proceed as follows. Differentiating \cref{uprime-phi} with respect to $y$ and noting that $\varphi'(y) = \frac{1}{y}\varphi(y)\big(1-\psi(y)\big)$ (by \cref{eq:ODESys-moderate}) yield
\begin{equation}\label{upp}
	u^{\prime\prime}(y)=\frac{1}{\rho y^2}\varphi(y)\psi(y),\quad 0<y<y_0,
\end{equation}
which is \cref{upp_phi_psi_main}. Therefore, 
\begin{align*}
	&\frac{\mu^2}{2\sig^2} y^2 u^{\prime\prime}(y) + G\big(y-\rho y u'(y)\big) 
	+ (\del - r - \rho) y u'(y) - \del u(y)\\*
	&=\frac{\mu^2}{2\rho\sigma^2}\varphi(y)\psi(y)+G(\varphi(y))+\frac{\delta-r-\rho}{\rho}(y-\varphi(y))-\delta u(y)=0,
\end{align*}
in which the last step follows from \cref{eq:u_phi_psi_moderate}. 

To show that $u$ is convex, we use \cref{upp} to obtain $u^{\prime\prime}\geq 0$ because of $\varphi,\psi\geq 0$. Furthermore, in view that $u'(y_0)=\frac{y_0-\phi_0}{\rho y_0}<0$ by \cref{uprime-phi}, and that $u^{\prime\prime}<0$, we conclude that $u'(y)<0$ for $y\in (0,y_0)$. Thus, $u$ is decreasing.

It only remains to show that $\lim_{y\to 0+}u'(y)=-\infty$. In Theorem \ref{thm:DualFBP-moderate}.(i), we have shown that either $\lim_{y\to 0+}\varphi(y)>0$ or $\lim_{y\to 0+}\psi(y)>0$ (or both). If $\lim_{y\to 0+}\varphi(y)>0$, we have $\lim_{y\to 0+}\frac{\varphi(y)}{y}=+\infty$. We thus deduce from \cref{uprime-phi} that $\lim_{y\to 0+}u'(y)=\frac{1}{\rho}-\frac{1}{\rho}\lim_{y\to 0+}\frac{\varphi(y)}{y}=-\infty$. If $\lim_{y\to 0+}\psi(y)=1$, we assume $\lim_{y\to 0+}u'(y)=-\infty$ does not hold and argue by contradiction. Because $u'$ is increasing and $u'(y_0)=\frac{y_0-\phi_0}{\rho y_0}<0$ by \cref{uprime-phi}, there exists a constant $m>0$ such that $\lim_{y\to0^+} u'(y) = -m$. It then follows from \cref{uprime-phi} that
\begin{align*}
	\lim_{y\to 0^+}\frac{\varphi(y)}{y} = 1+\rho m \in(1,+\infty),
\end{align*}
which yields $\lim_{y\to0^+}\varphi(y)=0$. By L'H\^{o}pital's rule, \cref{eq:ODESys-moderate}, and the fact $\lim_{y\to0^+}\psi(y)=1$ from part $(i)$, we then obtain 
\begin{align*}
	\lim_{y\to 0^+}\frac{\varphi(y)}{y} = \lim_{y\to 0^+}\varphi'(y)
	=\lim_{y\to 0^+} \frac{\varphi(y)}{y}\left(1-\psi(y)\right)
	=0,
\end{align*}
which yields a contradiction. Thus, we must have $\lim_{y\to0^+} u'(y)=-\infty$ as claimed.

%
%
\subsection{Proof of Lemma \ref{lem:trans}}\label{app:trans}

We first show that $v(x)$ is lower bounded. By \cref{eq:u1-moderate} and \cref{v_moderate}, we have
\begin{align*}
	v(x) = \frac{y_0-\phi_0}{\rho\lambda}\left(\frac{x}{x_0}\right)^{\frac{\lambda}{\lambda-1}} + \frac{U(0)}{\delta}, 
	\quad x\in(0,x_0).
\end{align*}
As $v(x)$ is strictly increasing (by Corollary \ref{coro:solveHJB}), we deduce that
\begin{align}\label{eq:VF_lowerbound}
	v(x)>v(0)=\frac{U(0)}{\delta}\in\Rb, \quad x>0.
\end{align}

Let $(\varphi(y), \psi(y), y_0)$ be the solution of the system of problems \cref{eq:ODESys-moderate} and \cref{eq:FB-moderate} in Theorem \ref{thm:DualFBP-moderate}, and let $u(y)$ be the solution of \cref{eq:modFBPDual1}--\cref{eq:modFBPDual4} (given in Corollary \ref{coro:solution_dual}). In light of the item $(i)$ of Theorem \ref{thm:DualFBP-moderate}, we shall split the proof into two separate cases, depending on whether  $\varphi(0)>0$ or $\psi(0)=1$.\vspace{1em} 

\noindent \textbf{Case 1:} Assume that $\varphi(0)>0$. In this case, we will show that $|v(x)|$ is in fact uniformly bounded for $x>0$, which readily yields \cref{eq:trans}. Let $(u(y), y_0)$ be the solution of \cref{eq:modFBPDual1}--\cref{eq:modFBPDual4} in Corollary \ref{coro:solution_dual}. It was shown in the proof of Corollary \ref{coro:solveHJB} that $(u')^{-1}(-x)\in(0,y_0)$ for $x>x_0$. Using \cref{v_moderate} and then applying the change-of-variable $y=(u')^{-1}(-x)\Leftrightarrow x=-u'(y)$, we obtain that
\begin{align}
	\nonumber
	&\lim_{x\to+\infty} v(x) 
	= \lim_{x\to+\infty} \left[u\big((u')^{-1}(-x)\big) + x\,(u')^{-1}(-x)\right]
	=\lim_{y\to0^+} \Big[u(y) - u'(y) y\Big]\\*
	\nonumber
	&= \frac{1}{\del} \lim_{y\to0^+} \left[
	\frac{\mu^2}{2\rho\sig^2} \varphi(y)\psi(y)
	+ G\big(\varphi(y)\big) +\frac{\del - r - \rho}{\rho} \big(y-\varphi(y)\big)\right]
	-\lim_{y\to0^+} \frac{y-\varphi(y)}{\rho}\\
	\label{eq:VF_upperbound}
	&= \frac{1}{\del} \left[
	\frac{\mu^2}{2\rho\sig^2} \varphi(0)\psi(0)+ G\big(\varphi(0)\big) -\frac{\del - r - \rho}{\rho} \varphi(0)\right]
	+\frac{\varphi(0)}{\rho}<+\infty.
\end{align}
The third equality follows from \cref{eq:u_phi_psi_moderate} and \cref{up_phi_main}. The last step follows from the boundedness of $\varphi(y)$ and $\psi(y)$ (see Theorem \ref{thm:DualFBP-moderate}.(i)) and that, by \cref{eq:G} and $\varphi(0)>0$, we have
\begin{align*}
	G\big(\varphi(0)\big) &:= 
	K+U_+\Big((U_+')^{-1}\big(\varphi(0)\big)\Big) - \varphi(0) \Big(\al+(U_+')^{-1}\big(\varphi(0)\big)\Big)<+\infty.
\end{align*}
Combing \cref{eq:VF_lowerbound}, \cref{eq:VF_upperbound}, and the fact that $v(x)$ is increasing (by Corollary \ref{coro:solveHJB}), we conclude that $|v(x)|$ is uniformly bounded for $x>0$, and \cref{eq:trans} easily follows.

\vspace{1em}

\noindent \textbf{Case 2:} Assume that $\psi(0)=1$. Take an arbitrary $x>0$ and any admissible $(\pi_t,c_t)_{t\geq 0}\in \Acr(x)$. Consider the process $\{Y_t\}_{t\geq 0}$ given by 
\begin{align}\label{eq:Y}
	Y_t := \exp\left(-\frac{\mu}{\sigma}B_t-\left[r+\rho+\frac{1}{2}\left(\frac{\mu}{\sigma}\right)^2 \right]t +\rho \int_0^t c_s\dd s\right),\quad t\ge 0.
\end{align}
Note that $\{Y_t\}_{t\geq 0}$ is the unique (strong) solution of the stochastic differential equation
\begin{align}\label{eq:Y-SDE}
	\begin{cases}
		\displaystyle\frac{\dd Y_t}{Y_t} = -\big(r+\rho(1-c_t)\big)\dd t - \frac{\mu}{\sig}\dd B_t,\quad t\ge0,\\
		Y_0 = 1.
	\end{cases}
\end{align}
From \cref{eq:X} and \cref{eq:Y-SDE}, it is easy to see that 
$X_T Y_T + \int_0^T c_s Y_s \dd s = x + \int_0^T \left(\sig \pi_s - \frac{\mu}{\sig} X_s\right) Y_s \dd B_s$, $T\ge0$.  
Therefore, by some standard localizing arguments, it holds that $\Eb[X_T Y_T]\le x$ for $T>0$. 
Using \cref{eq:Legendre} then yields
\begin{align*}
	\Eb\left[v(X_T)\right]
	&= \Eb\left[v(X_T) - X_T Y_T + X_T Y_T\right]
	\le \Eb\left[u(Y_T) + X_T Y_T\right]\\
	&\le \Eb\left[u(Y_T)\right] + x,\quad T>0.
\end{align*}
By taking \cref{eq:VF_lowerbound} into account, it follows that
\begin{align*}
	e^{-\delta T}\frac{U(0)}{\del}\le \Eb\left[e^{-\delta T}v(X_T)\right] \le \Eb\left[e^{-\delta T}u(Y_T)\right] + xe^{-\delta T},
	\quad T>0.
\end{align*}

That is, to prove the transversality condition \cref{eq:trans}, it is sufficient to show that
\begin{equation}\label{eq:trans:Y}
	\lim_{T \to +\infty}\Eb \left[e^{-\delta T}u(Y_T) \right]=0.
\end{equation}		
To this end, we first note that, by \cref{eq:u_phi_psi_moderate},
\begin{align}\label{eq:u_estimate1}
	u(y) &= \frac{1}{\del}\left[
	\frac{\mu^2}{2\rho\sig^2} \varphi(y)\psi(y)
	+ G\big(\varphi(y)\big) +\frac{\del - r - \rho}{\rho} \big(y-\varphi(y)\big)
	\right]\notag\\
	&\le \frac{\mu^2}{2\del\rho\sig^2} \phi_0 + \frac{\del - r - \rho}{\del\rho} y_0
	+ \frac{1}{\del} G\big(\varphi(y)\big),
\end{align}
for $0<y<y_0$. Furthermore, by Assumption \ref{assum:Growth}, there exists constants $A_1,A_2>0$ and $A_3\ge0$ such that
\begin{align}\label{eq:G_growth}
	G\big(\varphi(y)\big)\le A_1 + A_2 \varphi(y)^{-A_3}, \quad y\in(0,y_0].
\end{align}
Next, take an arbitrary constant $\eps>0$. From \cref{eq:u_estimate1}, \cref{eq:G_growth}, and Lemma \ref{lem:phi_estimate} below, it follows that,
\begin{align}\nonumber
	u(y) &\le \frac{\mu^2}{2\del\rho\sig^2} \phi_0 + \frac{\del - r - \rho}{\del\rho} y_0
	+ \frac{1}{\del} G\big(\varphi(y)\big)\\
	\nonumber
	&\le \frac{\mu^2}{2\del\rho\sig^2} \phi_0 + \frac{\del - r - \rho}{\del\rho} y_0
	+ \frac{A_1}{\del} + \frac{A_2}{\del} \varphi(y)^{-A_3}\\
	\label{eq:u_estimate2}
	&\le \frac{\mu^2}{2\del\rho\sig^2} \phi_0 + \frac{\del - r - \rho}{\del\rho} y_0
	+ \frac{A_1}{\del} + \frac{A_2 B_\eps^{-A_3}}{\del} y^{-\eps A_3},
\end{align}
for $y\in(0,y_0)$, in which $B_\eps>0$ is a constant that may depend on $\eps$ (see Lemma \ref{lem:phi_estimate} as below). By choosing appropriate constants $\At$ and $\Bt_\eps>0$, we may rewrite \cref{eq:u_estimate2} in a way that
\begin{align}\label{eq:u_estimate3}
	u(y) \le \At + \Bt_\eps\, y^{-\eps A_3},\quad 0<y\le y_0.
\end{align}
By virtue of \cref{eq:Y} and \cref{eq:u_estimate3}, we derive that
\begin{align*}
	&\Eb\left[e^{-\delta T}u(Y_T)\right]
	\le \At e^{-\delta T} + \Bt_\eps \Eb\left[\ee^{-\del T}(Y_T)^{-\eps A_3}\right]\\
	&= \At e^{-\delta T} + \Bt_\eps \Eb\left[\ee^{-\del T}\exp\left(-\frac{\mu}{\sigma}B_T-\left[r+\rho+\frac{1}{2}\left(\frac{\mu}{\sigma}\right)^2 \right]T +\rho \int_0^T c_s\dd s\right)^{-\eps A_3}\right]\\
	&= \At e^{-\delta T} + \Bt_\eps \Eb\left[\ee^{-\del T}\exp\left(
	\eps A_3 \frac{\mu}{\sigma}B_T
	+\eps A_3 \left[r+\rho+\frac{1}{2}\left(\frac{\mu}{\sigma}\right)^2 \right]T 
	-\eps A_3 \rho \int_0^T c_s\dd s
	\right)\right]\\
	&\le \At e^{-\delta T} + \Bt_\eps 
	\exp\left(\left[\eps A_3 \left(r+\rho+\frac{1}{2}\left(\frac{\mu}{\sigma}\right)^2 \right) - \del \right]T 
	\right)
	\Eb\left[\ee^{\eps A_3 \frac{\mu}{\sigma}B_T}\right]
	\\
	&\le \At e^{-\delta T} + \Bt_\eps 
	\exp\left(
	\left[
	\eps A_3 \left(r+\rho+\frac{1}{2}\left(\frac{\mu}{\sigma}\right)^2 \right) + \eps^2 A_3^2\frac{\mu^2}{\sig^2} - \del 
	\right] T 
	\right).
\end{align*}
On the other hand, by the Legendre transform and \cref{eq:VF_lowerbound}, it holds that $u(y)\ge\lim_{x\to0^+}\big(v(x)-xy\big)=v(0)$ for all $y>0$. It thus follows that
\begin{align}\nonumber
	&e^{-\delta T} v(0) \le \Eb\left[e^{-\delta T}u(Y_T)\right]\\
	\label{eq:u_estimate4}
	&\le \At e^{-\delta T} + \Bt_\eps 
	\exp\left(\textstyle
	\left[
	\eps A_3 \left(r+\rho+\frac{1}{2}\left(\frac{\mu}{\sigma}\right)^2 \right) + \eps^2 A_3^2\frac{\mu^2}{\sig^2} - \del 
	\right] T 
	\right),
\end{align}
for any $\eps>0$. Note that the value of the constant $A_3\ge0$ in \cref{eq:G_growth} does \emph{not} depend on $\eps$, we may choose a sufficiently small $\eps$ such that
\begin{align*}
	\eps A_3 \left(r+\rho+\frac{1}{2}\left(\frac{\mu}{\sigma}\right)^2 \right) + \eps^2 A_3^2\frac{\mu^2}{\sig^2} < \del.
\end{align*}
By choosing such a value of $\eps$ in \cref{eq:u_estimate4} and then letting $T\to+\infty$, we obtain \cref{eq:trans:Y}. This completes the proof.

To obtain \cref{eq:u_estimate2} in the last part of the proof, we have used the next lemma.	

\begin{lemma}\label{lem:phi_estimate}
	Assume that $\psi(0)=1$. Then, for any $\eps>0$, there exists a constant $B_\eps>0$ such that $\varphi(y)> B_\eps y^{\eps}$ for all $y\in(0,y_0]$.\qed
\end{lemma}
\begin{proof}
	Take an arbitrary $\eps>0$. As $\psi(0)=1$, there exists an $\eta\in(0,y_0)$ such that $1-\psi(y)<\eps$ for all $y\in(0,\eta)$.	It then follows that
	\begin{align*}
		\varphi'(y) - \frac{\eps}{y} \varphi(y) < \varphi'(y) - \frac{1-\psi(y)}{y} \varphi(y) = 0,\quad y\in(0,\eta),
	\end{align*}
	in which the last step follows from \cref{eq:ODESys-moderate}. Let $f(y):=\varphi(\eta)(y/\eta)^\eps$ for $y\in(0,\eta)$. Note that $f(\eta)=\varphi(\eta)$ and
	\begin{align*}
		\varphi'(y) - \frac{\eps}{y} \varphi(y) < 0 = f'(y) - \frac{\eps}{y} f(y),\quad y\in(0,\eta),
	\end{align*}
	we can apply the standard comparison theorem for boundary value ordinary differential equations (see, for instance, the Corollary on page 91 of \cite{Walter1998}) to obtain that
	\begin{align*}
		\varphi(y) > f(y) = \frac{\varphi(\eta)}{\eta^\eps} y^\eps \ge \frac{\varphi(\eta)}{y_0^\eps} y^\eps, \quad 0<y<\eta.
	\end{align*}
	As $\varphi(y)$ is increasing (according to Theorem \ref{thm:DualFBP-moderate}.(i)), we also have
	\begin{align*}
		\frac{\varphi(y)}{y^\eps} > \frac{\varphi(\eta)}{y_0^\eps}, \quad \eta\le y\le y_0.
	\end{align*}
	Finally, by setting $B_\eps := \varphi(\eta)/y_0^{\eps}>0$, the last two inequalities lead to the desired result $\varphi(y)\ge B_\eps y^{\eps}$ for all $0<y\le y_0$.
\end{proof}

%
%
\subsection{Proof of Theorem \ref{thm:verification_moderate}}\label{app:verify}
Take an arbitrary initial relative wealth $x>0$.
We first prove that \cref{SDE:optimal} has a positive unique strong solution. We consider two cases depending on whether $\varphi(0)>0$ or $\psi(0)=1$. These two cases are exhaustive according to Theorem \ref{thm:DualFBP-moderate}.(i).

\begin{itemize}
	\item \textbf{Case 1:} Assume that $\varphi(0)>0$. Note that if $\{X^*_t\}_{t\ge0}$ is a strong positive solution of \cref{SDE:optimal}, It\^o's formula yields that $\{Z^*_t=\log X^*_t\}_{t\ge0}$ is a strong solution of the stochastic differential equation
	\begin{equation}\label{SDE:optimal_Z}
		\begin{cases}
			\dd Z^*_t = \bt(Z^*_t)\dd t + \at(Z^*_t)\dd B_t,\\[1ex]
			Z^*_0 = \log x,
		\end{cases}
	\end{equation}
	in which $\bt(z):=r+\rho+\mu \pi^*(e^z)e^{-z}-(\ee^{-z}+\rho)c^*(e^z)-\frac{\sigma^2}{2}(\pi^*(e^z)e^{-z})^2$ and $\at(z) := \sigma \pi^*(e^z)e^{-z}$.
	Conversely, if $\{Z^*_t\}$ is a strong solution of \cref{SDE:optimal_Z}, then $\{X^*_t= \ee^{Z^*_t}\}_{t\ge0}$ is a strong positive solution of \cref{SDE:optimal}. Therefore, we only need to show that \cref{SDE:optimal_Z} has a (non-explosive) unique strong solution. In view that $\psi(\cdot)$ is bounded by Theorem \ref{thm:DualFBP-moderate}.(i), it follows from \cref{pi_optimal_moderate} that $\pis(x)/x$ is uniformly bounded for $x>0$. Therefore, $\pi^*(\ee^z)\ee^{-z}$ is uniformly bounded for $z\in\Rb$. Since $\varphi(0)>0$, it follows from \cref{c_optimal_moderate} that $\cs(\ee^z)$ is uniformly bounded for $z\in\Rb$. It then follows that $\bt(z)$ is uniformly bounded for $z\in\Rb$. Furthermore, by \cref{pi_optimal_moderate}, $\at:\Rb\to(0,+\infty)$ is Lipschitz continuous and bounded away from zero on compact subsets of $\Rb$. It then follows from Proposition 5.5.17 of \cite{KS1991} that \cref{SDE:optimal_Z} has a non-exploding unique strong solution. Therefore, \cref{SDE:optimal} has a positive unique strong solution (assuming that $\varphi(0)>0$).
	
	\item \textbf{Case 2:} Assume that $\psi(0)=1$. Note that in this scenario, we allow $\varphi(0)=0$. Therefore, we cannot guarantee that the function $\cs(\ee^z)$ (and thus $\bt(z)$) are bounded. Therefore, the argument in Case 1 is not applicable.\footnote{Conversely, the argument presented in Case 2 is not applicable to Case 1, since we cannot assume that \cref{eq:Psi0} holds for a constant $\Psi_0>0$ in Case 1. In other words, we must consider Case 1 and Case 2 separately.} 
	
	Let $b(\xi):=(r+\rho)\xi+\mu\pi^*(\xi)-(1+\rho \xi)c^*(\xi)$ and $a(\xi):=\sigma \pi^*(\xi)$ be the drift and diffusion functions of the stochastic differential equation \cref{SDE:optimal}.
	By \cref{c_optimal_moderate}, \cref{pi_optimal_moderate}, and properties of $\varphi(y)$ and $\psi(y)$ given by Theorem \ref{thm:DualFBP-moderate}.(i), we have that $b(\xi)$ is locally bounded and $a(\xi)$ is locally Lipschitz continuous and bounded away from 0 on compact subsets of $(0,+\infty)$. Theorem 5.5.15 of \cite{KS1991} then yields that \cref{SDE:optimal} has a (possibly exploding) unique weak solution. To show that \cref{SDE:optimal} has a positive unique strong solution, it suffices to show that any weak solution \cref{SDE:optimal} is non-exploding in the interval $(0,+\infty)$ (i.e. that it does not exit the interval $(0,+\infty)$ in finite time). To this end, define
	\begin{align}\label{eq:Feller_fn}
		f(\xi) := \int_{x_0}^\xi\int_{x_0}^y \frac{2}{a(z)^2}\exp\left(-2\int_z^y \frac{b(s)}{a(s)^2}\dd s\right) \dd z\dd y,
		\quad \xi>0.
	\end{align} 
	By Feller's test for explosions (see, for instance, Theorem 5.5.29 of \cite{KS1991}), a weak solution of \cref{SDE:optimal} is non-exploding in $(0,+\infty)$ if and only if
	\begin{align}\label{eq:Feller_test}
		\lim_{\xi\to0^+} f(\xi) = \lim_{\xi\to+\infty} f(\xi)=+\infty.
	\end{align}
	
	For $\xi\in(0,x_0)$, \cref{c_optimal_moderate} and \cref{pi_optimal_moderate} yield 
	$a(\xi) = e_0 \xi$ and $b(\xi)=b_0 \xi$ in which $e_0:=\frac{\mu(1-\lambda)}{\sigma}$ and $b_0:=r+\rho+\frac{\mu^2(1-\lambda)}{\sigma^2}$.
	By inserting these into \cref{eq:Feller_fn} and some algebra, we obtain that
	\begin{align}\label{eq:Feller_fn_0}
		f(\xi) =
		\begin{cases}
			\frac{2e_0^2}{(2b_0-e_0^2)^2}\left[
			\frac{2b_0-e_0^2}{e_0^2}\log\left(\frac{\xi}{x_0}\right)
			+\left(\frac{\xi}{x_0}\right)^{-\frac{2b_0-e_0^2}{e_0^2}} - 1
			\right],
			&\quad \text{if } 2b_0\ne e_0^2,\\[2em]
			\frac{1}{e_0^2}\left(\log\left(\frac{\xi}{x_0}\right)\right)^2,
			&\quad \text{if } 2b_0= e_0^2,
		\end{cases}
	\end{align}
	for $\xi\in(0,x_0)$. 
	It then follows that $\lim_{\xi\to0^+} f(\xi)=+\infty$ (note that, in the first expression of \cref{eq:Feller_fn_0}, we need to consider cases $2b_0>e_0^2$ and $2b_0<e_0^2$ separately). 
	
	It only remains to show that $\lim_{\xi\to+\infty} f(\xi)=+\infty$. Note that, since $\psi(0)=1$ (by the standing assumption of Case 2) and $\psi(y)$ is continuous for $y\in[0,y_0]$, there must exist a constant $\Psi_0\in(0,1)$ such that
	\begin{align}\label{eq:Psi0}
		0<\Psi_0\le \psi(y)\le 1,\quad 0<y<y_0.
	\end{align}
	For $y>z>x_0$, by \cref{c_optimal_moderate} and \cref{pi_optimal_moderate}, we have
	\begin{align}\nonumber\textstyle
		\int_z^y \frac{b(s)}{a(s)^2}\dd s
		&=\textstyle
		\int_z^y \frac{(r+\rho)s+\mu\pi^*(s)-(1+\rho s)c^*(s)}{\sig^2\pis(s)^2}\dd s
		\le \int_z^y \frac{(r+\rho)s+\mu\pi^*(s)}{\sig^2\pis(s)^2}\dd s\\
		\nonumber
		&= \textstyle
		\int_z^y \left(
		\frac{\sig^2\rho^2(r+\rho)s}{\mu^2(1+\rho s)^2\psi(v'(s))^2}
		+ \frac{\rho}{(1+\rho s)\psi(v'(s))}
		\right)\dd s\\
		\nonumber
		&\le \textstyle
		\frac{\rho}{\Psi_0}\int_z^y \frac{1+(\rho+e_1) s}{(1+\rho s)^2}\dd s
		=\frac{1}{\rho\Psi_0}\left[
		\frac{\rho e_1(z-y)}{(1+\rho y)(1+\rho z)}
		+ (\rho+e_1)\log\left(\frac{1+\rho y}{1+\rho z}\right)
		\right]\\
		\label{eq:Feller_fn_aux1}
		&\le\textstyle \frac{\rho+e_1}{\rho\Psi_0} \log\left(\frac{1+\rho y}{1+\rho z}\right).
	\end{align}
	To obtain the second inequity, we have used \cref{eq:Psi0} and defined $e_1:=\sig^2\rho(r+\rho)/(\mu^2\Psi_0)$. Define also
	\begin{align*}
		b_1 := \frac{\rho+e_1}{\rho\Psi_0} = \frac{(r+\rho)\sig^2}{\mu^2\Psi_0^2}+\frac{1}{\Psi_0}>1,
	\end{align*}
	in which the inequality holds since $\Psi_0\in(0,1)$ in \cref{eq:Psi0}.
	For $\xi>x_0$, \cref{eq:Feller_fn}, \cref{eq:Feller_fn_aux1}, and \cref{pi_optimal_moderate} yield
	\begin{align*}
		f(\xi) 
		&= \int_{x_0}^\xi\int_{x_0}^y \frac{2}{a(z)^2}\exp\left(-2\int_z^y \frac{b(s)}{a(s)^2}\dd s\right) \dd z\dd y\\
		&\ge 
		\int_{x_0}^\xi\int_{x_0}^y \frac{2}{\sigma^2 \pi^*(z)^2}\exp\left(-2b_1 \log\left(\frac{1+\rho y}{1+\rho z}\right)\right) \dd z\dd y\\
		&=
		\int_{x_0}^\xi\int_{x_0}^y 
		\frac{2\sigma^2\rho^2}{\mu^2(1+\rho z)^2\psi\big(v'(z)\big)^2}\left(\frac{1+\rho y}{1+\rho z}\right)^{-2b_1}\dd z\dd y\\
		&\ge
		\frac{2\sig^2\rho^2}{\mu^2}
		\int_{x_0}^\xi (1+\rho y)^{-2b_1} 
		\int_{x_0}^y (1+\rho z)^{2b_1-2}\dd z\dd y\\
		&=
		\frac{2\sig^2\rho}{\mu^2(2b_1-1)}
		\left[\log\left(\frac{1+\rho \xi}{1+\rho x_0}\right) 
		+ \frac{1}{2b_1-1}\left(
		\left(\frac{1+\rho x_0}{1+\rho \xi}\right)^{2b_1-1}
		- 1
		\right)
		\right],
	\end{align*}
	in which the second inequality holds because of \cref{eq:Psi0}. By letting $\xi\to+\infty$, we then conclude that $\lim_{\xi\to+\infty} f(\xi)=+\infty$. We have shown that \cref{eq:Feller_test} holds which, as we have argued, is equivalent to \cref{SDE:optimal} having a positive unique strong solution in Case 2.
\end{itemize}
We have shown that \cref{SDE:optimal} has a unique strong solution $\{X^*_t\}_{t\ge0}$. Therefore, the relative investment and consumption policy  $\big\{\big(\pi^*(X^*_t),c^*(X^*_t)\big)\big\}_{t\geq 0}$ is admissible (c.f. Definition \ref{def:rel-admis}). Henceforth, with a slight abuse of notations, we define
\begin{align*}
	c^*_t := \cs(X^*_t),\quad \text{ and } \quad \pi^*_t:=\pis(X^*_t),\quad t\ge0.
\end{align*}

Next, we show that the policy $\{(c^*_t,\pi^*_t)\}_{t\ge0}$ is optimal for the concavified problem \cref{eq:Vt}. We do so in two steps, by first showing that
\begin{align}\label{eq:Verfity_STEP1}
	v(x) = \Eb\left[\int_0^{+\infty}\ee^{-\del t} \Ut\left(c^*_t\right)\dd t\right] \le \Vt(x),
\end{align}
and then proving
\begin{align}\label{eq:Verfity_STEP2}
	v(x) \ge \Vt(x).
\end{align}

In Corollary \ref{coro:solveHJB} it was shown that $v(\cdot)$ and $x_0$ satisfy \cref{eq:HJB-FBP1-moderate}, \cref{eq:Ansatz2-moderate} and $v^{\prime\prime}(\cdot)<0$. It is also straightforward to verify that $\cs(\cdot)$ of \cref{c_optimal_moderate} satisfies \cref{eq:cs-moderate}, and that $\pis(\cdot)$ of \cref{pi_optimal_moderate} satisfies \cref{eq:piStar}. By the argument presented in Subsection \ref{subsec:candidatepolicies}, it then follows that, for all $\pi\in\Rb$, $c\ge0$, and $\xi>0$,
\begin{align}\nonumber
	&\textstyle
	-\del v(\xi) + \big((r+\rho) \xi + \mu \pi - (1+\rho \xi) c\big) v'(\xi) + \frac{1}{2}\sig^2\pi^2 v^{\prime\prime}(\xi) + \Ut(c)\\*
	\nonumber
	&\textstyle
	\le -\del v(\xi) + \big((r+\rho) \xi + \mu \pi^*(\xi) - (1+\rho \xi) \cs(\xi)\big) v'(\xi) + \frac{1}{2}\sig^2\pis(\xi)^2 v^{\prime\prime}(\xi) + \Ut\big(\cs(\xi)\big)\\
	\label{eq:HJB-concavified_alt}
	&=0.
\end{align}

To show \cref{eq:Verfity_STEP1}, denote $Z^*_t := \log X^*_t$ which satisfies \cref{SDE:optimal_Z}, and let $\tau^*_n := \inf\{t\geq 0: X^*_t \geq n {\rm\ or\ }X^*_t\leq 1/n \}$. Because $\{X^*_t\}_{t\geq 0}$ is non-exploding, we have $\lim_{n\to \infty}\tau_n^*=\infty$ almost surely. Moreover, on $[0,\tau_n^*]$, $X^*_t$ is bounded away from 0 and $\infty$. Therefore, $\int_0^{\tau_n^*} e^{-\delta s}|\pi^*(X^*_s)v'(X^*_s)|^2\dd s$ is almost surely bounded (with a bound possibly depending on $n$). By It\^o's formula and for an arbitrary constant $T>0$, we have 
\begin{align*}
	\ee^{-\delta (T\wedge \tau^*_n)}v(X^*_{T\wedge \tau^*_n})-v(x) = -\int_0^{T\wedge \tau_n^*}\ee^{-\delta s}\Ut\big(c^*(X^*_s)\big)\dd s + \int_0^{T\wedge \tau^*_n}\sigma v'(X^*_s)\pi^*(X^*_s) \dd B_s,
\end{align*}
where we have used the equality in \cref{eq:HJB-concavified_alt}. Taking expectations on both sides yields 
\begin{equation*}
	v(x) = \Eb \left[\ee^{-\delta (T\wedge \tau^*_n)}v(X_{T\wedge \tau^*_n}) \right]+\Eb\left[ \int_0^{T\wedge \tau^*_n}\ee^{-\delta s}\Ut(c^*(X^*_s))\dd s\right].
\end{equation*}
By first sending $n\to \infty$, then letting $T\to \infty$, we obtain from the transversality condition \ref{eq:trans} and the monotone convergence theorem that
\begin{align}\label{eq:v_Vt_optimal}
	v(x)=\E\left[\int_0^{\infty}\ee^{-\delta s}\Ut\big(c^*(X^*_s)\big) \dd s\right],
\end{align}
which is the equality in \cref{eq:Verfity_STEP1}. The inequality in \cref{eq:Verfity_STEP1} trivially follows from the definition of $\Vt(x)$ in \cref{eq:Vt}, as we have already established that $\{(\pi^*_t,c^*_t)\}_{t\geq 0}\in \Acr(x)$.

To prove \cref{eq:Verfity_STEP2}, take an arbitrary admissible relative policy $\{(\pi_t,c_t)\}_{t\geq 0}\in \Acr(x)$. Let $\{X_t\}_{t\geq 0}$ (i.e. the corresponding relative wealth process) be the strong solution of \cref{eq:X}. By replacing $\{(\pi^*_t,c^*_t)\}_{t\geq 0}$ with $\{(\pi_t,c_t)\}_{t\geq 0}$ in the arguments that yielded \cref{eq:v_Vt_optimal}, we obtain the following inequality
\begin{align}\label{eq:v_Vt_suboptimal}
	v(x)\ge \E\left[\int_0^{\infty}\ee^{-\delta s}\Ut(c_t) \dd s\right].
\end{align}
In particular, we can only use the inequality in \cref{eq:HJB-concavified_alt} (instead of the equality therein), which leads to an inequality in \cref{eq:v_Vt_suboptimal}. By taking the supremum among all  $\{(\pi_t,c_t)\}_{t\geq 0}\in \Acr(x)$ on the right side of \cref{eq:v_Vt_suboptimal} and using the definition of $\Vt(x)$ in \cref{eq:Vt}, we obtain \cref{eq:Verfity_STEP2}.

From \cref{eq:Verfity_STEP1} and \cref{eq:Verfity_STEP2}, and thanks to the previous result that $\{(\pi^*_t,c^*_t)\}_{t\geq 0}\in \Acr(x)$, we deduce that $v(x)=\Vt(x)$ and that $\{(\pi^*_t,c^*_t)\}_{t\geq 0}$ is an optimal policy for the concavified problem \cref{eq:Vt}.

To complete the proof of Theorem \ref{thm:verification_moderate}, it remains to show that $V(x)=\Vt(x)$ and that the policy $\{(\pi^*_t,c^*_t)\}_{t\geq 0}$ is also optimal in the original stochastic control problem \cref{eq:VF}. To this end, we first claim that
\begin{align}\label{eq:Ucs_Utcs}
	U\big(\cs(\xi)\big) = \Ut\big(\cs(\xi)\big),\quad \xi>0.
\end{align}
Indeed, if $0<\xi<x_0$, then $U\big(\cs(\xi)\big)=U(0)=\Ut(0)=\Ut\big(\cs(\xi)\big)$ by \cref{c_optimal_moderate} and \cref{eq:Ut-moderate}. Similarly, if $\xi\ge x_0$, then $\cs(\xi)=\al+(U_+')^{-1}\big(\varphi(v'(\xi))\big)\ge \al+(U_+')^{-1}\big(\varphi(v'(x_0))\big)=c_0$, and therefore $U\big(\cs(\xi)\big)=\Ut\big(\cs(\xi)\big)$ by \cref{eq:Ut-moderate}.

Finally, by virtue of \cref{eq:Ucs_Utcs}, it holds that 
\begin{align*}
	\Vt(x)
	=\Eb\left[\int_0^{+\infty}\ee^{-\del t} \Ut\big(\cs(X^*_t)\big)\dd t\right]
	=\Eb\left[\int_0^{+\infty}\ee^{-\del t} U\big(\cs(X^*_t)\big)\dd t\right]\leq V(x)\leq \Vt(x),
\end{align*}
where the second to last step follows from the definition of $V(x)$ in \cref{eq:VF}, and the last step holds in view of $U(c)\le\Ut(c)$ for $c>0$ by the definition of $\Ut(c)$ in \cref{eq:Ut_defined}. Thus, we conclude $V(x)=\Vt(x)$, which completes the proof.


\section*{Acknowledgments}
We are grateful to the Associate Editor and the anonymous referee for their helpful comments. 
B. Angoshtari acknowledges a Simons Foundation Travel Support for Mathematicians, ID 946036. X. Yu is supported by the Hong Kong RGC General Research Fund (GRF) under grant no. 15306523 and by the Research Center for Quantitative Finance at the Hong Kong Polytechnic University under grant no. P0042708. F. Yuan is supported by the start-up fund at the Chinese University of Hong Kong (Shenzhen) with grant no. UDF01004253.
%
%
\printbibliography 

\end{document}